\documentclass[11pt]{article} 

\usepackage{amsmath,amsthm,latexsym,amssymb,amsfonts,epsfig}

\newtheorem{Lemma}{Lemma}[section]

\newtheorem{Proposition}{Proposition}[section]

\newcommand{\be}{\begin{equation}}
\newcommand{\ee}{\end{equation}}
\newcommand{\ba}{\begin{eqnarray}}
\newcommand{\ea}{\end{eqnarray}}

\title{{\sf Asymptotically safe -- canonical quantum gravity junction}}
\author{
{\sf T. Thiemann}$^1$\thanks{{\sf 
thomas.thiemann@gravity.fau.de}}\\
\\
{\sf $^1$ Inst. for Quantum Gravity, FAU Erlangen -- N\"urnberg,}\\
{\sf Staudtstr. 7, 91058 Erlangen, Germany}\\
}
\date{{\small\sf \today}}

\makeatletter
\@addtoreset{equation}{section}
\makeatother

\begin{document} 

\maketitle

{\sf

\begin{abstract}
The canonical (CQG) and asymptotically safe (ASQG)
approach to quantum gravity share 
to be both non-perturbative programmes. However, apart from that they seem
to differ in several aspects such as: 1. Signature: CQG is Lorentzian while 
ASQG is mostly Euclidian. 2. Background Independence (BI): CQG is 
manifesly BI while ASQG is apparently not. 3. Truncations: CQG is 
apparently free of truncations while ASQG makes heavy use of them.

The purpose of the present work is to either overcome actual differences or 
to explain why apparent differences are actually absent. Thereby we intend 
to enhance the contact and communication between the two communities.
The focus of this contribution is on conceptual issues rather than
deep technical details such has  
high order truncations. On the other hand the paper tries to be self-contained 
in order to be useful to researchers from both communities.

The point of contact is the path integral formulation of Lorentzian CQG in its 
reduced phase space formulation which yields the formal generating functional 
of physical (i.e. gauge invariant) 
either Schwinger or Feynman N-point functions for (relational) 
observables. The corresponding effective actions of these 
generating functionals can then be subjected to the ASQG  
Wetterich type flow equations which serve in particular to find the 
rigorous generating fuctionals via the inverse Legendre transform of 
the fixed pointed effective action.
\end{abstract}

\section{Introduction}
\label{s1}

It is widely believed that only a non-perturbative approach can yield a 
predictive theory of quantum gravity (QG). The asymptotic safety approach 
(ASQG) \cite{1} and the canonical approach (CQG) \cite{2}, the most 
modern incarnation of which is Loop Quantum Gravity (LQG) \cite{3}, 
are two such programmes. These have been developed quite independently 
of each other and it is therefore an interesting challenge to try to 
bring them into contact in order to investigate their differences and their 
hopefully complementary strengths and weaknesses.

At first sight, this seems to be an easy task: Both theories start from 
a Lagrangian which can then be subjected to either canonical (CQG) 
or path integral (ASQG)
quantisation. The connection between those types of quantisations is that 
the path 
integral formulation allows to compute correlation functions of field
operator valued distributions of the canonical formulation. This 
relation is quite well understood, both conceptually and 
technically, for quantum field theories (QFT) 
in Minkowski space without gauge symmetries. On more general backgrounds 
with or without gauge symmetries a rigorous implementation of this relation 
is more challenging but at a formal level one can can still derive one 
formulation from the other via the reduced phase space path integral \cite{4}.

However, apart from these mostly technical challenges, there appear to be 
unsurmountable differences between the two programmes:\\
1. Signature:\\
Apart from few exceptions \cite{5}, ASQG is a path integral formulation 
for Euclidian signature metrics while CQG in its LQG incarnation 
is explicitly for Lorentzian signature \cite{6}. One cannot simply pass
between the signatures using analytic continuation in time as a quantum 
metric that one integrates over has quite arbitrary (generically non-analytic) 
time dependence.\\
2. Background dependence (BI):\\
ASQG makes heavy use of background metric dependent constructions such as 
the very definition of the averaging kernel that enters the Wetterich 
equation. Sometimes this background is kept arbitrary \cite{7}, sometimes it is 
chosen to be maximally symmetric \cite{8} and it is at least confusing to 
the newcomer in what sense the predictions of ASQG can be interpreted 
background independently. CQG on the other hand is manifestly background 
independent by construction.\\
3. Truncations:\\
The Wetterich equation of ASQG is an exact and non-perturbative 
renormalisation flow equation for the so-called effective average action 
(EAA), however, to actually solve this equation one must restrict that 
action to a finite number of terms and truncate the Wetterich equations 
to those terms. Such kind of truncations are absent in LQG.\\
\\
The purpose of this paper is two-fold:\\ 
A.\\
On the one hand we want to understand to what extent one can  
derive ASQG from CQG using the reduced phase space path integral. 
This should then ease the comparison between the two formulations.\\
B.\\
On the other hand, such a derivation will shed light in what sense the above 
three differences between the two programmes can be overcome.\\
\\
The architecture of this contribution is as follows:\\
\\
To be useful for both communities, in section \ref{s2} we will review 
the canonical quantisation programme as well as the reduced phase space path 
integral and in section \ref{s3} we will review the main ingredients of 
the asymptotic safety programme. Special attention will be paid 
on how to implement manifest gauge invariance in both programmes and how 
these can be translated into each other. This ensures that the cortrrelation
functions that one eventually computes are for (relational) Dirac observables
only.

While the review material laid out in sections \ref{s2}, \ref{s3} 
is quite general, in 
section \ref{s4} we pick a concrete class of theories within the 
Einstein-Scalar sector. While the matter content of that theory is not 
very physical, it serves the purpose to showcase for a model that 
is technically not too involved, 
that a connection between the two programmes can indeed be established.
We develop and prepare 
that model up to the point at which an ASQG tretament is
immediately possible. The ASQG tretament itself is carried out in
the companion paper \cite{companion}.

In section \ref{s5} we summarise and conclude. 

In the appendices we collect 
some background material:\\
In appendix \ref{sb} we establish the non-trivial relation between 
effective actions whose preimages under Legendre transformation 
are related by restriction. In appendix \ref{sa} we show how 
to get rid of square root actions in the path integral by a method which 
is inspired by Polyakov's trick to switch from the Nambu-Goto to the 
Polyakov action for the closed bosonic string \cite{10}.
In appendix
\ref{sc} we report a side result of our analysis, namely that 
the so-called optimised cut-off, which are much used in the Euclidian version 
of ASQG, is not in the image of the Laplace transform
even of a positive measure and thus raises the question whether 
it is in the image of the Laplace transform of some well defined 
mathematical object at all. Likewise some smoother cut-offs used in Euclidian 
ASQG are not in the Laplace transform image of 
functions but rather distributions so that the heat kernel time 
integrals that are required to compute the flow are ill-defined as they stand.
This
kind of analysis leads to 
appendix \ref{sd} where we introduce and explore a new class of
cut-off functions both for Euclidian and Lorentzian signature 
which safely are in the image of the Laplace and Fourier transform
respectively of a Schwartz function. In the Lorentzian regime 
the flow defined by this cut-off is necessarily in {\it complex coupling
constant space} but among the possible trajectories running into the 
fixed point we only allow 
{\it admissible oens} 
which have real valued, finite $k=0$ limit for the physical dimensionful 
couplings.\\ 
\\
As far as the obstructions I.-III. are concerned, we learn 
the following: 
\begin{itemize}
\item[I.] {\it Signature}\\
One can formulate the reduced phase space path integral either for 
Schwinger N-point functions or for Feynman time ordered N-point functions.
This is possible because in the reduced phase space formulation one has 
access to a Hamiltonian operator $H$ and Wick rotation of $e^{it H}$ 
with respect to $t$ has a well-defined meaning. However, the physical 
Hamiltonian $H$ is not of standard Schr\"odinger type, it typically 
involves a 
square root, so that integrating out 
the momenta is not immediately possible in exact form. The square root 
Hamiltonian is genericfor generally covariant theories and arises already 
for the relativistic particle \cite{9}. By the method provided in appendix 
\ref{sa} one can get rid of the square root and in that case, one can indeed 
integrate out the momenta and implement Euclidian signature ASQG. 
The resulting theory is however spatially very non-local. One can 
implement a more local version. Another method to get rid of the 
square root and to avoid non locality 
is to unfold the reduced phase space path integral to 
the unreduced one but then one must choose the Lorentzian 
route. One is then able to integrate out the momenta 
exactly and 
relate the path integral to the Einstein-Hilbert action for Lorentzian 
signature.
At this point one may think that comparing the programmes is no longer 
possible. However, as noted e.g. in \cite{5}, it is also possible 
to give a Lorentzian version of the Wetterich equation for the EAA although 
its interpretation is less clear.
\item[II.] {\it Background Independence}\\
The derivation of the reduced phase space path integral results in 
the Einstein-Hilbert action with metric dependent cosmological constant 
term and a non-standard gauge fixing structure that arises both from 
integrating out the meomenta and the gauge fixing conditions that we imposed.
If the corresponding generating functional 
was well-defined one could compute  all (connected) 
Feynman N-point functions from it. However, it is not and thus one may use the 
proposal of the ASQG programme to modify the action by quadratic
term that involves a so-called averaging
kernel which depends on a parameter $k\in \mathbb{R}$ and a background 
metric $\bar{g}$ of Lorentzian signature and $k^2$ can locally be 
related to the norm of a vector $k$ with respect to $\bar{g}$. The avearging 
kernel vanishes when $k=0$.

One may now pass from the modified generating functional to the EAA by 
Legendre transform which obeys a Wetterich like exact equation. From
the viewpoint of CQG,
the Wetterich equation serves the sole purpose to 
select a fixed point EAA (solution of the Wetterich in terms 
of the dimensionless couplings with finite limit as 
$k\to \infty$ fixing all but a finite 
number of those) 
which also has a limit as $k\to 0$ with respect to the running 
dimensionful couplings which then serves 
as the rigorous definition of the ordinary effective action of the 
guage invariant theory 
that we are interested in. Note that in the usual 
ASQG flow equation one considers the 
flow with respect to all degrees of freedom  
while actual effective action is restricted the gauge inavariant degrees of 
freedom. This restriction requires an additional step after the solution
of the Wettrich equation has been obtained.

Now the reason why the end result of this analysis is background independent 
is due to the use of the background method \cite{12} that one employs:
We keep the background arbitrary and leave it unspecified aside from assuming
that it has Lorentzian signature. This not only allows to unambiguously 
monitor the flow of the couplings that enter the EAA, it also leads to 
an effective action as follows: The Wetterich equation yields a functional
$\bar{\Gamma}_k(\hat{g},\bar{g})$ where $\hat{g}$ is the field that 
arises by Legendre transforming the generating functional 
$\bar{C}_k(f,\bar{g})$ of connected functions with respect to the current
$f$ using the background field method. 
If $\bar{\Gamma}_0(f,\bar{g})$ exists then 
$\Gamma(g):=\Gamma_0(0,g)$ is the standard 
effective action and its Legendre transform defines $C(g)$. Thus, 
the background method applied in this way is in fact background independent!
This is of course well known in field theory and in fact it was DeWitt,
one of the founding fathers of CQG, 
who co-invented this method.
\item{III.} {\it Truncations}\\     
In order to solve the Wetterich equation, we make an Ansatz 
for $\bar{\Gamma}_k$ that involves a small set of $k$ dependent coupling 
constants.
The left hand side of the Wetterich equation only depends on the derivatives 
of those coupling constants with respect to $k$. To match these terms to 
those that arise on the right hand side of the Wetterich equation one 
can make use of the Lorentzian version of the heat kernel method \cite{13}
which is based on DeWitt's seminal work on the Synge world function \cite{14}
and which plays a fundamental role in microlocal analysis \cite{15} and 
QFT in curved spacetime (CST) \cite{13}. We may call this the 
``Schr\"odinger kernel''
method. The properties used for the Synge world function are signature 
independent and can be used for our purposes as well. While in ASQG one 
often matches the gauge fixing in order to avoid certain ``off-diagonal'' 
terms (basically polynomials of covariant background derivatives that are 
not polynomials of the background d'Alambertian) the ``heat'' kernel method     
can be extended to off-diagonal situations \cite{17} and one does not even 
need to perform a transverse-traceless decomposition by suitable 
organisation of the truncation scheme (e.g. by orders of 
derivatives rather than 
orders of curvature polynomials). One is then left with performing integrals 
with respect to the ``heat'' kernel parameter which in Lorentzian signature 
are more difficult to perform than in Euclidian signature.
 
After all of this has been done, one must discard on the right hand side 
all terms that have no counter part in the Ansatz (truncation), then 
perform the fixed point analysis and finally take $k\to 0$ (if possible 
within the limits of the chosen truncation). In 
principle one can make the system of equations close iteratively by 
refining the Ansatz step by step by the terms that are produced by the 
Wetterich equation not already contained in the Ansatz. 

One may argue 
that in CQG this truncation is not necessary. However, the physical 
Hamiltonian in LQG suffers from many 
quantisation ambiguities which must be fixed using renormalisation 
methods \cite{18}, e.g. Hamiltonian renormalisation 
\cite{19}. In any renormalisation scheme one usually is forced 
to perform some form of approximation (here: truncations) 
due to practical limitations (in case of ASQG
it is the limits in complexity that symbolic computation programmes are able 
to handle in reasonable time). Thus renormalisation is present in both 
approaches and truncations are in principle not necessary except for reasons 
of pracicability. What is missing, to the best of our knowledge in every 
approach, is an error control analysis, i.e. 
a mathematical proof that a certain truncation does not miss out a relevant
coupling.             
\end{itemize}

\section{Relational observables, reduced phase space, canonical 
quantisation, path integral formulation}
\label{s2}

In this section we review elements of canonical quantisation of constrained
systems. CQG practitioners can safely skip this section.

\subsection{Classical phase space of constrained systems}
\label{s2.0}

Consider a phase space coordinatised by two sets of canonical pairs 
$z=((p_a,q^a),(y_j,x^j))$. We define the coordinate functions 
$Z=(P_a,Q^a,Y_j,X^j)$ on the phase space by $P_a(z)=p_a$ etc. The 
non vanishing Poisson brackets are 
\be \label{2.1}
\{P_a,Q^b\}=\delta_a^b,\; 
\{Y_j,X^k\}=\delta_j^k
\ee
The labels $a,b,c,..=1,2,..$ and $j,k,l,..-1,2,..$ 
can be taken from a discrete index set.
In classical mechanics these sets are finite, in classical field theory 
they are countably infinite. In the latter case we think of them as obtained as 
$P_a=<e_a,P>,\; Q^a=<f^a,Q>, Y_j=<b_j,Y>, X^j=<c^j,X>$ where $<.,.>$ denotes 
the inner product of the appropriate 1-particle Hilbert space including 
summation over gauge group representation labels 
(tensor representations of the diffeomorphism group, irreducible 
representations of Yang-Mills type gauge groups etc.) which we suppress, 
$P,Q,Y,X$ are the continuum fields and $(e_a,f^a),\; (b_j,c^j)$ are pairs 
of bases and co-bases respectively of the one particle Hilbert spaces, i.e. 
$<e_a,f^b>=\delta_a^b,\;<b_j,c^k>=\delta_j^k$.

These labels are adapted to a system of first class 
constraints  
\be \label{2.2}
\{C_j,C_k\}=f_{jk}\;^l\;\; C_l
\ee
where the structure functions $f_{jk}\;^l$ need not be constant on the phase 
space which is indeed the case for canonical GR. Again in field theory 
the constraints $C_j=<e_j,C>$ are smeared versions of the continuum 
constraints $C$. There may also be a true 
Hamiltonian $H$ already (not the case for generally covariant systems such as 
GR) which is gauge invariant i.e.
\be \label{2.3}
\{C_j,H\}=f_{j0}\;^k\;\; C_k
\ee
in which case we expand the phase space by an additional canonical pair 
$(y_0,x^0)$ which has vanishing Poisson brackets with all others and we add 
an additional constraint
\be \label{2.4}
C_0=y_0+H\;\; \Rightarrow\;\;\{C_j,C_0\}=f_{j0}\;^k\;\; C_k 
\ee
The augmented set of constraints is thus again first class so that w.l.g.
we may restrict to the case without true Hamiltonian.

The construction of the reduced phase space can now be described in two 
equivalent ways, either using so called relational Dirac observables 
which are gauge invariant or by using gauge fixing. Both descriptions 
take the following steps: First we solve the constraints $C_j=0$ 
for the momenta $y_j$ to arrive at equivalent sets of constraints
\be \label{2.5}
\hat{C}_j=y_j+h_j(p,q;\;x)
\ee
In general this involves picking a square root because typically the 
constraints depend quadratically on the momenta (recall the case of the 
relativistic particle). Also the dependence of the functions $h_j$ on 
$x$ is in general non-trivial, if it is trivial one says that the system 
is deparametrised. This happens only in very fortunate cases. In field theory
to write the constraints as (\ref{2.5}) often is not possible algebraically 
but one must also solve PDE's or equivalently invert non-diagonal 
infinite dimensional matrices which is technically challenging. 
One will therefore select the 
split of the canonical coordinates into the two sets 
in such a way as to avoid these complications as much as possible.

The virtue of the equivalent set of constraints (\ref{2.5}) is that they 
are Abelian, i.e. they have vanishing Poisson brackets among each 
other. To see this we note that $\{\hat{C}_j,\hat{C}_j\}$ must 
vanish when the constraints hold since $\hat{C}_j=M_j\;^k\; C_k$ for 
some matrix $M$. On the other hand, $\{\hat{C}_j,\hat{C}_j\}$ does not 
depend on $y_j$, it may therefore evaluated at any $y_j$ and in 
particular at $y_j=-h_j$. Thus $\{\hat{C}_j,\hat{C}_j\}$ vanishes identically, 
not only when the constraints hold. 

\subsection{Relational Dirac observables}
\label{s2.1}

This observation is the mechanism behind
relational observables \cite{20}. Consider ``coordinate functions'' 
$k^j:\; \mathbb{R}\to \mathbb{R};\;\; t\mapsto c^j(t)$ which are constants 
on the phase space and a one parameter set of gauge fixing conditions 
\be \label{2.6}
G^j(t):=X^j-k^j(t)   
\ee
Then for any function $F$ on the phase space the quantity 
\be \label{2.7}
O_F(t):=[\exp(s^j\; V_j\cdot)\; F]_{s=-G(t)}
\ee
is gauge invariant, i.e. it has vanishing Poisson brackets with the 
constraints $\hat{C}_j$. Here $V_j\cdot F:=\{\hat{C}_j,F\}$ is the 
Hamiltonian vector field of $\hat{C}_j$ and in (\ref{2.6}) one 
first computes the Hamitonian flow for $s$ a constant function on the phase 
space and after that evaluates at $s=-G(t)$. To verify that 
\be \label{2.8}
\{\hat{C}_j,O_F(t)\}=V_j\cdot O_F(t)=0
\ee
it is enough to use the Taylor expansion of (\ref{2.7}), that 
$V_j\cdot G^k=\delta_j^k$ and that $[V_j,V_k]=X_{\{\hat{C}_j,\hat{C}_k\}}=0$.

The observables $O_F(t)$ satisfy a number of interesting properties which 
are not difficult to prove \cite{4,20}. We mention two of them 
\be \label{2.9}
\{O_{F_1}(t),\; O_{F_2}(t)\}=O_{\{F_1,F_2\}^\ast}(t),\;
O_{F_1\circ F_2}(t)=F_1(O_{F_2}(t))
\ee 
where 
\be \label{2.10}
\{F_1,F_2\}^\ast=\{F_1,F_2\}
+\{F_1,\hat{C}_j\}\;\{G^j(t),F_2\}
-\{F_2,\hat{C}_j\}\;\{G^j(t),F_1\}
\ee
is the Dirac bracket defined by $\hat{C}_j, G^j(t)$ considered as a second 
class set of constraints.  The second relation in (\ref{2.9}) implies 
that 
\be \label{2.11}
[O_F(t)](z)=F([O_Z(t)](z))
\ee
i.e. to compute $O_F(t)$ we just need to compute the observables corresponding
to the coordinate functions $Z(z)$.
We easily see $O_{X^j}(t)=k^j(t)$ is just a constant 
function on the phase space and thus trivially has vanishing Poisson brackets 
with the constraints. Next, since $\hat{C}_j$ is already gauge invariant 
\be \label{2.12}
O_{\hat{C}_j}(t)=\hat{C}_j=O_{Y_j}(t)+h_j(O_P(t),O_Q(t);\;k(t))
\ee
Thus on the constraint surface $O_{Y_j}(t)$ can be expressed in terms 
of the $O_{P_a}(t), O_{Q^a}(t)$ which are therefore, at any time, 
the coordinates of the 
gauge invariant reduced phase space. Since $\{P_a,G^j(t)\}=
\{Q^a,G^j(t)\}=0$ from the first relation in (\ref{2.10})
\be \label{2.13}
\{O_{P_a}(t),O_{Q^b}(t)\}=O_{\{P_a,Q^b\}}(t)=\delta_a^b
\ee
thus the observables $O_{P_a}(t),O_{Q^a}(t)$ are canonically conjugate at 
all times $t$. This is important: If (\ref{2.13}) would be a complicated 
expression, it would be hard to find representations of the canonical 
commutation relations in the quantum theory.

Let us focus on a function $F$ which just depens on $p,q$. Then 
\be \label{2.14}
\frac{d}{dt}\; O_F(t)
=O_{\dot{k}^j(t)\{\hat{C}_j,F\}}
=O_{\dot{k}^j(t)\{h_j,F\}}
=O_{\dot{k}^j(t)\{h_j,F\}^\ast}
=\{O_{\dot{k}^j(t)\;h_j}(t),O_F(t)\}
\ee
where we used $\{F,G^j(t)\}=\{h_k,G^j(t)\}=0$ in the step before the last 
and then (\ref{2.9}). It follows that the gauge invariant, physical Hamiltonian 
that drives the one parameter time evolution of the those observables $O_F$ 
is given by 
\be \label{2.15}
H(t)=\dot{k}^j(t)\; h_j(O_P(0),O_Q(0);\;k(t)),\;
\{H(t),O_F\}([O_P(t),O_Q(t))=\frac{d}{dt} O_F(t)
\ee
It is explicitly time dependent unless $\dot{k}^j$ is a constant and 
if the dependence on $k(t)$ in $h_j$ vanishes (for example 
if the dependence on $x$ of $h_j$ is only through spatial derivatives).  
Of course a time independent physical Hamiltonian is much preferred 
(conservative system).

\subsection{Gauge fixing}
\label{s2.2}

The gauge fixing condition (\ref{2.6}) has already appeared in the relational
approach but there it served to extend functions $F(q,p)$ off the 
gauge cut $G^j=0$ of the constraint surface $C_j=0$ to the full phase space.
In the gauge fixing approach we consider the reduced phase space as 
coordinatised by the so called ``true degress of freedom''
$P,Q$ rather than $O_F(0),O_P(0)$. The ``gauge degrees of freedom''
$Y,X$ are discarded from this description altogether. This gives an 
equivalent description as follows: First we note that the canonical 
brackets among the $P,Q$ and the $O_P(0), O_Q(0)$ are identical. Next
consider a gauge transformation generated by the constrained 
Hamiltonian $\hat{C}(s)=s^j \hat{C}_j$. There 
is a residual 1-parameter set family of gauge transformations allowed   
that stabilise the gauge condition $G^j(t)=0$
\be \label{2.16}
\frac{d}{dt} G^j(t)=\partial_t G_j(t)+\{\hat{C}(s),G^j(t)\}
=-\dot{k}^j(t)+s^j=0
\ee
which fixes the Lagrange multiplier $s^j$. We note the fixed values 
\be \label{2.17}
s^j_\ast=\dot{k}^j(t)\;1,\;
X^j_\ast=k^j(t)\;1,\;
Y_j^\ast=-h_j(Q,P;\;k(t))
\ee
Let $F$ be a function of the true degrees of freedom $q,p$ only. Then 
the reduced Hamiltion $H_\ast(t)=H_\ast(q,p;t)$ is defined by 
\be \label{2.17a}
\{H_\ast(t),F\}(q,p):=(\{\hat{C}(s),F\}_{s=s_\ast,x=x_\ast,y=y^\ast})(q,p)
=\dot{k}^j(t)\{h_j(.,.;c(t)),F\}(q,p)
\ee
whence 
\be \label{2.18}
H_\ast(q,p;t)=\dot{k}^j(t)\; h_j(q,p;k(t))
\ee
Comparing (\ref{2.18}) and (\ref{2.18}) we see that the reduced Hamiltonian 
is the same function of the true degrees of freedom as is the physical 
Hamiltonian of the relational Dirac degrees of freedom. The two descriptions
are truly isomorphic. However one should keep in mind that the physical
interpretation of the true degrees of freedom and the reduced Hamiltonian 
depends on the choice of gauge fixing: $Q,P$ are those observables 
which on the gauge cut defined by $G^j(0)=0$ coincide with the gauge
invariant observables $O_P(0),O_Q(0)$ and their evolution is generated 
by the reduced Hamiltonian induced by $G^j(t)=0$.

\subsection{Reduced phase space quantisation: Hamiltonian formulation}
\label{s2.3}

Since the gauge invariant content of the theory can be described in terms 
of the true degrees of freedom $q,p$ and the reduced Hamiltonian $H_\ast(t)$
we consider canonical and path integral quantisation of this description.
We restrict to the case of a reduced Hamiltonian which is not explicitly 
time depndent. This will be sufficient for our application in section 
\ref{s4}. An extension to the explicitly time dependent case is possible 
but more complicated.\\
\\
Canonical quantisation involves finding a representation of the canonical
commutation relations and $^\ast-$relations (we set $\hbar=1$ and just 
consider bosonic variables for simplicity)
\be \label{2.19}
[P_a,Q^b]=i\;\delta_a^b\;1,\;\;
P_a^\ast=P^a, (Q^a)^\ast=Q^a
\ee
The algebra generated by the generators $P_a,Q^a,1$ is called the 
Heisenberg algebra and has the disadvantage that representations $\rho$
of $P_a,Q^a$ are unbounded operators. Thus one passes to the Weyl elements
\be \label{2.20}
W(f,g):=e^{i\;[f_a Q^a+g^a P_a]} 
\ee 
which obey the corresponding Weyl relations induced by (\ref{2.19}) 
\cite{20}. The Weyl elements generated the corresponding Weyl algebra
$\mathfrak{A}$ which is a $^\ast-$algebra (even a $C^\ast-$algebra).

We consider cyclic representations of $\mathfrak{A}$. These are in 
1-1 correspondence with states $\omega$ on $\mathfrak{A}$ i.e. positive,
normalised, linear functionals on $\mathfrak{A}$: 
$\omega(a^\ast a)\ge 0,\;\omega(1)=1,\; \omega(z a + z' a')=
z\omega(a)+z'\omega(a');\; z,z'\in \mathbb{C};\;a,a'\in \mathfrak{A}$.
Moreover $\omega$ gives rise to GNS data $({\cal H},\Omega,\rho)$ where 
${\cal H}$ is a Hilbert space, $\Omega$ a unit vector in $\cal H$ and 
$\rho$ a $^\ast-$representation of $\mathfrak{A}$ as bounded operators
on $\cal H$. The vector $\Omega$ is cyclic for $\rho(\mathfrak{A})$ i.e. 
${\cal D}:=\rho(\mathfrak{A})\Omega$ is a dense, common, invariant domain 
for the operators $\rho(a)$. Note the representation properties 
$\rho(a)\rho(b)=\rho(ab),\;\rho(a)+\rho(b)=\rho(a+b),\; 
\rho(a)^\dagger=\rho(a^\ast)$. The correspondence between $\omega$ and 
the GNS data is via the formula
\be \label{2.21}
\omega(a)=<\Omega,\rho(a)\Omega>_{{\cal H}}   
\ee
i.e. correlation functions among the operators $\rho(a)$ with respect 
to the cyclic vector $\Omega$ are computed by evaluating $\omega$ on them.
One shows that the GNS data are determined by $\omega$ up to unitary 
transformations between representations, see e.g. \cite{21} for all 
the details.

Cyclic representations of, or equivalently states on, $\mathfrak{A}$ are not 
difficult to find, any Fock representation is of that form. The real challenge
is to find states $\omega$ that support a quantum version 
of the reduced Hamiltonian $H_\ast$. We will drop the subscript $_\ast$ in 
what follows in order not to get confused with the $^\ast-$operation and 
thus just denote it by $H$. It is here where all the difficulties of 
interacting QFT enter and Haag's theorem \cite{22} tells us that such 
representations must be carefully adapted to $H$ if they exist at all.
On the other hand, this physical selection criterion on $\omega$ is 
welcome as otherwise we would be confronted with the question which 
unitary equivalence class of representation to choose (recall that 
spectra, i.e. measurement values, of operators depend on the unitary 
equivalence class). 

The construction of a quantum version of $H$ faces 
many issues such as: 1. operator ordering choices, 2. ultra-violet
(UV) singularities (summation over an infinite number of mode labels $a$) and 
3. infra-red (IR) singularities (working in non-compact spatial manifolds).
It is a common theme in constructive QFT \cite{23} to tame these issues 
by imposing a mode (UR) cut-off $a<M$, an IR cut-off (compactification 
of the spatial manifold e.g. on a torus of size $R$ with periodic 
boundary conditions) and to pick some truncation and ordering $H_M$ of 
$H$ with respect to the $P_a, Q^a,\; a\le M$ such that $H_M$ is a 
self-adjoint operator on ${\cal H}_M$ where $({\cal H}_M,\Omega_M,\rho_M)$ 
are the GNS data of a state $\omega_M$ on $\mathfrak{A}_M$ where 
$\mathfrak{A}_M$ is the Weyl algebra generated by the Weyl elements 
(\ref{2.20}) subject to the condition $f_a=0,\; a>M,\;g^a=0,\;a>M$. As
we are now in a situation with a finite number of degrees of freedom 
the problem of finding suitable $\omega_M$ is trivial as the von Neumann
theorem \cite{21} 
says that there is only one unitary equivalence class available 
if one asks that the representation is regular (i.e. we can differentiate 
the Weyl elements to obtain self-adjoint operators 
$\rho_M(P_a),\rho_M(Q^a),\;a\le M$ by Stone's theorem). The remaining 
task, at fixed $M$, is then to find a suitable operator $\rho_M(H_M)$. 
We assume that this task can be performed at any finite $M,R$. We will
abuse notation and write $H_M$ for $\rho_M(H_M)$ and $W(f,g)$ for 
$\rho_M(W(f,g))$ in what follows in order not to clutter the formulae.

We now have constructed a family of theories labelled by $M,R$. They
can be interpreted as a compactified version of the theory that we actually
want at finite resolution $M$. The idea of renormalisation is to impose 
that all those theories labelled by $M$ at fixed $R$ descend from a common
theory at the same fixed $R$ by coarse graining, i.e. by restricting 
observables to mode labels $a\le M$. In other words we want that 
\be \label{2.22}
\omega(a)=\omega_M(a)\;\forall\; a\in \mathfrak{A}_M,\;\;
P_M \; H\; P_M=H_M
\ee
where $P_M:\; {\cal H}\to {\cal H}_M$ is the orthogonal projection. 
These conditions are not automatically met using some choice of 
$\omega_M,H_M$ that we pick at any given $M$. The idea of Hamiltonian
renormalisation \cite{18} is now to define a renormalisation 
flow on the family of theories $(\omega_M,H_M)$ starting with some initial
family $(\omega^{(0)}_M,H^{(0)}_M)$. For this one picks coarse graining 
maps $J^{(n)}_{MM'},\; M<M',\;n\in \mathbb{N}_0$ with 
$J^{(n)}_{MM'}:\; {\cal H}^{(n+1)}_M\; \to {\cal H}^{(n)}_{M'}$ and picks 
the states $\omega^{(n+1)}_M$ such that these are isometric injections. 
Then one defines $H^{(n+1)}_M:=J_{MM'}^\dagger H^{(n)}_{M'} J^{(n)}_{MM'}$.
This defines a flow of families $n\mapsto (\omega^{(n)}_M,H^{(n)}_M)$ a 
fixed point of which then accomplishes the task (\ref{2.22}). Namely
a fixed point $(\omega,H)$ 
of this flow results in an inductive limit Hilbert space 
\cite{21} 
$\cal H$ with isometric injections $J_M\; {\cal H}_M\to {\cal H}$   
and where $P_M=J_M\; J_M^\dagger$ is the corresponding projection 
together with a quadratic form $H$ on $\cal H$ such that 
$H_M=J_M^\dagger H J_M$. If the quadratic form is closable and is bounded
from below, we do get a self-adjoint operator via the Friedrichs extension
\cite{24}.

The 
details of the flow depend on the choice of the coarse graining maps
which are typically of the form 
\be \label{2.23}
J^{(n)}_{MM'} W(f,g)\Omega^{(n+1)}_M=W(I_{MM'}(f,g))\Omega^{(n)}_{M'}
\ee
where $I_{MM'}$ is a coarse graining map on the smearing functions which 
maps those at coarse resolution $M$ into those at finer resolution $M'$. 
In order that the inductive limit construction works one must pick 
those $I_{MM'}$ such that $I_{M_2 M_3}\circ I_{M_1 M_2}=I_{M_1 M_3}$    
(``cylindrical consistency'') for $M_1<M_2<M_3$ and such that the label
set $\cal M$ of cut-offs is partially ordered and directed. This 
can be ensured for example using multi resolution analysis familiar from
wavelet theory (see e.g. \cite{25} and references therein).

If one has managed to remove the cut-off $M$ in this manner, one still 
must take the 
theromdynamic limit $R\to \infty$ which meets new technical challenges
and can give rise to different ``phases'' of the theory \cite{21}.\\
\\
Summary:\\
The reduced phase space quantisation can be carried out entirely within   
the Hamiltonian or canonical framework but in QFT renormalisation
is typically required. Since the flow equations can typically not be solved 
in closed form, some approximation scheme is required. This is the canonical
or Hamiltonian equivalent of truncations performed in ASQG.

\subsection{Reduced phase space quantisation: Path integral formulation}
\label{s2.4}

We now turn to a path integral framework. We start again from some 
$\omega_M,H_M$ with GNS data $({\cal H}_M,\Omega_M,\rho_M)$. The cyclic 
vector $\Omega_M$ is generically not an eigenstate of $H_M$. Suppose 
that $H_M$ is bounded from below (w.l.g. by zero)
and has a unique ground state $\Omega_{0,M}$ i.e. $H_M \Omega_{0,M}=0$.
For sufficiently complicated $H_M$ this vector will not be computable 
in closed form.
Let $U_M(t)=exp(-i t H_M)\; t\in \mathbb{R};\;\;C_M(s)=\exp(-s H_M),\;
s\ge 0$ 
be the unitary group and contraction semi-group generated by $H_M$ respectively
and consider the following normalised vectors
\be \label{2.24}
\Omega^T_{\pm,M}:=U_M(\pm T)\;\Omega_M,\;\;
\Omega^T_M:=\frac{C_M(T) \Omega_M}{||C_M(T)\;\Omega_M||}
\ee
If the strong limits of (\ref{2.24}) exist, denoted by $\Omega^\pm_M,
\Omega^C_M$ respectively, then these obey 
$U_M(t)\Omega^\pm_M=\Omega^\pm_M,\; C_M(s)\Omega^C_M=\Omega^C_M$, i.e. 
they are ground states of $H_M$. By the assumed uniquenes 
\be \label{2.25}
\Omega^\pm_M=e^{i\alpha_\pm} \Omega_{0,M},\;
\Omega^C_M=e^{i\alpha_C} \Omega_{0,M}
\ee
for some phases $\alpha_\pm,\alpha_C$. We have
\be \label{2.25a}
\lim_T\; <U_M(T)\Omega_M,U_M(-T)\Omega_M>=e^{-i[\alpha_+ -\alpha_-]},\;\;
\lim_T\; ||C_M(t)\Omega_M||^2=\lim_T <\Omega_M,C_M(-2T) \Omega_M>
\ee
whence 
\ba \label{2.25b}
&& <\Omega_{0,M},.\;\Omega_{0,M}>=
\lim_{T\to \infty}
<U_M(T)\Omega_M,.\;U_M(-T) \;\Omega_M>
e^{i[\alpha_+ - \alpha_-]}
=\lim_{T\to \infty} \frac{<U_M(T)\Omega_M,.\;U_M(-T) \;\Omega_M>}
{<\Omega_M,.\;U_M(-2T) \;\Omega_M>}
\nonumber\\
&& <\Omega_{0,M},.\;\Omega_{0,M}>
=\lim_T \frac{<C_M(T)\Omega_M,.\;C_M(T) \;\Omega_M>}
{<\Omega_M,.\;C_M(2T) \;\Omega_M>}
\ea

\subsubsection{Lorentzian formulation}
\label{s2.4.1}

In the Lorentzian formulation 
we are interested in the time ordered correlation functions (Feynman 
N-point functions) 
\be \label{2.26}
F_{N,M}((t_1,f^1),..,(t_N,f^N))
:=<\Omega_{0,M},\;{\cal T}[W_{t_N}(f^N,0)..W_{t_1}(f^1,0)]\;
\Omega_{0,M}>_{{\cal H}_M} 
\ee
where $W_t(f,g):=U_M(t) W(f,g) U_M(t)^{-1}$ and $f_a,g^a$ vanish for 
$a>M$. The time ordering symbol $\cal T$ orders the time dependence from 
left to right with largest time to the outmost left. Upon derivation by 
$f^k_{a_k},\; k=1,..,N; a\le M$ at $f^1=..=f^N=0$ we obtain 
\be \label{2.27}
F_{N,M}((t_1,a_1),..,(t_N,a_N))
:=<\Omega_{0,M},\;{\cal T}[Q_M^{a_N}(t_N)..Q_M^{a_1}(t_1)]\;
\Omega_{0,M}>_{{\cal H}_M} 
\ee
with $Q^a_M(t):=U_M(t) Q^a U_M(t)^{-1},\; a\le M$. This can also be written
as 
\be \label{2.28}
F_{N,M}((t_1,a_1),..,(t_N,a_N))=i^{-N}\;
[\frac{\delta^N}{\delta f_{a_N}(t_N)..\delta f_{a_1}(t_1)}\; \chi_M(f)]_{f=0},\;
\ee
with the generating functional 
\be \label{2.29}
\chi_M(f):=<\Omega_{0,M},\;{\cal T}[e^{i\int_{\mathbb{R}}\;dt\; f_a(t)\;
Q_M^a(t)}]\;\Omega_{0,M}>_{{\cal H}_M}
\ee
The way in which (\ref{2.29}) written is not practically useful as we 
have not explicit access to $\Omega_{0,M}$ but rather to $\Omega_M$. 
Suppose that $f$ has compact time support in $[-\tau,\tau]$ then we have 
explicitly using (\ref{2.25b})
\ba \label{2.29a}
&& \chi_M(f) 
= \lim_{T\to\infty}
\frac{
<\Omega_M, e^{iT H_M}\; {\cal T}[e^{i \int_{-\tau}^\tau\;
\;dt\; f_a(t) Q^a_M(t)}] \;e^{iT H_M} \Omega_M>
}
{
<\Omega_M, e^{2iT H_M} \Omega_M>
}
\\
&=& \lim_{T\to\infty}\lim_{N\to\infty}
\frac{
<\Omega_M, e^{i(T-t_{N-1}) H_M}\; e^{i\Delta_N f_a(t_{N-1}) Q^a}\;
e^{i\Delta_N H_M}\;..\;
e^{i\Delta_N H_M}\; e^{i\Delta_N f_a(t_{-N}) Q^a}
e^{i(t_{-N}-(-T)) H_M}\;\Omega_M>
}
{
<\Omega_M, e^{2iT H_M} \Omega_M>
}
\nonumber
\ea
where $\Delta_N=\frac{\tau}{N},\;t_k=k\Delta_N,\;k=-N,..,N-1$. 
For suffieciently large $T$ we have $T\ge \tau$ and by extending $f$ by 
zero to $\mathbb{R}-[\tau,-\tau]$ we can identify $\tau=T$. Then 
\ba \label{2.30}
\chi_M(f) &=&
\lim_{T\to\infty}\lim_{N\to\infty}
\frac{Z_{M,T,N}(f)}{Z_{M,T,N}(0)},\;\;
\\
Z_{M,T,N}(f) &=&
<\Omega_M, e^{i\Delta_N H_M}\; e^{i\Delta_N f_a(t_{N-1}) Q^a}\;
e^{i\Delta_N H_M}\;..\;
e^{i\Delta_N H_M}\; e^{i\Delta_N f_a(t_{-N}) Q^a}
e^{i\Delta_N H_M}\;\Omega_M>
\nonumber
\ea
We consider the integral kernel of the unitary propagator
\be \label{2.31}
\int [dq']_M\; U_{M,T,N}(q,q')\psi_M(q'):=(e^{i\Delta_N H_M} \psi_M)(q)
\ee
where $[dq]_M$ is the Lebesgue measure of the $q^a,\; a\le M$. 
As at finite $M$ the representation is in the unitarity class of the 
Schr\"odinger representation we consider $Q^a$ as multiplication by $q^a$ 
and $P_a$ as $i\partial/\partial q^a$. Then one has to argue that for 
$N\to \infty$ at finite $T$
\ba \label{2.32}
&& U_{M,T,N}(q,q')=\int\; [d\hat{q}]_M\;U_{M,T,N}(q,\hat{q}) 
\delta_{q'}(\hat{q})
=[e^{i\Delta H_M(Q,P)} \delta_{q'}](q)
=\int\;[d\frac{p}{2\pi}]_M\; e^{i\Delta H_M(q,p)} \; 
e^{-i\;\sum_{a\le M}\; p_a(q-q')^a}
\nonumber\\
&=& \int\;[d\frac{p}{2\pi}]_M\; e^{i\Delta H_M(q,p)} \; 
e^{-i\;\sum_{a\le M} p_a(q-q')^a}
\ea
Without going into the details of $H_M$ this step is hard to justify,
for Schr\"odinger type of operators one may use Feynman-Kac arguments 
\cite{24}. 

Using this heuristics we obtain 
\ba \label{2.33}
&& Z_{M,T,N}(f)=
\prod_{k=-N}^N \; [dq_k]_M\; \prod_{l=-N}^{N-1} \; [\frac{dp^l}{2\pi}]_M
\;\overline{\Omega_M(q_N)} \;\Omega_M(q_{-N}) \;
e^{-i\sum_{k=-N}^{N-1}\; \Delta_N [\sum_{a\le M} p^k_a 
\frac{q^a_{k+1}-q^a_k}{\Delta_N}-H(p^k,q_k)]}\; 
\times \nonumber\\
&& e^{i\sum_{k=-N}^{N-1} \sum_{a\le M}\; f_a(t_k) q^a_k}
\ea
In what follows we mean by the symbolic expression
\be \label{2.34}
Z(f)=\int\; [dp]\; [dq]\; \overline{\Omega(q(\infty))}\;\Omega(q(-\infty))
\;
e^{i\int\; dt\; f_a(t) q^a(t)} \; 
e^{-i\int\; dt\; [p_a(t)\dot{q}^a(t)\;-H(q(t),p(t))]} 
\ee
some form of limit of (\ref{2.33}) as $N\to\infty, T\to\infty, M\to\infty,
R\to \infty$ in that precise order, if it exists. 
$Z(f)$ is the generating functional 
of all Feynman functions and in $Z(f)/Z(0)$ the vacuum to vacuum 
correlations are removed. 

It is important to note that the path integral (\ref{2.34}) is over phase 
space and not over configuration space to begin with. Integrating out the 
momenta explicitly is only possible when $H$ has a sufficiently simple 
dependence on $p$ and even then it may modify the ``Lebesgue measure''
$[dq]$. We will see this explicitly for the case of GR in section 
\ref{s4}. Note also the dependence of $Z(f)$ on the cyclic vector $\Omega$
or equivalently the corresponding state state $\omega$. 

\subsubsection{Euclidian formulation}
\label{s2.4.2}

There is also an Euclidian version of (\ref{2.34}). We obtain it by 
going back to (\ref{2.29}) and now using the second 
way in (\ref{2.25a}) to express $\Omega_{0,M}$ in terms of $\Omega_M$
\be \label{2.35}
\chi_M(f)=\lim_{T\to\infty}\lim_{N\to\infty}
\frac{
<\Omega_M,e^{-T H_M}\; e^{-i t_{N-1} H_M} 
e^{i \Delta_N\; f_a(t_{N-1}) Q^a} e^{i\Delta_N H_M}..
e^{i\Delta_N f_a(t_{-N}) Q^a} e^{i t_{-N} H_M} e^{-T H_M}\Omega_M>
}
{
<\Omega_M, e^{-(T-(-T)) H_M} \Omega_M>
}
\ee
We define $2N$ smearing functions $f^k_a$ by $f^k_a:=f_a(t_k)$ 
with respect to the label $a$ (which are
considered independent of $t$). Then
we analytically continue $\tau\to i\sigma,\sigma>0$ which  
means $t_k=k \tau/N\to i s_k,\; s_k=k \sigma/N$. Then we 
define, abusing the notation $f_a(s_k):=-f^k_a$. We take $T$ sufficiently 
large $T\ge \sigma$ and extend $f$ trivially to $\mathbb{R}-[-\sigma,\sigma]$
so that we can set $T=\sigma$. 
Then (\ref{2.35}) becomes with $\Delta_N=\frac{T}{N}$
\be \label{2.36}
\chi_M(f)=
\lim_{T\to\infty}\lim_{N\to\infty}
\frac{Z^E_{M,T,N}(f)}{Z^E_{M,T,N}(0)},\;
Z^E_{M,T,N}(f)=<\Omega_M, e^{-\Delta_N H_M}\; e^{\Delta_N f_a(s_{N-1})
Q^a}\;..\; e^{-\Delta_N f_a(s_{-N})}\; e^{-\Delta_N H_M} \Omega_M>
\ee
The rest of the analysis is the same as in the ``Lorentzian'' case and 
results in the symbolic formula
\be \label{2.37}
Z^E(f) 
\int\; [dp]\; [dq]\; \overline{\Omega(q(\infty))}\;\Omega(q(-\infty))
\;
e^{\int\; ds\; f_a(s) q^a(s)} \; 
e^{-\int\; ds\; [i\; p_a(s)\dot{q}^a(s)\;+H(q(s),p(s))]} 
\ee
which is the generating functional of Schwinger functions which are obtained 
as (without vacuum correlators)
\be \label{2.38} 
S_N((s_N,a_N),..,(s_1,a_1))=[\frac{\delta^N}{\delta 
f_{a_N}(s_N)..\delta f_{a_1}(s_1)} \frac{Z^E(f)}{Z^E(0)}]_{f=0}
\ee

\subsubsection{Unconstrained phase space path integral}
\label{s2.4.3}

We now want to relate these reduced phase space path integrals to path 
integrals over the unconstrained phase space. This can be accomplished 
by recalling that $y_j$ is eliminated by using the constraint 
$\hat{C}_j=y_j+h_j=0$ while $x^j$ is eliminated using the gauge fixing 
condition $G^j=x^j-k^j$. Consider again first the Lorentzian version. Then 
\be \label{2.39}
Z(f)=\int\; [dq]\; [dp]\; [dy]\; [dx]\; 
\delta[\hat{C}]\; \delta[G]\; \overline{\Omega(q(\infty))}\Omega(q(-\infty))
e^{i<f\cdot q>} \; e^{-i[<p\cdot \dot{q}-H>]}
\ee
with $\delta[\hat{C}]=\prod_t \delta(\hat{C}_j(z(t))$ and  
$\delta[G]=\prod_t \delta(x^j(t)-k^j(t))$ and $<(.)>=\int\;dt (.)$. 
Using the relation $C_j=M_j^k \hat{C}_k$ between the original constraints 
$C$ and the constraints $\hat{C}$ solved for the $y$ we have 
$\{C_j,G^k\}=M_j^k$ when $\hat{C}_k=0$ whence 
\be \label{2.40}
\delta[\hat{C}]=|\det[M]|\; \delta[C],\;
\det[M]=\prod_t \det(M(t))
\ee
Strictly speaking, formula (\ref{2.40}) is not entirely correct: If $C_j$
labels constraints
quadratic in the momenta $y$ then the full constraint surface 
$\bar{\Gamma}$ defined by
$C_j=0$ for all $j$ 
of the phase space $\Gamma$ 
has in fact many branches corresponding, for each $j$, to a choice of
the two possible roots $y_j=-h_j^\pm$ where in the canonical quantisation
we only quantised the sector $y_j=-h_j,\; h_j:=h^+_j$ corresponding 
to the positive root for all $j$. 
Since typically $h_j^-$ is negative that sector is selected 
if we restrict the integral over 
$y_j$ to the negative half-axis but this would prevent us from 
using Gaussian integral techniques to integrate out $y$. To analyse 
this issue, let $\sigma_j=\pm 1$ and  
$\bar{\Gamma}_\sigma=\{(q,p,x,y)\;\; y_j=-h^{\sigma_j}(q,p,x)$ then we have 
a disjoint (up to measure zero sets)
union $\bar{\Gamma}=\cup_\sigma \bar{\Gamma}_\sigma$. Accordingly, we may 
pick on each sector $\bar{\Gamma}_\sigma$ an individual gauge fixing 
condition $G^j_\sigma=x^j-k^j_\sigma$ which defines the 
corresponding sector $\hat{\Gamma}_\sigma=\{(p,q,x,y),\;
y=-h^\sigma(q,p,x),\;x=k_\sigma\}$ of the reduced phase space.
We assume now 1. that we can write 
the original constraints $C_j$, i.e. before solving them for $y_j$, in the
form      
\be \label{2.40a}
C_j=\sum_k P_j^k\; \hat{C}_k^+\;\hat{C}_k^-,\;\hat{C}_k^\pm=y_k+h_k^\pm
\ee
for some invertible matrix $P$.
This is not the most general situation but it will be satisfied 
for the concrete theory that we study later on. Then the gauge 
stability condition on the sector $\sigma$ is 
given by 
\be \label{2.40b}
\dot{k}_\sigma^j=\{f^k C_k,G^j_\sigma\}_{\bar{\Gamma}_\sigma}
=-f^k\; P_k^j\; \sigma_j (h^+_j - h^-_j)=:f^k\; [M_\sigma]_k^j
\ee
which can be solved for $f^j=f^j_\sigma$. The reduced Hamiltonian 
on the sector $\sigma$ is determined by the condition that for 
any function $F$ of $q,p$
\be \label{2.40c}
\{H_\sigma,F\}=\{f^k C_k,F\}_{\hat{\Gamma}_\sigma,f=f_\sigma}
=-f^k_\sigma\; \sum_j P_k^j\; \sigma_j\; (h^+_j - h^j_-)\; 
\{h_j^{\sigma_j},F\}
=\sum_j\;\dot{k}^j_\sigma \{h^{\sigma_j},F\}
\ee
thus 
\be \label{2.40ca}
H_\sigma=\sum_j \dot{k}^j_\sigma\; h^{\sigma_j}_j(q,p,x=k_\sigma)
=-[\sum_j \dot{x}^j\;y_j]_{\hat{\Gamma}_\sigma}
\ee
which returns the earlier result for the ``totally positive'' sector.
Suppose now that in fact 2. $h_j^\sigma=\sigma h_j$ and that 
3. $h_j(q,p,x)=h_j(q,p,-x)$. Then for the choice 
$k^j_\sigma:=\sigma_j\; k^j$ for $k^j$ the choice for the totally 
positive sector we find 
\be \label{2.40cb}
H_\sigma=\sum_j \dot{k}^j\; h_j(q,p,x=k)=H
\ee
is {\it sector independent} and agrees with the value of the 
totally positive sector. Next we have 
(subscript $\sigma$ denotes 
sector restriction)   
\ba \label{2.40d}
&& \delta(G)\; |\det(\{C,G\})|\; \delta(C)
=\sum_\sigma\; \delta(G)_\sigma \;|\det(\{C,G\})|_\sigma\; \delta(C)_\sigma
\nonumber\\
&=& \sum_\sigma \;\delta(G)_\sigma \;|\det(M_\sigma\})|\; 
\prod_j\; \delta(C_j)_\sigma
= \sum_\sigma \;\delta(G)_\sigma \;|\det(M_\sigma\})|\; 
\prod_j\; \delta(-\sum_k P_j^k (h_k^+ - h_k^-)\sigma_k 
\hat{C}^{\sigma_k}_k)
\nonumber\\
&=& \sum_\sigma \; \delta(G)_\sigma \;\delta(\hat{C}_\sigma),\;\;\;
\delta(\hat{C}_\sigma)= \prod_j \delta(\hat{C}_j^{\sigma_j})
\ea
Therefore using (\ref{2.40c})
\ba \label{2.40e}
&&[\sum_\sigma\;1] [\int\; [dq\;dp]\; e^{i\; <f,q>-i\;[<p,\dot{q}>-H]}]
= \sum_\sigma \int\; [dq\;dp\;dx\;dy]\; \delta(G_\sigma)\; 
\delta(\hat{C}_\sigma)
e^{i\; <f,q>-i\;[<p,\dot{q}>-H_\sigma]}]
\nonumber\\
&=& \int\; [dq\;dp\;dx\;dy]\; |\det(\{G,C\})|\; \delta(G)\; 
\delta(C)\;
e^{i\; <f,q>-i\;[<p,\dot{q}>+<y,\dot{x}]}
\ea
It follows that under the assumptions 1.-3. formula (\ref{2.40}) 
is correct up to the constant factor $\sum_\sigma\; 1$ which drops 
out in the quotient $Z(f)/Z(0)$ that removes pure vacuum to vacuum 
amplitudes.

We also can bring the delta
distribution for the constraints to the exponent using a Lagrange 
multiplier $N^j$. This results in (again constant factors e.g. powers 
of $2\pi$ are 
dropped as we consider the fraction $Z(f)/Z(0)$)
\be \label{2.41}
Z(f)=
\int\; [dq]\; [dp]\; [dy]\; [dx]\;[dN]\; 
|[\det(M)]|\; \delta[G]\; \overline{\Omega(q(\infty))}\Omega(q(-\infty))
e^{i<f\cdot q>} \; e^{-i[<p\cdot \dot{q}+y\cdot\dot{x} -N\cdot C>]}
\ee
If we introduce independent ghost and and anti-ghost fields $\eta^j,\rho_j$ 
together with their ``Berezin measures'' $[d\eta] [d\rho]$ one can 
also bring the Dirac matrix determinant into the exponent. The 
$\delta$ distribution of $G$ can be brought to the exponent 
by 
introducing an additional Lagrange multiplier integral 
$\delta[G]=\int\;[dl]\;\exp(i<l\cdot G>)$. An argument often used is 
that $Z(f)$ is unchanged when replacing 
$G$ by $G-l$ as when $Z(f)$ is gauge invariant. If that was true 
in our case we could integrate 
both numerator and denominator with $\int [dl] \exp(-<l\cdot\kappa l>$ where 
$\kappa$ is a positive integral kernel which would 
then replace $\delta[G]$ by $\exp(-<l\cdot\kappa l>)$. However, this 
is manifestly not the case because unwinding what we have done would then 
replace the physical Hamiltonian by $H=\dot{k}\cdot h(P,Q,;k)\to 
[\dot{k}+\dot{l}]\cdot h(Q,P;k+l)\not=H$ unless $h$ is independent 
of $k$ (which is true in the case considered here (no explicit 
time dependence)) and if $\dot{l}=0$. However, the latter condition 
would imply that $l$ is the same at all times and thus we could 
not solve $\delta[G]=\prod_{t,j} \delta(G^j(t)-l^j)$ for all $t$.   
To obtain no explicit time dependence we need $\dot{k}=$const. but not 
$\dot{k}=0$ (then we would have $H=0$). 
The reason why this 
happens is due to the different notions of gauge transformations in 
the Lagrangian and Hamiltonian setting \cite{28}: The Lagrangian gauge  
transformations are kinematical, i.e. they come from a Lie algebra
(in GR it is the Lie algebra of vector fields, the Lie algebra 
of the diffeomorphism group).
In the Hamiltonian setting the gauge transformations are generated by the 
constraints and they do not form a Lie algebra if the structure functions 
are not constant on the phase space (which is precisely the case in GR).
The two notions coincide only on shell, when the classical 
equations of motion hold
but of course in the path integral the set of classical paths has measure
zero.  

We will therefore leave $\delta[G]$ untouched and obtain
the final expression
\be \label{2.42}
Z(f)=   
\int\; [dq]\; [dp]\; [dy]\; [dx]\;[dN]\;
|\int\; [d\eta]\; [d\rho]\; e^{-i<\eta\cdot M\cdot \rho>}|\;
\overline{\Omega(q(\infty))}\Omega(q(-\infty))\; \delta[G]
e^{i<f\cdot q>} \; e^{-i[<p\cdot \dot{q}+y\cdot\dot{x} -N\cdot C>]}\;
\ee
where again $M=(\{C,G\})$ is the Dirac matrix between constraints and gauge 
fixing conditions. Note that the ghost 
integral produces $\det(M)$ rather than $|\det(M)|$ (times an infinite power 
of $i$ which drops out in the modulus) which is why we keep the modulus,
a fact often ignored. 

The expression (\ref{2.42}) is as close as one can 
get to the usual formula $\int\; [d\Phi] e^{-iS(\Phi)}$, where $\Phi$ denotes 
the collection of all configuration fields, without furher specifying 
the constraints. Indeed one would like to 
integrate out $p,y$ in (\ref{2.42}) which would turn the Hamiltonian 
action that has appeared in the exponent into the Lagrangian action.
Whether this is possible explicitly depends on the $y,p$ dependence of 
$C$. Even when possible, this may result in an additional measure 
contribution coming from a corresponding Jacobean. We will see that this 
precisely happens in GR where for our model 
it turns out that the dependence of $<N\cdot C>$ and $<\eta, M\cdot \rho>$
is at most quadratic and linear in $y$ respectively and that the 
appearing Jacobean can be absorbed into a redifinition of the Dirac 
matrix $M$. That this is possible in the Lorentzian framework rests on the 
following elementary property of the Gaussian integral: 
\be \label{2.42a}
\int_{\mathbb{R}}\; du\; e^{z\;u^2}
\ee
It exists for all $z\in \mathbb{C}$ such that either $\Re(z)<0$ or 
$\Re(z)=0, \Im(z)\not=0$. It does not exist for $\Re(z)>0$.   
Finally note also that correctly there is only a smearing 
function $f$ for the true configuration degrees of freedom $q$ and not for 
$x$.\\
\\
We now repeat this analysis for the Euclidian formulation. Almost all
the work has been done already, we just need to remove two factors of $i$.
We find
\be \label{2.43}
Z^E(f)=   
\int\; [dq]\; [dp]\; [dy]\; [dx]\;[dN]\;
|\int \;[d\eta]\; [d\rho]\; e^{<-\eta\cdot M\cdot \rho>}|\;
\overline{\Omega(q(\infty))}\Omega(q(-\infty))\; \delta[G]
e^{<f\cdot q>} \; e^{-i[<p\cdot \dot{q}+i\;y\cdot\dot{x} -N\cdot C>]}\;
\ee
This looks innocent but it is actually desastrous for the Euclidian 
programme: When integrating over $y$ one does not invert the Legendre 
transform between the Lagrangean and the Hamiltonian formulation 
because the symplectic term $y\cdot x$ comes with a relative factor of 
$i$ as comprared to the constrained Hamiltonian $N\cdot C$. This is the 
reason why we prefer the Lorentzian version of the Wetterich equation.
We will see this in more detail in section \ref{s4}. One 
may argue that the Euclidian path integral is better defined than the 
Lorentzian version because the Lorentzian integral is oscillatory while 
the Euclidian one is damped. This is in fact true when the Hamiltonian 
is bounded from below. However, exactly for GR the situation is opposite
because the Euclidian action is not bounded from below due to the famous 
negative conformal mode \cite{26}. This corresponds to case
$\Re(z)>0$ for which (\ref{2.42a}) does not exist.

\section{Asymptotically safe quantum gravity: Summary and logic from 
point of view of CQG}
\label{s3}

As was stressed and reviewed in the previous section one can go back and 
forth between the Hamiltonian and path integral formulation of the 
QFT of a gauge theory. However, whatever formulation one prefers, 
the rigorous matehmatical formulation is a challenge in interacting QFT 
and renormalisation methods are an essential ingredient in order 
to arrive at a mathematically well defined theory with predictive power.
From that perspective we consider ASQG as a particular incarnation of that 
general theme. 

The starting point of the asymptotic safety approach for a general 
theory is the generating functional $Z(f)$ of time ordered N-point functions 
(Lorentzian version) or $Z^E(f)$ of N-point Schwinger functions (Euclidian
version) from which one can obtain 
effective actions by the usual methods. 
As we stressed in the derivation of the previous section,
the symbolic expressions (\ref{2.34}) and (\ref{2.37}) (reduced phase
space path integral) and (\ref{2.42}) and (\ref{2.43}) (full phase space 
path integral), while well motivated, they and their effective actions
are ill-defined as they stand. The idea is to assume that the effective 
actions  
are well-defined, to derive a renormalisation group equation from them 
which by itself is well defined (Wetterich equation) and then to 
forget about the derivation and rather take the renormalisation group 
equation (RGE) as the fundamental guiding principle for the construction of the 
theory. A solution of the RGE then {\it defines} the theory. In 
the subsequent subsections we sketch elements of that programme.\\ 
\\
ASQG practitioners can mostly skip this section, except that we develop
both the Lorentzian and Euclidian versions which leads to adaptions 
by powers of $i$ in the Wetterich equation, the heat kernel calculus and 
the Laplace versus Fourier transform with respect to heat kernel time.
There are also some remarks concerning the treatment of the ghost part 
of the effective action
and the functional derivatives at zero of the Wetterich equation with respect 
to the field on which the EAA depends, which are not often mentioned 
in the ASQG literature, see however \cite{26a}.

\subsection{Effective action without gauge invariance}
\label{s3.1}

Given the generating functional $Z_L(f), Z_E(f)$ of time ordered or Schwinger 
N-point functions respectively, formally defined by 
\be \label{3.1}
Z_L(f):=\int\; [dq]\; e^{-i S_L(q)}\; e^{i<f,q>},\;\;
Z_E(f):=\int\; [dq]\; e^{- S_E(q)}\; e^{<f,q>},\;\;
<f,q>:=<f\cdot q>
\ee
the functional 
\be \label{3.1a}
C_L(f):=i^{-1} \ln(Z_L(f),\; C_E(f):=\ln(Z_E(f)) 
\ee
is the generating functional of ``connected'' such functions and its 
Legendre transform 
\be \label{3.2}
\Gamma_L(\hat{q}):={\rm extr}_f\;[-C_L(f)+<f,\hat{q}>],\;
\Gamma_E(\hat{q}):={\rm extr}_f\;[-C_E(f)+<f,\hat{q}>]
\ee
the effective action or generating functional of ``one particle irreducible''
(1PI) such functions. 
This terminology comes from their graphical 
interpretation,
see e.g. \cite{27}. If any of these objects is well defined in the sense 
that its functional derivatives yield non-singular tempered distributions,
so are the others and one has ``solved the theory''. The effective action is 
the most compact definition of the theory because the 1-PI functions are 
the ``atoms'' of all correlation functions.

\subsection{Background field method without gauge invariance}
\label{s3.2}

It maybe convenient for various purposes to introduce a background 
field $\bar{q}$, especially in GR where it can be used to define 
preferred covariant derivatives with respect to a background metric
in order to construct additional structures (such as those used 
in asymptotic safety). The method has been invented in \cite{12} in fact
first in the Lorentzian signature setting and is naturally used in QFT 
in CST for example in order to define the Hadamard condition on the 
2-point function \cite{15,13}. 

The background objects will carry an additional overbar and are defined
by (Euclidian: $s=0$; Lorentzian: $s=1$) 
\ba \label{3.3}
\bar{Z}_s(f,\bar{q}) &:=& \int\;[dh]\; e^{-i^s S_s(\bar{q}+h)}  
\; e^{i^s<f\cdot h>}
\nonumber\\
\bar{C}_s(f,\bar{q}) &:=& i^{-s} \ln(\bar{Z}_s(f,\bar{q}))
\nonumber\\
\bar{\Gamma}_s(\hat{q},\bar{q}) &:=& {\rm extr}_f\; [-\bar{C}_s(f,\bar{q})
+<f\cdot \hat{q}>]
\ea
They are explicitly dependent on $\bar{q}$. The relation to the background 
independent objects are as follows:\\
Introducing a new integration variable $q=\bar{q}+h$ in the first line 
we find 
\ba \label{3.4}
\bar{Z}_s(f,\bar{q}) &=& e^{-i^s\; <f,\bar{q}>}\; Z_s(f)
\nonumber\\
\bar{C}_s(f,\bar{q}) &=& C_s(f)-<f\cdot \bar{q}>
\nonumber\\
\bar{\Gamma}_s(\hat{q},\bar{q}) &=& \Gamma_s(\hat{q}+\bar{q})
\ea
Since what we are interested in is really the background independent 
object $\Gamma_s(\hat{q})$ we may obtain it from the background dependent 
object $\bar{\Gamma}_s(\hat{q},\bar{q})$ by the simple formula
\be \label{3.5}
\Gamma_s(\hat{q})=
[\bar{\Gamma}_s(\hat{q}',\bar{q})]_{\hat{q}':=0,\bar{q}:=\hat{q}}
\ee
The reason for not setting instead $\hat{q}'=\hat{q},\; \bar{q}=0$ is 
because a zero background may jeopardise intermediate constructions that 
one performs with it (e.g. a zero background metric cannot be inverted).  

\subsection{Effective action with gauge invariance}
\label{s3.3}

We now consider a gauge theory with gauge group $\mathfrak{G}$ and Lie algebra 
$\mathfrak{g}$. The group $\mathfrak{G}$ acts on the configuration 
space $\mathfrak{Q}$ of the fields $q$ by gauge 
transformations 
\be \label{3.6}
\alpha:\; \mathfrak{G}\times \mathfrak{Q}\to \mathfrak{Q};\;
q\mapsto \alpha_g(q),\; \alpha_{g}\circ \alpha_{g'}=\alpha_{gg'}
\ee 
for $g,g'\in \mathfrak{G}$. A gauge fixing condition 
\be \label{3.7}
G:\; \mathfrak{Q} \to \mathfrak{g};\; q\mapsto G(q)
\ee
is a Lie algebra valued function on the configuration space subject to the 
condition that for each $q\in \mathfrak{Q}$ there exists a unique 
$g(q)\in \mathfrak{G}$ such that $G(\alpha_{g(q)}(q))=0$. If 
$\exp:\; \mathfrak{g}\to \mathfrak{G}$ is the exponential map then 
one can instead require that for each $q\in \mathfrak{Q}$ there exists 
a unique $u(q)\in \mathfrak{g}$ such that $G(\alpha_{\exp(u(q))}(q))=0$.   
The conditions are not equivalent because the exponential map is not a 
bijection, typically it is neither injective (e.g. the Lie algebra 
of $U(1)$ is $\mathbb{R}$ but all $u=v+2\pi n,n\in \mathbb{Z}, v\in[0,2\pi)$ 
map to $g=e^{iu}=e^{iv}$) nor 
surjective (general $g\in \mathfrak{G}$ are not in the component of the 
identity). We will, as is customary ignore these complications and 
work with the Lie algebra condition on the gauge fixing condition.
Then we have the Fadeev-Popov identity
\be \label{3.8}
1=\int_{\mathfrak{g}}\; [du]\; \delta[G(\alpha_{\exp(u)}(q))]\;
|\det[\frac{\delta}{\delta u}\;G(\alpha_{\exp(u)}(q))]|
\ee
for any $q\in \mathfrak{Q}$. 

Let now $F:\; \mathfrak{Q}\to \mathbb{C}$ be a gauge invariant function
$F\circ \alpha_g=F\;\forall \; g\in \mathfrak{G}$. Then the symbolic functional
integral
\be \label{3.9}
I_F:=\int_{\mathfrak{Q}}\;[dq]\;F(q) 
\ee
can also be written
\be \label{3.9b}
I_F:=\int_{\mathfrak{g}}\; [du] \;\int_{\mathfrak{Q}}\;[dq]\;F(q) 
\delta[G(\alpha_{\exp(u)}(q))]\;
|\det[\frac{\delta}{\delta u}\;G(\alpha_{\exp(u)}(q))]|
\ee
We introduce a new integration variable $q'=\alpha_{\exp(u)}(q)$. If
$\alpha_g$ is a bijection for any $g\in G$ (e.g. no fixed points) then 
$\alpha_{\exp(u)}$ is a permutation on $\mathfrak{Q}$ so that the Jacobean
or Radon-Nikodym derivative is
\be \label{3.10}
J(q,u):=\frac{d[\alpha_{\exp(u)}(q)]}{[dq]}=
|\det[\frac{\delta\alpha_{\exp(u)(q)}}{\delta q}]|=1
\ee
It would be enough to require that $J(q,u)=J(u)$ just depends on $u$. 
Clearly $F(q)=F(q')$ by gauge invariance and provided 
that 
\be \label{3.11}
\det[\frac{\delta}{\delta} G(\alpha_{\exp(v+u)}(q))]
-\det[\frac{\delta}{\delta v} G(\alpha_{\exp(v)\exp(u)}(q))]=O(v) 
\ee
we have
\be \label{3.12}
\det[\frac{\delta}{\delta u} G(\alpha_{\exp(u)}(q))]    
=\{\det[\frac{\delta}{\delta v} G(\alpha_{\exp(v+u)}(q))]\}_{v=0}
=\{\det[\frac{\delta}{\delta v} G(\alpha_{\exp(v)}(q'))]\}_{v=0}
\ee
Then after relabelling $q'$ by $q$
\be \label{3.13}
I_F=\int_{\mathfrak{g}}\;[du]\;J(u)\;
\int\;[dq]\; F(q)\;  
\delta[G(q)]\; 
|\det[\frac{\delta}{\delta v} G(\alpha_{\exp(v)}(q))]|_{v=0}
\ee
or for any fixed gauge invariant functional $F_0$
\be \label{3.14}
\frac{I_F}{I_{F_0}}=
\frac{
\int\;[dq]\; F(q)\;  
\delta[G(q)]\; 
|\det[\frac{\delta}{\delta v} G(\alpha_{\exp(v)}(q))]|_{v=0}
}
{
\int\;[dq]\; F_0(q)\;  
\delta[G(q)]\; 
|\det[\frac{\delta}{\delta v} G(\alpha_{\exp(v)}(q))]|_{v=0}
}
\ee
If $G$ meets the stated conditions then so does $G-w$ where 
$w\in \mathfrak{g}$ and since (\ref{3.14}) does not depend on the choice of 
this class of $G$ we can integrate both numerator and denominator 
against $\int\; [dw]\;exp(-<w,\kappa\cdot w>)$ with $G$ replaced by 
$G-w$ where 
$\kappa:\mathfrak{g}\to \mathfrak{g}$ is an integral kernel. Finally 
introducing ghost and anti-ghost $\eta,\rho$ Berezin integrals 
one finds 
\be \label{3.15}
\frac{I_F}{I_{F_0}}=
\frac{
\int\;[dq]\;[d\eta]\;[d\rho]\; F(q)\; e^{-<G(q),\kappa G(q)>}\;
\exp(<\eta,[\frac{\delta}{\delta v} G(\alpha_{\exp(v)}(q))]_{v=0}\cdot\rho>)
}
{
\int\;[dq]\;[d\eta]\;[d\rho]\; F_0(q)\; e^{-<G(q),\kappa G(q)>}\;
\exp(<\eta,[\frac{\delta}{\delta v} G(\alpha_{\exp(v)}(q))]_{v=0}\cdot\rho>)
}
\ee
The application of (\ref{3.15}) consists in picking a gauge invariant 
action $S_s(q)$ and $F_0=e^{-i^s S_s}$ while $F=e^{-i^s S_s}\; O$ where 
$O$ is any gauge invariant functional that one may obtain by 
functionally differentiating $e^{i^s <f,q>}$ at $f=0$ which thus 
gives rise to the generating functional (rescaling $\kappa$ and the ghost 
term by 
$e^{is}$)
\be \label{3.16}
Z_s(f)= 
\int\;[dq]\;[d\eta]\;[d\rho]\; e^{-i^s S_s(q)}\; e^{-e^{is} 
<G(q),\kappa G(q)>}\;
\exp(e^{is}\;<\eta,
[\frac{\delta}{\delta v} G(\alpha_{\exp(v)}(q))]_{v=0}\cdot\rho>)
\; e^{i^s <f\cdot q>}
\ee
for the case of a gauge theory. From here on the definition 
of $C_s(f), \Gamma_s(\hat{q})$ is identical to the case without gauge 
invariance.

\subsection{Background field method with gauge invariance}
\label{s3.4}

The background dependent generating functional 
in the case with gauge invariance is defined by
\be \label{3.17}
\bar{Z}_s(f,\bar{q}) =  
\int\;[dh]\;[d\eta]\;[d\rho]\; e^{-i^s S_s(\bar{q}+h)}\; 
e^{-e^{is} <\bar{G}_{\bar{q}}(h),\kappa \bar{G}_{\bar{q}}(h)>}\;
\exp(e^{is} <\eta,[\frac{\delta}{\delta v} 
\bar{G}_{\bar{q}}(\alpha^{\bar{q}}_{\exp(v)}(h))]_{v=0}\cdot\rho>)
\; e^{i^s <f\cdot h>}
\ee
while $\bar{C}_s(f,\bar{q}),\;\bar{\Gamma}(\hat{q},\bar{q})$ are defined 
as in the case without gauge invariance. Here $\bar{G}_{\bar{q}}(h)$ is a 
background dependent gauge fixing condition which is as admissable as any 
other one subject to the conditions spelled out in the previous subsection
and 
\be \label{3.18}
\alpha^{\bar{q}}_{\exp(u)}(h):=\alpha_{\exp(u)}(\bar{q}+h)-\bar{q}
\ee
If we now introduce $q=\bar{q}+h$ as a new integration variable 
we find again 
\be \label{3.19}
\bar{Z}_s(f,\bar{q})=Z_s(f)\; e^{-i^s <f,\bar{q}>}
\ee
provided that we choose in $Z_s$ the gauge fixing condition
\be \label{3.20}
G(q):=\bar{G}_{\bar{q}}(q-\bar{q})
\ee
In particular we have again (\ref{3.5}). Thus although we may use a 
background heavily in order to define $\bar{Z}_s(f,\bar{q})$, we can 
use the corresponding $\bar{\Gamma}_s(\hat{q},\bar{q})$ to define
the background independent $\Gamma_s(q)$.

The real power of the background method in usual QFT arises due to the 
existence of distinguished, so-called ``background gauges''. These are 
particular gauge conditions with the property that $Z_s(f,\bar{q})$ 
is in fact invariant under simultaneous  
gauge transformations of both $f$ and $\bar{q}$ for suitable choice
of kernel $\kappa$ (see \cite{12}). Here a gauge transformation of 
$\bar{q}$ is defined to be the same as for $q$ while the gauge transformation 
of $f$ is defined by duality 
\be \label{3.21}
<f,\alpha_g(q)>=:<\alpha^\ast_{g^{-1}}(f),q>
\ee
In case of GR these background gauge transformations are just the ordinary 
transformations of $f,\bar{q}$ viewed as ordinary tensor fields and 
thus suitable background gauges have to be diffeomorphism covariant 
in the usual sense, e.g. the deDonder gauge would be such a gauge.  
This has the advantage that the 
terms that appear in the effective action are restricted to be gauge 
invariant functionals of $\hat{q}$ under the usual action $\alpha_g$ 
which implies a huge reduction in the number of all possible terms.     

\subsection{Effective average action}
\label{s3.5}

We now modify the generating functional $\bar{Z}_s(f,\bar{q})$ by 
\ba \label{3.22}
&& \bar{\Gamma}_{s,k}(f,\bar{q})
:=\int\; 
\int\;[dh]\;[d\eta]\;[d\rho]\; e^{-i^s S_s(\bar{q}+h)}\; 
e^{-e^{is} <\bar{G}_{\bar{q}}(h),\kappa \bar{G}_{\bar{q}}(h)>}\;
\exp(e^{is} <\eta,[\frac{\delta}{\delta v} 
\bar{G}_{\bar{q}}(\alpha^{\bar{q}}_{\exp(v)}(h))]_{v=0}\cdot\rho>)
\times \nonumber\\
&&
\; e^{i^s <f\cdot h>}\; e^{-\frac{i^s}{2}<h,R_k(\bar{q}) h>}
\ea
If we are given a system without gauge invariance, the terms involving $G$ 
and the ghost integrals have to be dropped.

The one parameter family $k\mapsto R_k(\bar{q})$ 
of possibly background dependent and coupling dependent (coupling 
dependence increases the non-linearity in the flow of couplings as derived 
below)
integral kernels is supposed to 
obey a number of properties. Consider harmonic analysis on the 
spacetime determined by 1. the manifold on which the given QFT is defined and 
2. the background metric determined by $\bar{q}$ 
as well as the Fourier modes $\hat{R}_k(p)$ of $\hat{R}_k(\bar{q})$ 
where $p$ 
are the mode labels with respect to a generalised eigenbasis of the background 
d'Alembertian ($s=1$) or Laplacian ($s=0$) respectively. Then:\\
1. $\lim_{k\to 0} \hat{R}_k(p)=0$ at fixed $p$. 
This property will ensure that the $k-$dependent 
effective average action that we obtain from it as below reduces 
to the background dependent effective action which in turn determines
the actual object of interest, the background independent effective action.\\
2. $\lim_{k\to \infty} \hat{R}_k(p)=\infty$ at fixed $p$. This 
property makes the additional term suppress non-zero values of $h$ either 
due to oscillations ($s=1$) or due to damping so that the effective 
average action becomes almost exactly equal to $S_s$ in the saddle point 
approximation.\\
3. 
$\lim_{p\to 0} \hat{R}_k(p)=const.$ at fixed $k$. This property ensures 
that the additional term acts like an additional mass term for $p\le k$ which 
acts suppressingly either due to oscillations or damping.\\
4.
$\lim_{p\to \infty} \hat{R}_k(p)=0$ at fixed $k$. This property ensures that 
the momenta $p\ge k$ are unsuppressed by the additional term and thus are 
integrated over in the usual way. \\
\\
Altogether $\hat{R}_k(p)\approx \theta(k-p) \hat{R}(k)$ with $\hat{R}(0)=0$, 
monotonously 
increasing as $k\to\infty$ where it diverges. As $R_k(\bar{q})$ 
appears in the Wetterich equation (see below) in the form 
$\partial_k R_k(\bar{q})$ inside a trace (integral over $p$), that trace 
receives 
essential contributions only from a neighbourhood of $k$ which improves
convergence of the trace. Following the intuition that $R_k(\bar{q})$ 
suppresses integrating out modes $p\le k$, it follows that 
$\bar{Z}_{s,k}(f,\bar{q})$ has the more momenta integrated out the lower 
$k$, i.e. lowering $k$ is similar to, but different from, 
coarse graining in a strict Wilsonian renormalisation group sense.\\
\\
We define now in almost the usual pattern the k-dependent functionals
\ba \label{3.23}
\bar{C}_{s,k}(f,\bar{q})&=& i^{-s}\;\ln(\bar{Z}_{s,k}(f,\bar{q}))
\nonumber\\
\bar{\Gamma}_{s,k}(\hat{q},,\bar{q}) &=& {\rm extr}_f\;[-\bar{C}_{s,k}(f,\bar{q})
+<f,\hat{q}>]-\frac{1}{2}<\hat{q},R_k(\bar{q})\;\hat{q}>
\ea
Note the additional $R_k(\bar{q})$ dependent 
term that we subtract in the second line that is non-standard as compared to 
the previous sections. Again, due to $R_0(\bar{q})=0$ that additional 
term does not spoil that 
\be \label{3.24}
\bar{\Gamma}_{s,k=0}(\hat{q},\bar{q})=
\bar{\Gamma}_s(\hat{q},\bar{q})
\ee

\subsection{Wetterich equation}
\label{s3.6}

To show that the Wetterich equation really is just a consequence 
of the definition of the EAA, we give here an elementary 
derivation for completeness. We will suppress
the argument $\bar{q}$ of the functions that depend on it in order not to 
clutter the formulae and because it plays no role in the following 
analysis.\\
\\
By definition of the Legendre transform, the extremal condition
\be \label{3.25}
\frac{\delta\bar{C}_{s,k}(f)}{\delta f}-\hat{q}=0
\ee
must be solved for $f=f_{s,k}(\hat{q})$. We assume that the solution is 
unique and thus find explicitly
\be \label{3.26}
\bar{\Gamma}_{s,k}(\hat{q})=
[-\bar{C}_{s,k}(f)+<f,\hat{q}>]_{f=f_{s,k}(\hat{q})}
-\frac{1}{2}<\hat{q},R_k\hat{q}>
\ee
Since (\ref{3.25}) becomes an identity when evaluated at $f=f_{s,k}(\hat{q})$
\be \label{3.27}
\hat{q}\equiv [\frac{\delta\bar{C}_{s,k}(f)}{\delta f}]_{f=f_{s,k}(\hat{q})}
\ee
it follows
\be \label{3.28}
\frac{\delta\bar{\Gamma}_{s,k}(\hat{q})}{\delta \hat{q}}=
f_{s,k}(\hat{q})-R_k\cdot \hat{q}
\ee
The functional derivative of (\ref{3.27}) with respect to $\hat{q}$ yields
($1$ denotes the identity integral kernel on the space $\mathfrak{Q}$)
\ba \label{3.29}
1 &=&
<[\frac{\delta^2 \bar{C}_{s,k}(f)}{\delta f\; \delta f}]_{f=f_{s,k}(\hat{q})},
\frac{\delta f_{s,k}(\hat{q})}{\delta \hat{q}}>
\nonumber\\
&=&
<[\frac{\delta^2 \bar{C}_{s,k}(f)}{\delta f\; \delta f}]_{f=f_{s,k}(\hat{q})},
[\frac{\delta^2 \bar{\Gamma}_{s,k}(\hat{q})}
{\delta \hat{q}\otimes\delta\hat{q}}+R_k]>
\ea
where in the second step we inserted the functional derivative of (\ref{3.28}).

Next we write 
\ba \label{3.30}
&& \bar{Z}_{s,k}(f)=\int\; 
d\mu[h]\;e^{i^s\;<f,h>}\;e^{-\frac{i^s}{2}<h,R_k h>},\;
\frac{d\mu[h]}{d[h]}=   
\int\;[d\eta]\;[d\rho]\; e^{-i^s S_s(\bar{q}+h)}\; 
e^{-e^{is} <\bar{G}_{\bar{q}}(h),\kappa \bar{G}_{\bar{q}}(h)>}\;
\times\nonumber\\
&& \exp(e^{is} <\eta,[\frac{\delta}{\delta v} 
\bar{G}_{\bar{q}}(\alpha^{\bar{q}}_{\exp(v)}(h))]_{v=0}\cdot\rho>)
\ea
and find 
\ba \label{3.31}
\frac{\delta \bar{C}_{s,k}(f)}{\delta f} &=& 
\frac{1}{\bar{Z}_{s,k}(f)}\; \int\; d\mu[h]\; h\; e^{i^s <f,h>}\; 
e^{-\frac{i^s}{2}<h,R_k h>}=:\hat{q}_{s,k}(f)
\nonumber\\
\frac{\delta^2 \bar{C}_{s,k}(f)}{\delta f\otimes \delta f} &=&
i^s\;[ 
\frac{1}{\bar{Z}_{s,k}(f)}\; \int\; d\mu[h]\; h\otimes h\;  e^{i^s <f,h>}\; 
e^{-\frac{i^s}{2}<h,R_k h>}-\hat{q}_{s,k}(f)\otimes \hat{q}_{s,k}(f)]
\ea
and 
\ba \label{3.32}
&& \frac{\partial \bar{\Gamma}_{s,k}(\hat{q})}{\partial k}
=-\frac{\partial \bar{C}_{s,k}(f)}{\partial k}]_{f=f_{s,k}(\hat{q})}
-\frac{1}{2}<\hat{q},[\partial_k R_k]\cdot \hat{q}>
\nonumber\\
&=& \frac{1}{2}[-<\hat{q},[\partial_k R_k]\hat{q}>
+(\frac{i^s}{\bar{Z}_{s,k}(f)}\int\;d\mu[h]\;<h,[\partial_k R_k] h> \;
e^{i^s <f,h>}\; 
e^{-\frac{i^s}{2}<h,R_k h>})_{f=f_{s,k}(\hat{q}}]
\ea
We see that in (\ref{3.31}) the combination $h\otimes h$ appears at general 
$f$ while 
in (\ref{3.32}) the combination 
\be \label{3.33}
<h,[\partial_k R_k] h>={\rm Tr}([\partial_k R_k]\;\cdot \; h\otimes h)
\ee
appears at $f=f_{s,k}(\hat{q}))$. We can therefore relate (\ref{3.32}) 
to (\ref{3.32}) if we trace (\ref{3.32}) with $\partial_k R_k$ at 
$f=f_{s,k}(\hat{q})$. Note the identity 
\be \label{3.34}
[\hat{q}_{s,k}(f)]_{f=f_{s,k}(\hat{q})}=\hat{q}   
\ee      
Then 
\ba \label{3.35}
&& 
{\rm Tr}(\frac{\delta^2 \bar{C}_{s,k}(f)}{\delta f \otimes \delta f}
[\partial_k R_k])_{f=f_{s,k}(\hat{q})}
\\
&=& 
i^s
[(\frac{1}{\bar{Z}_{s,k}(f)}\; \int\; d\mu[h]\; <h,[\partial_k R_k] h>\;  
e^{i^s <f,h>}\; 
e^{-\frac{i^s}{2}<h,R_k h>})_{f=f_{s,k}(\hat{q}}
-<\hat{q},[\partial_k R_k] \hat{q}>]
\nonumber\\
&=& 2\; i^s \partial_k \bar{\Gamma}_{s,k}(\hat{q})
\nonumber\\
&=&
{\rm Tr}(
[\frac{\delta^2\bar{\Gamma}_{s,k}(\hat{q})}
{\delta\hat{q}\otimes \delta\hat{q}}+R_k]^{-1}\; [\partial_k R_k])
\ea
where in the last step we solved (\ref{3.29}) for 
$[\frac{\delta^2 \bar{C}_{s,k}(f)}{\delta f \otimes \delta f}
]_{f=f_{s,k}(\hat{q})}$. This is the celebrated {\it Wetterich equation}
\cite{29} for both signatures (where we display the dependence 
on the background again)
\be \label{3.36}
\partial_k \bar{\Gamma}_{s,k}(\hat{q},\bar{q})
=\frac{i^{-s}}{2}{\rm Tr}\{
[\bar{\Gamma}^{(2)}_{s,k}(\hat{q},\bar{q}) 
+R_k(\bar{q})]^{-1}\;[\partial_k R_k](\bar{q})\}
\ee
This is an exact and non-perturbative identity which has to be obeyed by the 
EAA if it comes from a well defined $\bar{Z}_{s,k}(f,\bar{q})$. The 
idea is now to consider (\ref{3.26}) as the fundamental equation, to solve 
it and try to find interesting solutions. By a solution we mean
a solution of (\ref{3.36}) which exists {\it for all values of $k$},
in particular for $k\to \infty$ and for $k\to 0$, in the sense described 
in the next subsection, so that our 
object of interest
$\lim_{k\to 0} \bar{\Gamma}_{s,k}(\hat{q},\bar{q})$ exists. 
The requirement that the solution has a suitable limit as $k\to \infty$ is 
the selection criterion of {\it asymptotic safety} for a viable and 
predictive theory because it typically dictates that all but a finite 
number of coulings have to be fine tuned to their fixed point values 
while the remaining couplings run with $k$ and approach finite values at 
$k=0$ in their dimensionful form. We stress that from the point of 
view of CQG the value $k=0$ is the only one for which the EAA has a 
clear physical meaning, namely it becomes the standard effective action,
because $k$ has no physical interpretation.\\
\\
Before closing this subsection we note that in the above derivation 
the ghost action was considered to be integrated out and thus gives 
a ghost determinant $\det(M)=\exp(\ln(\det(M))=:\exp(-i^s S^g)$. 
The ghost action $S^g$ and the gauge fixing action 
$S^{gf}=<\bar{G}_{\bar{q}},\kappa\cdot \bar{G}_{\bar{q}}>$ are thereby 
considered as corrections of the classical action. Therefore all three
parts flow according to the above Wetterich equation. It is sometimes
convenient to treat the ghost part individually by not integrating 
it out thereby avoiding the exponentiated 
logarithms. Thus one generalises the 
EAA by introducing also anticommuting test fields $e_j,r^j$
for the ghosts and a similar $R_k$ term bilinear in $\eta,\rho$.
Then we introduce into the integrand of $Z_{s,k}(f,\bar{q})$ the additional 
factor
\be \label{3.37}
\exp(i^s \; [<e,\eta>+<r,\rho>])\; 
\exp(-i^s<\eta,R^g_k(\bar{q})\rho>)
\ee
where $R^g_k(\bar{q})$ is a ghost kernel with the same properties 
as $R_k(\bar{q})$. Then $\bar{Z}_{s,k}(f,\bar{q})$ becomes  
$\bar{Z}_{s,k}(f,e,r,\bar{q})$ and 
$\bar{C}_{s,k}(f,e,r,\bar{\rho})=i^{-s} \ln(\bar{Z}_{s,k}(f,e,r,\bar{\rho})$.
Then 
\be \label{3.38}
\bar{\Gamma}_{s,k}(\hat{q},\hat{\eta},\hat{\rho},\bar{q})
:={\rm extr}_{f,e,r}[<f,\hat{q}>+<e,\hat{\eta}>+<r,\hat{\rho}>
-\bar{C}_{s,k}(f,e,r,\bar{q})]-[
\frac{1}{2}<\hat{q},R_k(\bar{q})\hat{q}>
+<\hat{\eta},R^g_k(\bar{q})\hat{\rho}>]
\ee
Then a repetition of the above chain of steps gives the generalised 
Wetterich equation (we collect $\hat{\phi}:=(\hat{q},\hat{c}), 
\hat{c}=\hat{\eta},\hat{\rho})$)
and suppress the dependence on $\bar{q}$ for notational simplicity) 
\be \label{3.39}
\partial_k \bar{\Gamma}_{s,k}(\hat{\phi})
=\frac{i^{-s}}{2}{\rm Tr}\{
[\frac{\delta^2 \bar{\Gamma}^{(2)}_{s,k}(\hat{\phi})}{[\delta\hat{q}]^2}
+R_k]^{-1}\;[\partial_k R_k]\}
-i^{-s}{\rm Tr}\{
[\frac{\delta^2 \bar{\Gamma}^{(2)}_{s,k}(\hat{\phi})}{[\delta\hat{c}]^2}
+R^g_k]^{-1}\;[\partial_k R^g_k]\}
\ee
where the relative factor of $-2$ between bosonic and ferminic term comes 
from the anticommuativity of Grassman valued derivatives and the independence 
of $\hat{\rho},\hat{\eta}$. 

In view of lemma \ref{la4.1} one must however
pay attention to the following:
It is not true that one recovers the original EAA (ghosts integrated out)
by simply setting
$\hat{c}=0$. Rather we must apply 
lemma \ref{la4.1} of the appendix to find the exact relation
between (\ref{3.39}) and (\ref{3.36}).

\subsection{Solving the Wetterich equation in principle and asymptotic 
safety prescription}
\label{s3.7}

The EAA is supposed to yield well defined tempered distributions when 
taking functional derivatives of any order and thus can be written as 
\be \label{3.40}
\bar{\Gamma}_{s,k}(\hat{q},\bar{q})
=\sum_{N=0}^\infty \; \bar{\Gamma}^N_{s,k}(\hat{\phi},\bar{q})
\ee
where  
$\bar{\Gamma}^N_{s,k}(\hat{q},\bar{q})$ is a monomial in $\hat{\phi}$ of 
order $N$. One can use a geometric series expansion 
\be \label{3.41}
[\bar{\Gamma}^{(2)}_{s,k}+R_k]^{-1}=R_k^{-1} \sum_{M=}^\infty (-1)^M 
[\bar{\Gamma}^{(2)}_{s,k} R_k^{-1}]
\ee
and derive an infinite recursion among the $\bar{\Gamma}^N_{s,k}$.
In practice one often combines $R_k$ with a piece of 
$\bar{\Gamma}^{(2)}_{s,k}(\hat{q}=0,\bar{q})$ in the geometric series 
expansion, see below. 
To make this practically useful one will try to find a basis of distributions
such that the recursion closes. After such a basis $T_\alpha(\bar{q})$ of 
distributions has been identified,
the coefficients $g_\alpha(k)$ of that basis have to be determined from 
the Wetterich equation where $\alpha$ runs through a countably infinite 
index set $\cal A$. These coefficients are in general dimensionful of 
mass dimension $d_\alpha$ 
and it is convenient to pass to dimensionfree objects $\hat{g}_\alpha(k)
=k^{-d_\alpha} g_\alpha(k)$. Setting $t=\ln(k/k_0)$ for some arbitrary
$k_0$ and collecting $\hat{g}(t)=\{g_\alpha(k_0 e^t)\}_{\alpha \in {\cal A}}$
one obtains an infinite first order ODE system
\be \label{3.41a}  
\frac{d}{dt}\;\hat{g}(t)=\beta(\hat{g}(t))
\ee
which is autonomous, i.e. $\beta$ does not depend explicitly on $t$
if $k$ is the only scale in the problem, but it is  
highly non-linear. By the general theory of finite dimensional 
ODE systems, maximal solutions exist and are unique given initial 
conditions at $k_0$ or equivalently integration constants $c$.
We assume that we can also find here maximal solutions whose allowed 
domain $D_c$ with respect to $t$ depends on the choice $c$. Let then 
$\hat{g}_c(t),\; t\in D_c$ be the unique maximal solution with initial data  
$c$ at $t=1$ or $k=k_0$. 

Such a solution of the flow equation (\ref{3.41}) is an integral 
curve of the vector field $\beta$. Let $\hat{g}_\ast$ be a zero of 
of the vector field $\beta$, if it exists, $\beta(\hat{g}_\ast)=0$.    
Now consider a general solution $\hat{g}_c(t)$. If we have to fine tune 
all but a finite number of the initial data $c$ in oder that 
$\lim_{t\to \infty} \hat{g}_c(t)=\hat{g}_\ast$ then the zero is called 
{\it predictive} for it says that we have only a finite dimensional space of 
admissable solutions (labelled by the non fined tuned, so called relevant, 
initial data) whose domain includes $(1,+\infty)$ and then 
$\lim_{t\to \infty} \hat{g}_c(t)=\hat{g}_\ast$ i.e. on that 
space of solutions $\hat{g}_\ast$ is an UV fixed point 
of the flow. We may write 
$c=(c_r,c_i^\ast)$ where $c_i^\ast$ is a point in an infinite dimensional 
space corresponding to the fine tuned parameters while $c_r$ may still 
range in a finite dimensional space. The value $c_i^\ast$ is 
considered a prediction of the theory given the zero $\hat{g}_\ast$. 

We now return to the corresponding
dimensionful couplings $g_{c_r,c_i^\ast}(k)$ which 
have a physical meaning at $k=0$ and can be measured there. Only a finite 
number of them need to be measured to determine the remaining initial 
data $c_r$, after that the theory is completely fixed. We call the 
zero $\hat{g}_\ast$ {\it complete} if the $g_{c_r,c_i^\ast}(0)$ 
are also finite. Note that this is not automatically the case.

To 
have an example in mind for such a scenario suppose that we have a 
basis consisting of only two distributions which come with dimensionful
constants $g_1(k),g_2(k)$ of dimension cm$^2$ and cm$^{-2}$ respectively.
The dimensionless objects are $\hat{g}_1(k)= k^2 g_1(k),\; 
\hat{g}_2(k)=k^{-2} g_2(k)$. Suppose that 
$\hat{g}_1(k)=\frac{k^2}{k^2+L_1^{-2}} \hat{g}_1^\ast,\;  
\hat{g}_2(k)=\frac{k^2+L_2^{-2}}{k^2} \hat{g}_2^\ast$
with integration constants $L_1,L_2$ of dimension of 
length and some numerical values 
$\hat{g}_1^\ast, \hat{g}_2^\ast$. Then both couplings are relevant, both 
run into the zero of the vector field as $k\to\infty$ where 
they reach their fixed point values $\hat{g}_1^\ast, \hat{g}_2^\ast$. 
On the other hand 
$g_1(k)=\frac{\hat{g}_1^\ast}{k^2+L_1^{-2}},\; 
g_2(k)=\hat{g}_2^\ast\;(k^2+L_2^{-2})$.    
Thus $g_1(0)=L_1^2 \hat{g}_1^\ast,\;g_2(0)=L_2^{-2} \hat{g}_2^\ast$ are 
supposed to equal measured values $g_{10}, g_{20}$ which thus fixes 
$L_1,L_2$. 
We see that it is important to require the finiteness of the 
{\it dimensionless} couplings at $k\to\infty$ and of the 
{\it dimensionful} couplings at $k\to 0$: If the couplings run, i.e. 
have non-trivial $k$ dependence then 
typically the running will be monotonous and thus to have finiteness  
of say the dimensionful coupling at both $k=0,\infty$ would be impossible. 

The application to GR in $n$ dimensions is as follows:
We expect the EAA to involve curvature polynomials of order $N=0,1,2,..$ 
such as $g_N(k) R^N$ where $R$ is the Ricci scalar and $g_N(k)$ is its 
dimensionful coupling. As $d^n x\; g_N\; R^N$ is dimensionless, 
$g_N(k)$ has dimension cm$^{2N-n}$, hence $\hat{g}_N(k)=k^{2N-n}\; g_N(k)$.
Accordingly a running with $k$ of the form 
$g_N(k)=\frac{L_N^{2N-n}}{1+(L_Nk)^{2N-n}}\;\hat{g}^\ast_N$ would be 
compatible with both requirements where $L_N$ is a constant of length 
dimension. Note that we do not require that all but a finite number of 
dimensionful couplings vanish, rather that they depend only on a finite 
number of such free length parametrs $L_N$.

\subsection{Solving the Wetterich equation in practice}
\label{s3.8} 

It is of course out of question to complete the programme of the previous 
subsection for all couplings. One has to resort to the so-called 
truncations in the space of couplings. It would be nice to have some 
kind of error control at one's disposal which would grant that given 
some error $\epsilon$ and a fixed point discovered by a certain truncation
space, all extensions of that truncation space exceeding some 
dimension depending on $\epsilon$ lead to corrections of the couplings 
of the given truncation of less than $\epsilon$. It is of course very hard 
to come up with such an error control.

The still exact and non-polynomial Wetterich equation 
(we just display the non-ghost version)
\be \label{3.42}
\partial_k\bar{\Gamma}_{s,k}(\hat{q},\bar{q})=\frac{i^{-s}}{2}\;
{\rm Tr}\{[\bar{\Gamma}^{(2)}_{s,k}(\hat{q},\bar{q})
+R_k(\bar{q}]^{-1}\;[\partial_k R_k(\bar{q})]\}
\ee
is equivalent to an infinite hierarchy of non-polynomial equations among the 
1-PI distributions 
\be \label{3.43}
T^n_{s,k}(\bar{q}):=
\{\frac{\delta^n \bar{\Gamma}_{s,k}(\hat{q},\bar{q})}{\otimes^n \delta\hat{q}}  
\}_{\hat{q}=0}
\ee
with $n=0,1,2..$. 
A compact way to write the $n-$th such equation is 
\be \label{3.44}
\frac{\delta^n}{\otimes^n \delta\hat{q}}\;
\{\partial_k \bar{\Gamma}_{s,k}(\hat{q},\bar{q})
-\frac{i^{-s}}{2}\;
{\rm Tr}([\bar{\Gamma}_{s,k}(\hat{q},\bar{q})+R_k]^{-1}\;
[\partial_k R_k(\bar{q})])
\}_{\hat{q}=0}=0
\ee
which no longer depends on $\hat{q}$. For $n=0$ we simply get 
\be \label{3.45}
\partial_k\;T^0_{s,k}(\bar{q})=\frac{i^{-s}}{2}\;
{\rm Tr}([T^2_{s,k}(\bar{q})+R_k(\bar{q}]^{-1}\;[\partial R_k(\bar{q})])
\ee
which is the equation most studied in the literature. For $n=1$ we get 
\be \label{3.46}
\partial_k\; T^1_{s,k}(\bar{q})(.)=-\frac{i^{-s}}{2}\;
{\rm Tr}(
[T^2_{s,k}(\bar{q})+R_k(\bar{q})]^{-1}\;
T^3_{s,k}(\bar{q})(.)\; [T^2_{s,k}(\bar{q})+R_k(\bar{q})]^{-1}\;
\;[\partial_k R_k(\bar{q})])
\ee
etc. where $(.)$ denotes the point at which one performs the 
functional derivatives. The {\it first truncation parameter} is therefore 
the number $N$ of those first equations that one considers, $n=0,1,..,N-1$.

Each of these equations can be expanded into a geometric series by 
splitting 
\be \label{3.47}
T^2_{s,k}(\bar{q})=D^2_{s,k}(\bar{q})+N^2_{s,k}(\bar{q})
\ee
where $D^2_{s,k}(\bar{q})$ (``diagonal part'') can be easily treated by the 
heat kernel methods sketched below while $N^2_{s,k}(\bar{q})$ 
(``off-diagonal part'') requires more work, and then  expanding 
around $R'_{s,k}(\bar{q})=R_k(\bar{q})+D^2_{s,k}(\bar{q})$
\be \label{3.48}
[T^2_{s,k}(\bar{q})+R_k(\bar{q})]^{-1}
=\lim_{M\to \infty} R'_{s,k}(\bar{q})^{-1}\sum_{m=0}^M\;(-1)^m\;
[N^2_{s,k}(\bar{q})\; R'_{s,k}(\bar{q})^{-1}]^m 
\ee 
A {\it second 
truncation parameter} consists in the number $M$ of terms $m=0,1,..,M-1$
that one keeps in that
expansion (for the $n-$th equations there are $n+1$ such series and for 
each of them we keep all $M$ terms).  

Having decided on the numbers $N,M$ one now makes an Ansatz for
$\bar{\Gamma}_{s,k}(\hat{q},\bar{q})$. The Ansatz is usually motivated 
by the action functional $S_s(q)$ that one started from. This is justified 
in the following sense:\\
Recall that we have 
\be \label{3.49}
\bar{Z}_s(f,\bar{q})=\int\;[dh]\; 
J[\bar{q}+h]\; e^{i^s[<f,h>-S^t_s(\bar{q}+h)]},\;\;
\bar{C}_s(f,\bar{q})=i^{-s} \ln(\bar{Z}_s(f,\bar{q})),\; 
\bar{\Gamma}_s(\hat{q},\bar{q})
={\rm extr}_f\;[<f,\hat{q}>-\bar{C}_s(f,\bar{q})]
\ee
where we have formally integrated out the ghosts giving rise 
to the measure factor $J=\det(M)$, i.e. the functional
determinant of the ghost matrix $M$ and we have combined the 
classical action $S_s$ and the gauge fixing action $S^{gf}$ 
into $S^t_s=S_s+S^{gf}$. 
   
Let $F=F(\hat{q},\bar{q})$ be the solution of 
$\bar{C}^{(1)}_s(f,\bar{q})=\hat{q}$ where 
we denote 
the $n-th$ order functional derivatives of the functional $F$ with respect to 
its argument by $F^{(n)}$. Then $\Gamma^{(1)}_s=F$ and suppressing the 
dependence on $\bar{q}$ 
\ba \label{3.50}
e^{-i^s \bar{\Gamma}_s(\hat{q})}
&=& \exp(i^s[\bar{C}_s(F)-<F,\hat{q})>])
=\bar{Z}_s(F) e^{-i^s <F,\hat{q}>i}
=\int\;[dh]\; J[\bar{q}+h]\;e^{i^s[<F,h-\hat{q}>-S^t_s(h+\bar{q})]}
\nonumber\\
&=&
\int\;[dh]\; J[\bar{q}+h]\;
e^{i^s[<\bar{\Gamma}^{(1)}(\hat{q}),h-\hat{q}>-S^t_s(\bar{q}+h)]}
\ea
which is a functional integro differential equation for 
$\bar{\Gamma}_s(\hat{q})$.
We Taylor expand 
\be \label{3.51}
S^t_s(\bar{q}+h)=\sum_{n=0}^\infty\; \frac{1}{n!}\; 
<S_s^{t(n)}(\bar{q}+\hat{q}),\otimes^n [h-\hat{q}]>
\ee
and similar for $J$
and shift the integration variable to obtain to second order
\ba \label{3.52}
e^{-i^s[\bar{\Gamma}_s(\hat{q},\bar{q})-S^t_s(\bar{q}+\hat{q})]}
&=& \int\; [dh]\; e^{-i^s[\frac{1}{2}<S_s^{t(2)}(\bar{q}+\hat{q}),\otimes^2 h>
+<(S_s^{t(1)}(\bar{q}+\hat{q})-\bar{\Gamma}_s^{(1)}(\hat{q},\bar{q})),h>]}
\times \nonumber\\
&& \{J[\bar{q}+\hat{q}]
+<J^{(1)}(\bar{q}+\hat{q}),h>
+\frac{1}{2}<J^{(2)}(\bar{q}+\hat{q}),h\otimes h>\}
\ea
The Gaussian integrals can be performed. It yields 
an iteration for $\bar{\Gamma}_s-S_s$. If one truncates 
the iteration at the lowest order (1-loop approximation) one 
obtains 
\be \label{3.53}
\bar{\Gamma}_s(\hat{q},\bar{q})=S^t_s(\bar{q}+\hat{q})+\frac{1}{2\;i^s}
{\rm Tr}(\ln(S^{t(2)}_s(\bar{q}+\hat{q})
-\frac{1}{i^s} {\rm Tr}(\ln(J(\bar{q}+\hat{q}))
\ee     
Thus a well motivated starting Ansatz for the EAA is always the 
classical action plus gauge fixing action 
plus the above additional logarithmic term 
with argument split into background $\bar{q}$ and argument $\hat{q}$. 
Note that for $s=1$ the EAA, like the usual effective action, 
is necessarily complex valued rather than real valued.

The ingredients $S_s(q),S_{gf}(q),{\sf Tr}(\ln(M))$ 
can be decomposed into several terms 
$g_\alpha\; S_{s,\alpha}(q)$ differing 
by the order of the derivatives of the field $q$ each of which comes with 
its own coupling constant $g_\alpha$. We make these couplings $k$ 
dependent by hand, $g_\alpha\to g_\alpha(k)$ and write 
$S_{s,\alpha}(\hat{q},\bar{q}):=S_{s,\alpha}(\bar{q}+\hat{q})$ by expansion 
where we use the covariant differential defined by $\bar{q}$. Accordingly 
we now step with 
\be \label{3.54}
T^n_{s,k}(\bar{q}):=\sum_{\alpha=1}^A\; g_\alpha(k)\; 
S^n_{s,\alpha}(\bar{q}),
\;\;
S^n_{s,\alpha}(\bar{q})
=[\frac{\delta^n S_{s,\alpha}(\hat{q},\bar{q})}
{\otimes^n \delta \hat{q}}]_{\hat{q}=0}  
\ee
into the above hierarchy. Thus a {\it third truncation parameter} consists 
in the number $A$ of terms in the above Ansatz (assuming they can be 
enumerated in some way). 

Next one must isolate on the r.h.s. of the Wetterich equation truncated at 
$N,M$ as detailed above the contributions $S^n_{s,\alpha}(\bar{q})$ and 
compare coefficients. For this isolation of terms we require the 
heat kernel techniques as detailed below. The reason for this is that the 
l.h.s., when written in terms of the covariant differential $\nabla$ compatible 
with $\bar{q}$, generates distributions which are of the form of polynomials  
of $\nabla$ applied to products of $\delta$ distributions which generates 
curvature polynomials. The heat kernel technique is an elegant method 
for rewriting also the r.h.s. in terms of such distributions.
    
Some terms that appear on 
the r.h.s. of the Wetterich equation are not of the form 
$S^n_{s,\alpha}(\bar{q}),\; \alpha\le A$ (they corrspond to some 
$\alpha>A$) and thus cannot be compared to the l.h.s. 
These are discarded by hand. One must tune $A$ to $N,M$ 
such that after discarding those unmatchable terms one gets as many 
independent equations as we have couplings, that is, $A$. After this,
one gets a closed, finite system of equations $\partial_k g(k)=\beta(g(k))$
which is then treated along the logic laid out in the previous subsection.
  
One can then improve the analysis and increase the numbers $N,M,A$. 
A systematic way to do this is concisely described in \cite{17,1}. 
The hope is of course that the fixed points and (finite!) 
number of relevant couplings 
stabilises beyond certain values of $N,A,M$.

\subsection{Comparing coefficients: Synge world function,  
Schr\"odinger kernel, heat kernel}
\label{s3.9}

The material contained in this subsection is mostly well known in the 
Euclidian setting, we include it here because it allows to pin point 
differences withe the Lorentzian version and to make this review section 
self contained.

\subsubsection{Heat kernel: definition}
\label{s3.9.1}

Let $(M,m)$ be an $n-$dimensional manifold of Lorentzian 
($s=1$) or Euclidian ($s=0$) signature, 
$\nabla$ the Levi Civita covariant differential 
determined by $m$ and $\Delta:=m^{\mu\nu}\nabla_\mu\nabla_\nu$ the 
Laplacian (for $s=1$ it is better called the d'Alembertian). The 
heat kernel of $\Delta$ (for $s=1$ it is better called the 
Schr\"odinger kernel) is the solution to the initial value problem
\be \label{3.55}
[i^{-s}\; \frac{d}{dt}-\Delta_x]\;H^s_t(x,y)=0,\;
H_0(x,y)=\delta(x,y)
\ee
In flat space $m_{\mu\nu}=\eta_{\mu\nu}$ (Minkowski metric $s=1$) or
$m_{\mu\nu}=\delta_{\mu\nu}$ (Euclidian metric $s=0$) this equation can 
be explicitly solved by 
Fourier transformation 
\be \label{3.56}
H_t(x,y)=H_t(x-y)=e^{i^s\; t \Delta_x} \;\delta(x,y)
=\int_{\mathbb{R^n}}\; \frac{d^n l}{[2\pi]^n}\;
e^{-i k_\mu (x-y)^\mu}\; e^{-i^s \; t\; m^{\mu\nu} k^\mu k_\nu}
\ee
Using the elementary integrals
\be \label{3.57}
\int_\mathbb{R}\; du\; e^{-a \;u^2}=\sqrt{\frac{\pi}{a}}\;\forall a>0;\;\;\; 
\int_\mathbb{R}\; du\; e^{i a \;u^2}=\sqrt{\frac{\pi}{|a|}}
e^{i\frac{\pi}{4}{\rm sgn}(a)}\;\forall a\in \mathbb{R}-\{0\}
\ee
one finds 
\be \label{3.58}
H^s_t(x,y)=[4\pi |t|]^{-n/2}\; e^{\frac{i^s}{4t} 
m_{\mu\nu}(x-y)^\mu (x-y)^\nu}\; e^{is\frac{\pi}{4}{\rm sgn}(t)(2-n)}
\ee
where for $s=0$ only $t\ge 0$ is allowed. For $s=1$ this is a 
one parameter group of unitary operators,
for $s=0$ a one parameter semi-group of contraction operators. 

Irrespective of signature one observes that the $x,y$ dependence is determined
by the square of the 
geodesic distance $m(x-y,x-y)$ between the points $x,y$. It is 
thus reasonable to expect that at least for sufficiently close points $x,y$
in a general spacetime the heat kernel is also largely 
determined by the square of 
the geodesic distance. 

\subsubsection{Synge world function}
\label{s3.9.2}

In flat spacetime 
the geodesic between any two points is uniquely given by the straight line 
through them. In curved spacetime the geodesics through two points are 
unique only within a convex normal neighbourhood. Consider then two such 
points $x,y$ in $(M,m)$ and let $\gamma_{x,y}(s)$ be the unique geodesic 
with starting and end point $\gamma_{x,y}(0)=x,\;\gamma_{x,y}(1)=y$ 
respectively (we may assume that the convex neighbourhood is contained in a 
chart, otherwise restrict the neighbourhood). 
We know that the geodesic is the solution of the Euler Lagrange 
equations following from variation of the action 
\be \label{3.59}
S[\gamma]:=\int_0^1\; ds\; w_\gamma(s),\; w_\gamma(s)=
\sqrt{|m_{\mu\nu}(\gamma(s)) 
\dot{\gamma}^\mu(s)\;\dot{\gamma}^\nu(s)|}
\ee
subject to the boundary condition $\gamma(0)=x,\; \gamma(1)=y$. When 
evaluated on the solution $\gamma_{x,y}$ (\ref{3.59}) becomes the 
Hamilton-Jacobi function $S(x,y)$ that equals the absolute 
value of the geodesic distance.      
Using the Euler Lagrange equations and the boundary conditions one finds 
\be \label{3.60}
\partial_{y^\mu} S(x,y)=\delta\;
\frac{m_{\mu\nu}(\gamma_{x,y}(1))\;\dot{\gamma}^\nu_{x,y}(1)}
{w_{\gamma_{x,y}}(1)} 
\ee
where $\delta=1$ for $s=0$ and $\delta=\pm 1$ if the geodesic is spacelike
or timelike respectively when $s=1$. When the geodesic 
is null then $S(x,y)=0$ of course and the r.h.s. of (\ref{3.60}) 
is to be replaced by zero. Thus one finds the Hamilton-Jacobi
equation (replace r.h.s. by zero in the null case).  
\be \label{3.61}
m^{\mu\nu}(y)\; S_{,y^\mu}(x,y) S_{,y^\nu}(x,y)=\delta
\ee
We now define the {\it Synge world function}
\be \label{3.62}
\sigma(x,y):=\delta\frac{1}{2} S(x,y)^2
\ee
as the {\it signed} half of the square of the geodesic distance. Note
that in flat space $\sigma=\delta |m(x-y,x-y)|/2=m(x-y,x-y)/2$ is the 
function on which heat kernel exponent depends. Then we obtain the 
{\it master equations} 
\be \label{3.63}
m^{\mu\nu}(y) 
[\nabla^y_\mu \sigma(x,y)]\;[\nabla^y_\nu \sigma(x,y)]
=2\sigma(x,y),\; \sigma(x,x)=0
\ee
Remarkably, this equation, and all identities that follow from it, are 
completely insensitive to the signature and in the Lorentzian 
signature case holds for all three types of geodesics. 
One shows that (\ref{3.63}) 
also holds if we replace $\nabla^y, m(y)$ by $\nabla^x, m(x)$. 

The Synge world function ist the simplest example of a bi-scalar. In general 
a bi-tensor $T(x,y)$ of type $(a,b),(a',b')$ 
transforms 
as a tensor of type $(a,b)$ with respect to diffeomorphisms 
of $x$ and 
as a tensor of type $(a',b')$ with respect to diffeomorphisms 
of $y$. The coincidence limity of a bi-tensor is denoted 
as $\bar{T}(x)=T(x,x)$ and is a tensor of type $(a+a',b+b')$.
In the Euclidian case the coincidence is unambiguously defined, for 
the Lorentzian case we take $y\to x$ along a spacelike path for definiteness. 
For our purposes it will be sufficient to consider bi-tensors of type 
$(a,b),(0,0)$ i.e. it transforms as a scalar w.r.t. $y$. 

We have seen above that $\nabla^x_\mu \sigma(x,y)=\propto S(x,y) 
\partial^x_\mu S(x,y)$ is either zero for a null geodesic or 
$\nabla^x_\mu S(x,y)$ is a (timelike or spacelike) unit vector. Therefore 
\be \label{3.64}
\overline{\nabla^x_\mu \sigma(x,.)}=0
\ee
We obtain complete information about the coincidence limits of 
all covariant derivatives of $\sigma$ 
\be \label{3.65}
s_{\mu_1..\mu_n}(x)
:=\overline{
\nabla^x_{\mu_1}..\nabla^x_{\mu_n}\sigma(x,.)}
\ee
as follows: We
compute succesively the covariant derivatives of $\sigma(x,y)$ of order 
$n=1,2,..$ 
w.r.t. $x$ and obtain relations from them using the master equations. 
The $n-th$ derivative relation
depends on the r.h.s. on covariant derivatives of $\sigma$
up to order  
$n+1$ where that $(n+1)$-th derivative is multiplied with a first derivative.
Assuming inductively that the coincidence limit of the $(n+1)$-th
derivative is regular one obtains an identity for the coincidence limits 
of the n-th derivative due to (\ref{3.64}).  

For instance (we drop the arguments $x,y$ from now on and understand that 
all derivatives are taken at $x$ at fixed $y\not=x$ 
and all contractions and curvature tensors 
are evaluated at $x$)
\ba \label{3.66}
\nabla_\mu \sigma &=& m^{\alpha\beta} [\nabla_\mu\nabla_\alpha\sigma]\; 
[\nabla_\beta \sigma]
\nonumber\\
\nabla_\mu \nabla_\nu \sigma &=& 
m^{\alpha\beta} \{
[\nabla_\mu \nabla_\nu\nabla_\alpha\sigma]\; 
[\nabla_\beta \sigma]
+[\nabla_\nu\nabla_\alpha\sigma]\; 
[\nabla_\mu\nabla_\beta \sigma]\}
\ea
Assuming $s_{\mu\nu\rho}$ to be regular we find from the second relation
that 
\be \label{3.67}
s_{\mu\nu}=m^{\alpha\beta} s_{\mu\alpha} s_{\nu\beta}
\ee
which means that $s_\mu\;^\nu$ is a projection. Since $\sigma$ for 
sufficiently close $x,y$ approaches the flat space expression in 
suitable coordinates, the projection
has maximal rank, therefore 
\be \label{3.68}
s_{\mu\nu}=m_{\mu\nu}   
\ee
Continuing like this one finds 
\be \label{3.68a}
s_{\mu\nu\rho}=0,\;
s_{\mu\nu\rho\sigma}=\frac{2}{3}\; R_{\mu(\rho\sigma)\nu}
\ee
etc.
Thus we can consider the coincidence limits (\ref{3.65}) as known 
tensors constructed from $m$. 

Next we consider an arbitary bi-tensor of type $(0,b),(0,0)$ (we can 
use $m(x)$ to reduce the case $(b-a,a),(0,0);\; b\ge a$ to this case). 
We are interested in the coincidence limit of of its covariant derivatives at
$x$. To that end we expand it with respect to its $y$ dependence which due 
to (\ref{3.64}) is equivalent to an expansion in $\nabla\sigma$ 
\be \label{3.69}
T_{\mu_1 .. \mu_b}(x,y)=\sum_{l=0}^\infty\; \frac{1}{l!}\;
[T^n]_{\mu_1..\mu_b}^{\nu_1..\nu_l}(x) \prod_{k=1}^l\; 
[\nabla^x_{\nu_k} \sigma(x,y)]
\ee
where the coefficients are automatically comepletely symmetric 
in $\nu_1..\nu_l$.
Then one finds using the above results for the coincidence limit of 
the $\sigma$ derivatives
\be \label{3.70}
\overline{T_{\mu_1 .. \mu_b}}
=[T^0]_{\mu_1 .. \mu_b},\;\;
\overline{\nabla_\nu T_{\mu_1 .. \mu_b}}
=\nabla_\nu [T^0]_{\mu_1 .. \mu_b}+[T^1]_{\mu_1..\mu_b}\;^\rho\; 
m_{\rho\nu}
\ee
etc., i.e. we can relate all expansion coefficients in (\ref{3.69}) to 
coincidence limits of covariant derivatives of $T$.
   
\subsubsection{Heat kernel: computation}
\label{s3.9.3}   

The Ansatz for the heat kernel in general $(M,m)$ is therefore well 
motivated to be 
\be \label{3.71} 
H^s_t(x,y)=[4\pi |t|]^{-n/2}\; e^{\frac{i^s}{2t} \sigma(x,y)}\; 
e^{is\frac{\pi}{4}{\rm sgn}(t)(2-n)}\; \Omega_t(x,y)
\ee
which contains a correction factor $\Omega_t(x,y)$ that captures the 
curvature effects. We know already 
\be \label{3.72}
\overline{\Omega_{t=0}}=1
\ee
as the other factor in (\ref{3.71}) 
reduces to $\delta(x,y)$ at $t=0$. Once we have 
determined $\Omega_t$ to sufficient accuracy we know the heat kernel 
to sufficient accuracy. To determine $\Omega_t$ we plug (\ref{3.71})
into the heat equation. We obtain for $t\not=0$ using (\ref{3.63})
(all derivatives and contractions at $x$)
\be \label{3.73}
\frac{[\Delta \sigma]-n}{2t}\Omega_t+\partial_t\Omega_t
+\frac{1}{t}m^{\mu\nu} [\nabla_\mu\Omega_t]\;[\nabla_\nu \sigma]
-i^s[\Delta \Omega_t]=0
\ee
We Taylor expand with respect to the heat time $t$ 
\be \label{3.74}
\Omega_t(x,y)=\sum_{k=0}^\infty\;
(i^s t)^k\; \Omega_k(x,y)
\ee
and obtain for $k=0,1,...$
\be \label{3.75}
(\frac{[\Delta \sigma]-n}{2}+k)\Omega_k
+m^{\mu\nu} [\nabla_\mu\Omega_k]\;[\nabla_\nu \sigma]
-[\Delta \Omega_{k-1}]=0
\ee
with the convention $\Omega_{-1}(x,y)=0$. Also $\Omega_{t=0}(x,x)=
\Omega_{k=0}(x,x)=1=\overline{\Omega_0}$. Equation (\ref{3.75}) is 
independent of signature and as all coefficients are real valued we can pick 
$\Omega_k$ real valued.  

We can explicitly determine the $\Omega_k$ using the Taylor expansion
(\ref{3.70}) together with (\ref{3.75}) 
\be \label{3.76}
\Omega_k(x,y)=\sum_l \frac{1}{l!}\; [\Omega_{k,l}]^{\mu_1..\mu_l}(x) 
\;[\nabla^x_{\mu_1}\sigma(x,y)]..[\nabla^x_{\mu_l} \sigma(x,y)]   
\ee
where $\Omega_{k,l}^{\mu_1..\mu_l}(x)$ is determined in terms of the 
coincidence limits of the $\nabla_{\mu_1}..\nabla_{\mu_r} \Omega_k,\;
0\le r\le l$ via (\ref{3.70}). To obtain the coincidence limits of 
the $\nabla^l \Omega_k$ 
we take the $\nabla^l$ 
derivative of (\ref{3.75}) and then take the coincidence limit. This 
involves only $\nabla^l \Omega_k$ and $\nabla^{l+2}\Omega_{k-1}$ because
the $\nabla^{l+1}\Omega_k$ term has a coefficient $\nabla\sigma$ which 
vanishes in the coincidence limit. Thus, as $\Omega_{-1}\equiv 0$ we 
can compute all $\Omega_{k,l}$ to any sesired order.  

We exhibit this for $l=0,1$ and general 
$k$: 
\be \label{3.77} 
[\Omega_{k,0}]=\overline{\Omega_k},\;
[\Omega_{k,1}]^\mu=\overline{\nabla^\mu \Omega_k}-\nabla_\mu
\overline{\Omega_k}
\ee
The coincidence limit of (\ref{3.75}) is 
\be \label{3.78}
k\overline{\Omega_k}-\overline{\Delta \Omega_{k-1}}=0
\ee
and the first derivative of (\ref{3.75}) is
\be \label{3.78a}
[k-\frac{n}{2}+\frac{1}{2}\Delta\sigma][\nabla_\mu \Omega_k]
+\frac{1}{2} [\nabla_\mu \Delta\sigma] \Omega_k
+[\nabla_\mu\nabla_\nu \Omega_k][\nabla^\nu\sigma]
+[\nabla_\nu \Omega_k][\nabla_\mu \nabla^\nu\sigma]
-\nabla_\mu\Delta \Omega_{k-1}=0
\ee
whose coincidence limit is 
\be \label{3.79}
(k+1)\overline{\nabla_\mu \Omega_k}-\overline{\nabla_\mu \Delta \Omega_{k-1}}
=0
\ee
For $k=0$ (\ref{3.78}) is identically satisfied while (\ref{3.79}) gives 
$\overline{\nabla_\mu \Omega_1}=0$. Thus by (\ref{3.77}) 
$[\Omega_{0,0}]=1, [\Omega_{0,1}]^\mu=0$ and taking this one step 
further one finds $\overline{\nabla_\mu \nabla_\nu\Omega_0}=\frac{1}{6}
R_{\mu\nu}$ (Ricci tensor). Thus 
\be \label{3.80}
\Omega_0=1+\frac{1}{12} R^{\mu\nu} \sigma_{,\mu} \sigma_{,\nu}
\ee
up to second order. Inserted into (\ref{3.78}) for $k=1$ gives 
$\overline{\Omega_1}=\frac{1}{6} R$ (Ricci scalar). To compute (\ref{3.79})
for $k=1$ requires to know $\nabla^3 \Omega_0$ etc.\\
\\
Concluding, we can consider the heat kernel to be known to any desired 
order $k,l$ in $t,[\nabla\sigma]$. Also arbitary covariant derivatives 
of the heat kernel can be obtained from again computing derivatives 
of (\ref{3.75}) and taking coincidence limits as well as the known 
coincidence limits of $\sigma$. 

\subsubsection{Functions of the heat kernel and traces}
\label{s3.9.4}

Let for, $s=0$, $f(z)$ be a complex valued function defined for $\Re(z)>0$ and
$\hat{f}(t)$ its pre-image under the Laplace transform
\be \label{3.81}
f(z)=\int_{\mathbb{R}^+}\; dt\; e^{-tz}\; \hat{f}(t),\;
\hat{f}(t)=[2\pi i]^{-1}\int_{x+i\mathbb{R}}\; dz\; e^{tz}\; f(z);\;x>0
\ee
Let, for $s=1$, $f(z)$ be a complex valued function defined on the 
real line and
$\hat{f}(t)$ its pre-image under the Fourier transform
\be \label{3.82}
f(z)=\int_{\mathbb{R}}\; dt\; e^{-itz}\; \hat{f}(t),\;
\hat{f}(t)=[2\pi]^{-1}\int_{\mathbb{R}}\; dz\; e^{itz}\; f(z);\;x>0
\ee
For a function $f=f(-\Delta)$ we may use the spectral theorem 
\be \label{3.83}
f(-\Delta)=\left\{ \begin{array}{cc}
\int_0^\infty \; dt\; \hat{f}(t) \; e^{t\Delta} & s=0\\
\int_{-\infty}^\infty \; \hat{f}(t) dt\; e^{it\Delta} & s=1
\end{array}
\right.
\ee
Then the trace of $f(-\Delta)$ with respect to the volume form of $m$ 
is given by 
\be \label{3.83a}
\int\; dt\; \int\; d^n x\; |\det(m)|^{1/2} H_t(x,x)
\ee
which demonstrates why we are interested in the coincidence limit and 
justifies the assumption about the convex normal neighbourhood above. 
Obviously $H_t(x,x)\propto \overline{\Omega_t}(x)$ which is an 
expansion in terms of curvature invariants.     

As laid out in the previous subsection, on the r.h.s of the 
Wetterich equation what we are interested in are 
expressions of the form (recall (\ref{3.47}), (\ref{3.48})
\be \label{3.84}
{\rm Tr}(P_0(-\Delta))\; Q_1(\nabla)\; P_2(-\Delta)\;..\;
Q_m(\nabla) P_m(-\Delta))
\ee
where $m\le M-1$ and 
where the $P_k$ just depend on $\Delta$ (``diagonal operators'')  while
the $Q_k$ may depend on all combinations of $\nabla$. We may write 
these expressions as 
\ba \label{3.85}
&& \int\ d^{m+1}t\; \hat{P}_0(t_0)\;..\;\hat{P}_m(t_m)\;
{\rm Tr}(e^{i^s t_0\; \Delta}\; Q_1(\nabla) e^{i^s t_1 \Delta} .. 
Q_m(\nabla) e^{i^s t_m \Delta})
\\
&=& 
\int\ d^{m+1}t\; \hat{P}_0(t_0)\;..\;\hat{P}_m(t_m)\;
{\rm Tr}([e^{i^s r_1\; \Delta}\; Q_1(\nabla) e^{-i^s r_1 \Delta}]
e^{i^s r_2\Delta}\; Q_1(\nabla) e^{-i^s r_2 \Delta}]
..
[e^{i^s r_m\Delta}\; Q_m(\nabla) e^{-i^s r_m \Delta}]
e^{i^s r_{m+1}\Delta)}
\nonumber
\ea
where $r_k=t_0+..t_{k-1},\;k=1,..,m+1$. Then we may use 
\be \label{3.86}
e^{i^s r \Delta} Q(\nabla) e^{-i^s \Delta}
=\sum_{n=0}^\infty \frac{(i^s r)^n}{n!} [\Delta,Q(\nabla)]_{(n)},\;
[A,B]_{(0)}:=B, 
[A,B]_{(n+1)}:=[A,[A,B]_{(n)}]
\ee
and discard in the sums over $n$ all terms above a certain order 
order that is dictated by the chosen values of $N,M,A$.
Then working out the multiple commutators we end up, inside the integrals 
over $t_0,..,t_m$ with an
expression of the form 
\be \label{3.87}
{\rm Tr}(Q(\nabla) H_r(\Delta)) 
\ee
which can be evaluated by the tools of the previous subsection. Also here 
we discard terms above a certain order in the expression of 
$\overline{\Omega_r}$ above a certain order in terms of $k,l$ as dictated
by $N,M,A$.

\subsubsection{Performing the Laplace or Fourier integrals}
\label{s3.9.6}

It remains to perform the integrals over the heat kernel times $t_0,..,t_m;\;
m\le M-1$. Their convergence is of course much determined by the properties of 
the Laplace or Fourier transform of $R_k$. As an example consider 
the truncation $N=1$ for the case $N^2_{s,k}(\bar{q})=0$ and 
$D^2_{s,k}(\bar{q})=\Delta+R_k(\Delta)$ then it suffices to consider 
the case $M=1$ and we need 
\be \label{3.88}
{\rm Tr}([\Delta+R_k(\Delta)]^{-1} [\partial_k R_k(\Delta)])
=\int\; dt\; \hat{f}_k(t) {\rm Tr}(e^{i^s t \Delta})
=\int\; dt\; \hat{f}_k(t)\; [4\pi |t|]^{-n/2} e^{is\;{\rm sgn}(t)\pi/4} 
\int\; d^n x\; [\sum_{k'=0}^\infty (i^s t)^{k'} \overline{\Omega_{k'}}(x)] 
\ee
i.e. we require the integrals for $k'=0,1,2,..$
\be \label{3.89}
q_k(n,k'):=\int\; \hat{f}_k(t)\; t^{k'}\; |t|^{-n/2} 
e^{is\;{\rm sgn}(t)\frac{\pi}{4}}
\ee
over $\mathbb{R}_+$ ($s=0$) and $\mathbb{R}$ ($s=1$) respectively 
where $\hat{f}_k(t)$ is the {\it inverse} of the 
Laplace respectively Fourier transform of 
\be \label{3.90}
f_k(x)=\frac{\partial_k R_k(x)}{x+R_k(x)}
\ee
Strictly speaking, what we are given is $f_k(x)$ and we assume that there
exists $\hat{f}_k(t)$ such that $f_k(x)=\int_{\mathbb{R}^+}\;dt\; 
e^{-xt}\hat{f}_k(t)$ for $s=0$ and $f_k(x)=\int_{\mathbb{R}}\;dt\; 
e^{ixt}\hat{f}_k(t)$ for $s=1$ respectively. This formula can be inverted
for $\hat{f}_k(t)$ for $s=1$ on the space of Schwartz functions and 
tempered distributions while 
for $s=0$ the existence of $\hat{f}_k$ is not granted given $f_k$
(there exists an inverse if instead $\hat{f}_k$ is given). This 
fact has important consequences for some of the results quoted in the 
literature in particular in connection with the so called {\it optimised
cut-off} \cite{1}. 
As we show in appendix \ref{sc}, the optimised cut-off is not granted (and 
is likely not)  
to lie in the image of the Laplace transform of a 
meaningful mathematical object (function, distribution, measure,..)
which would mean that the formulas that 
one quotes for its $q_k$ integrals do not hold and have to be revisited.

Whether the integrals (\ref{3.89}) exist depends on the chosen
shape of the suppressing function $R_k(x)$ which may 
depend on the chosen signature $s$. For the case $s=1$ we may pick 
$f_k(x)$ as the the Fourier transform of $\hat{f}_k$ where besides 
the already stated conditions 1.-4. we require $\hat{f}_k$ to be  
smooth of rapid decrease at $t=0,\pm\infty$. A controllable class of 
such functions and its integrals was studied recently in \cite{27a}
and we showcase their usefulness in appendix \ref{sd}. 
Again one will discard 
terms above a maximal value of $k'$ in agreement with the chosen values 
of $M,N,A$.\\
\\
In conclusion for suitable choices of $N,M,A$, upon discarding, we may evaluate 
both the l.h.s. and the r.h.s. of the Wetterich equation in closed form
producing a closed systsem of equations $\partial_k g(k)=\beta(g(k))$.

\section{Contact between CQG and ASQG in a concrete model}
\label{s4}

In this section we combine the two frameworks and study a concrete model
for Lorentzian GR coupled to some form of matter.
In the first subsection we introduce the classical Lagrangian of the model
and carry out the steps of section \ref{s2} to pass from the Hamiltonian
formulation to the path integral. In the second subsection we release 
the asymptotic safety machinery to the model. To avoid confusion note 
that the meaning of the indices $a,b,c,..,j,k,l,..$ in this section 
is different from section \ref{s2}: There they took infinite range 
labelling test smearing functions while here they take finite range 
$1,..,n-1$ labelling tensor components and field sepcies respectively.

\subsection{CQG treatment of the model}
\label{s4.1}

\subsubsection{Classical analysis}
\label{s4.1.1}

The classical Lagrangian is given by 
\be \label{4.1}
L=\frac{1}{G}\;|g|^{1/2}[R^{(n)}(g)-2\Lambda]
-\frac{1}{2}\;|g|^{1/2}\;g^{\mu\nu}\; 
S_{IJ}\; \phi^I_{,\mu}\;\phi^J_{,\nu}
\ee
where $g$ is a Lorentzian signature metric on an $n-$dimensional manifold 
$M$ diffeomorphic to $\mathbb{R}\times \Sigma$ and 
$\phi^I,\;I=0,..,n-1$ are $n$ scalar fields, \; $\mu,\nu=0,..,n-1$ are 
tensor indices. The real valued, constant matrix $S_{IJ}$ is positive 
definite, $G,\Lambda$ are Newton's and the cosmological constant 
respectively. Thus (\ref{4.1}) is just the Einstein-Hilbert Lagrangian 
minimally coupled to $n$ massless Klein-Gordon fields without potential.
This model was also treated by the methods of LQG in \cite{27b}.
To avoid discussions about boundary terms we assume that $\Sigma$ is 
compact without boundary.

We note that since $S$ is positive we find non-singular matrices 
$s^I_J$ such that $M_{IJ}=\delta_{KL} s^K_I s^L_J$ and thus could replace 
$S_{IJ}$ by $\delta_{IJ}$ by the field redefinition $s^i_J \phi^J \mapsto \phi^I$.
We will however keep $S_{IJ}$ general as it does not blow up the 
formalism. The above model can be considered as a dark matter model 
as the scalars couple only gravitationally, albeit not a very realistic
one as there is no mass term.

The classical Hamiltonian formulation of this system is well known
\cite{28}. 
An arbitrary foliation of $M$, using time $t$ and spatial coordinates 
$x^a$ on $\Sigma$ and 
as parametrised by lapse $N$ and shift $N^a$
functions $a,b,..=1,..,n-1$ leads to canonical pairs 
$(P,N);(P_a,N^a);(p^{ab},q_{ab}),(\pi_I,\phi^I)$ where 
$P,P_a$ are constrained to vanish (primary constraints) while 
\be \label{4.2}    
p^{ab}=\sqrt{\det(q)}[q^{ac} q^{bd}-q^{ab} q^{cd}]\;k_{cd},\;
k_{ab}=\frac{1}{2N}[\dot{q}_{ab}-[L_{\vec{N}}q]_{ab}],\;
\pi_I=\sqrt{\det(q)}\; S_{IJ} [L_n \phi^I]
\ee
where $L_u$ denotes the Lie derivative with respecto to $u$ and where  
$n=N^{-1}(\partial_t-\vec{N}),\; \vec{N}=N^a \partial_{x^a}$. Here $q_{ab}$ 
is the intrinsic spatial metric of Euclidian signature on the leaves of the 
foliation and $k_{ab}$ their extrinsic curvature and $n$ is a unit timelike 
normal to the leaves while $\vec{N}$ is tangential. 

The secondary constraints are (all spatial indices are moved with 
$q_{ab},\ q^{ab};\;q_{ac} q^{cb}=\delta_a^b$ and $Q:=[\det(q)]^{1/2}$)
\ba \label{4.3}
C_a &=& -2 D_b \; p^b_a+\pi_I \phi^I_{,a};\;\;\;
C= C^g+C^s;
\\
C^g &=& Q^{-1}\;[p^{ab} p_{ab}-\frac{1}{n-2} (p^a\;_a)^2]
-Q\;[R^{(n-1)}(q)-\Lambda],\;\;
C^s=\frac{1}{2}[Q^{-1}\; S^{IJ} \pi_I \pi_J+Q\; S_{IJ}\; q^{ab}\; 
\phi^I_{,a}\phi^J_{,b}]
\nonumber
\ea
known as spatial diffemorphism and Hamiltonian constraint respectively.
Here $D$ is the Levi-Civita differential of $q$ (note that $P^{ab}$ is a 
spatial tensor density of weight one). It is straightforward but tedious 
to check that the $2n$ constraints are first class, i.e. their mutual 
Poisson brackets vanish when all constraints hold.

The constrained Hamiltonian density is 
\be \label{4.4}
H=v\; P+v^a\; P_a+N^a C_a+N C
\ee
where $v,v^a$ are undetermined Lagrange multipliers as the Legendre transform
does not determine the velocities of $N,N^a$. This Hamiltonian is a sum 
of constraints and its Hamiltonian flow is to be interpreted as gauge 
transformations.  
   
We now consider the reduced phase space of the system. 
We assume that the $n$ scalar fields serve as a material reference system
in the sense that 
\be \label{4.5}
\det([\partial \phi/\partial (t,x))\not=0
\ee
so that $\phi$ in fact defines a diffeomorphism. It is then natural to 
adopt the gauge fixing condition
\be \label{4.6}
G^I(t,x)=\phi^I(t,x)-k^I(t,x), G^0=N-c(t,x), G^a=N^a-c^a(t,x)
\ee
where $k^I,c,c^a$ are coordinate functions independent of the phase space 
variables subject to the constraint $\det(\partial k/\partial (t,x))\not=0$.
That is, we consider $\phi^I,N,N^a$ and their conjugate momenta
as gauge degrees of freedom while $q_{ab},p^{ab}$ are the true degrees of 
freedom. 
The correspondence with the general theory laid out in section \ref{s2}
is that $(x,y)$ are given by $(P,P_a,\phi^I),(N,N^a,\pi_I)$ and 
$(p,q)$ by $(p_{ab},q^{ab})$.
The model thus has $\frac{1}{2}n(n+1)+n -2n=\frac{1}{2}n(n-1)$ 
physical degrees of freedom, i.e. $n$ more than in vacuum due to the 
presence of the $n$ scalar fields. That we  choose the true degrees 
of freedom purely in terms of $q_{ab}$ rather than say the  
$\frac{1}{2}n(n+1)-2n=\frac{1}{2}n(n-3)$ gravitational wave polarisations and 
$n$ scalar fields is because the description of the reduced phase space
is much simpler then, as we can solve the constraints algebraically
(otherwise we would need to solve PDE's). This is similar to the 
Higgs mechanism where rather than 2 polarisations for the $W_+,W_-,Z$ boson 
we have 3 because 3 of 4 degrees of freedom of the complex Higgs doublet 
are gauge fixed. The additional polarisation degrees of freedom could 
be called Goldstone bosons.  

In detail, let $j,k,..=1,..,n-1$ then we may solve algebraically 
\ba \label{4.7}
&& \hat{C}_j := \pi_j+h_j,\;\hat{C}=\pi_0+h_0
\\
h_j &=& \phi_j^a\;[-2 D_b P^b_a+\pi_0 \phi^0_{,a}] =:h_j^g+\pi_0 \phi^0_j;
\;\; \phi^j_{,a} \phi_j^b=\delta_a^b
\nonumber\\   
h_0 &=&\frac{B}{A}\pm\sqrt{-\frac{Q}{A}\;[\frac{B^2}{A}+ 
C_g+\frac{1}{2} (Q^{-1}\; 
S^{jk} h_j^g h_k^g+ Q\; S_{IJ}\; q^{ab}\; \phi^I_{,a} \phi^J_{,b}]}
\nonumber\\
A &=& S^{00}-2S^{0j} \phi^0_j+S^{jk}\phi^0_j\phi^0_k
\nonumber\\
B &=& (S^{0j}+S^{jk} \phi^0_k) h^g_j
\ea
where the solution $\pi_0=-h_0$ is to be inserted in the first line. 

Note that as detailed in section \ref{s2} there are two possible roots 
for $\pi_0$. Note also that the solutions $h_I$ drastically simplify 
at the gauge cut $\phi^I=k^I$ when $k^0_{,a}=0$ 
so that $k^0_j=k^a_j k^0_{,a}=0$
and if $S^{0j}=0$ i.e. $S_{0j}=0$ i.e. $S_{IJ}$ is block diagonal 
in which case $B=0, A=S^{00}=S_{00}^{-1}$. In fact, in order to apply
the results of section \ref{s2} we must adopt this choice because 
in this case 1. the quadratic constraints are of the type $C_j=\sum_k 
P_j^k \hat{C}_k^+\; C^k_-$ with $\hat{C}^\sigma_k=y_k+h_k^\sigma$, 2. 
$h_k^\pm=\pm h_k$ and 3. $h_k$ depends only quadratically on $x^j$.   

The reduced Hamiltonian density is according to the general theory of section 
\ref{s2} given by
\be \label{4.8}
H_\ast(q,p;t)=\dot{k}^I(t)\; h_I(q,p;\phi=k(t))
\ee 
It is not 
explicitly time dependent iff $\dot{k}$ is a constant and the particular 
way in which $k$ appears in $h_I$ loses its time dependence. Since the 
square root in $h_0$ depends on the term $S_{IJ} q^{ab} k^I_{,a} k^J_{,b}$
we require that $k^I_{,a}$ is time independent. The terms $A,B$ in $h_0$ 
are then also time independent since $k^0_j=k^a_j k^0_{,a}$ and 
$k^a_j$ is the inverse of $k^j_{,a}$. Accordingly the gauge fixing condition 
is of the form $k^I(t,x)=k^I_0\;t+k^I_1(x)$ where $k^I_0$ are $n$ constants 
and $k^I_1(x)$ are functions of the spatial coordinates only, subject 
to the condition that 
\be \label{4.9}
\det(\partial k/\partial(t,x))
=\frac{1}{(n-1)!}\; k_0^{I_0}\epsilon_{I_0..I_{n-1}}\epsilon^{a_1..a_{n-1}} 
\prod_{l=1}^{n-1} k^{I_l}_{1,a_l}
\not=0
\ee
A particularly simple choice is $k^I_0=\kappa_0 \delta^I_0$ in which 
case (\ref{4.9}) reduces to the condition that 
$\det(\partial k_1/\partial x)\not=0$ which is for example 
achieved for $k^j_1(x)=[\kappa_1]^j_a x^a$ and $\kappa_1$ is an 
invertible matrix. In this case 
\be \label{4.10}
S_{IJ} q^{ab} k^I_{,a} k^J_{,b}=q^{ab} S_{ab},\; S_{ab}=S_{jk} 
[\kappa_1]^j_a [\kappa_1]^k_b
\ee
Note that we could keep $H_\ast$ free 
of explicit time dependence in this gauge if we would add a mass and 
potential term which only depend on $\phi^j$ but not on $\phi^0$. 

At this point we have declared $q_{ab}, p^{ab}$ as the independent 
physical degrees of freedom of the model. However, to define an entire 
spacetime metric it is not enough to know the time evolution of 
$g_{ab}=q_{ab}$, we also need to know $g_{tt}=-N^2+q_{ab} N^b$ and 
$g_{ta}=q_{ab} N^b$ to have access to the full spacetime metric, i.e. 
we need to know $N,N^a$ as functions of $q_{ab},p^{ab}$. This dependence 
is provided by the stability condition on the gauge condition, i.e.
\be \label{4.10a}
[\partial_t G^I+\{C(N),G^I\}]_{\phi=k,\pi=-h}=-\dot{k}^I+N^a k^I_,a
-\frac{N}{Q}\;S^{IJ} h_J=0
\ee
which can be solved for $N,N^a$. For the choice discussed above which 
avoids explicit time dependence of the reduced Hamiltonian one finds 
\be \label{4.10b}
N=-\kappa_0 S_{00} \frac{Q}{h_0},\;
N^a=-\kappa_0 \; S_{00}\; ([\kappa_1]^{-1})^a_j\; S^{jk}\; \frac{k_k}{h_0}
\ee

\subsubsection{Canonical quantisation}
\label{s4.1.2}

In the canonical quantisation of this model one now considers the 
Weyl algebra generated by 
\be \label{4.11}
W(f,g)=\exp(i\int_{\Sigma}\; d^{n-1}x\; [f^{ab} q_{ab}+g_{ab} p^{ab}])
\ee
and representations thereof which allow a quantisation of $H_\ast$. 
This is quite a challenging task due to the square root involved and 
all the other non-linearities that are involved in $C^g$ which depends 
on inverse powers of $Q$.  
To make this task as simple as possible, one will resort to the 
simplest gauge condition $k^0=\kappa_0 t, k^j=[\kappa_1]^j_a x^a$ 
and block diagonal $S$ so that for the choice of the positive root  
\be \label{4.12}
H_\ast(q,p) = \frac{\kappa_0}{S_{00}^{1/2}}\;\sqrt{-Q\;[ 
C_g+\frac{1}{2}(Q^{-1}\; S^{ab} h_a^g h_b^g+ Q\; S_{ab}\; q^{ab})]}
\ee
where we have used 
\be \label{4.13}
S^{jk} h_j^g h_k^g=S^{ab}\; h_a^g h_a^g,\;h_a^g=-2 D_b p^b_a 
\ee
with $S^{ac} S_{cb}=\delta^a_b$ and $S_{ab}$ is just a constant matrix.  
One possibility to proceed is to notice that the expression $-E$ 
under the square
root is classically constrained to be positive and therefore classically 
\be \label{4.14}
H_\ast = 
\frac{\kappa_0}{S_{00}^{1/2}}\; \sqrt{|-E|}
=\frac{\kappa_0}{S_{00}^{1/2}}\; [E^2]^{1/4},\;\;
E=Q\;[ 
C_g+\frac{1}{2}(S^{ab} h_a^g h_b^g+ Q\; S_{ab}; q^{ab}]
\ee
which now at least frees us of making sure that the expression under the 
square root is positive after quantisation. The function $E^2$ can now 
be ordered symmetrically and one can implement this as an honest symmetric 
operator on a finite lattice (recall that $\sigma$ was chosen compact) and  
then study its Hamiltonian renormalisation flow.

\subsubsection{Path integral formulation}
\label{s4.1.3}

In the path integral approach outlined in section \ref{s2} one can now 
formally obtain the generating functional of Feynman ($s=1$) or Schwinger 
($s=0$) functions
\be \label{4.15}
Z_s(f)=\int[dq]\;[dp]\; \overline{\Omega_0(q(t=-\infty))}\;
\Omega_0(q(t=\infty))\; 
e^{-\int\; d^n x[i\; p_{ab}\dot{q}^{ab}+ (-1)^s\;i^s\; H_\ast(q,p)]}
e^{i^s \int\; d^n x\; f^{ab} q_{ab}}
\ee
This formula is rigorous as long as one understands that the 
integrals involved are in fact finite Riemann sums, and accordingly
$[dp], [dq]$ finite product Lebesgue measures, in a discetisation 
of the phase space on the compact spacetime $[-T,T]\times \sigma$ 
upon which $t=\pm \infty$ in the cyclic state is to be replaced 
by $t=\pm T$. As usual we keep a continuum notation but the steps 
that follow can only be strictly justified with that 
discretised understanding.

There are now two possibilities. Either one stays within this 
strictly reduced phase space and tries to integrate out the momenta 
$p^{ab}$. This is not possible directly because of the square root 
involved in $H_\ast(q,p)$. However, one can get rid of the square 
root using an exact integral transform based on a Lagrange multiplier 
field. Or one passes to an unreduced phase space path integral 
along the lines of section \ref{s2} which is another 
way to get rid of the square root. We will describe both methods below.
For both methods we obtain a problem in the case $s=0$ which is 
specific to GR. The problem of the first approach is that while 
the argument of the square root involved in $H_\ast$ is classically 
constrained to be positive, in the integral over all paths that 
condition cannot be maintained, the Hamiltonian is unbounded from below 
and the momentum integral diverges which is partly caused by the conformal 
mode. If one ignores the infinite factors that appear in carrying 
out the momentum integrals because they 
formally drop out in the non vacuum correlation functions one still 
faces the problem that the resulting expression is spatially rather 
non-local. We leave it for future investigations.
The problem of the second approach is that while there is no 
convergence issue in integrating out the momenta, we face the problem 
already mentioned in section \ref{s2}, namely that the ``Euclidian'' 
action becomes complex valued rather than real valued. In the second 
approach we are therefore 
unambiguously directed to consider the $s=1$ formulation 
(Feynman functions rather than Schwinger functions). In both 
formulations, the configuration space measure receives non-trivial
corrections to the ``naive'' Lebesgue measure.\\
\\  
\\
{\bf Reduced configuration space path integral}\\
\\
\\
The square root Hamiltonian 
reminds of the square root action for the relativistic particle or the 
Nambu-Goto action for closed bosonic string and one may recall 
\cite{10} that in a lowest order saddle point approximation 
\be \label{4.16}
\int\; d\lambda\; e^{\frac{c}{2}[\lambda f+\lambda^{-1}}])
\propto e^{c \sqrt{f}}
\ee
(the infinite constant involved
would drop out in $Z_s(f)/Z_s(0)$) would let us pass 
to the analog of the Polyakov action without square root. If
(\ref{4.16}) was exact one could now integrate out $p$ as $H_\ast^2$ depends 
only quadratically on $p$ yielding a Gaussian inegral. 
Unfortunately (\ref{4.16}) is not exact. 
However, there exists a more complicated version
of (\ref{4.16}) which is based on lemma \ref{la.a.1} proved in the appendix
and which to the best of our knowledge has not been reported before.

We write $H_\ast=:Q\sqrt{h}$ which displays $h$ as a scalar with zero 
density weight while $H_\ast,Q$ have density weight unity. 
Let $V_I:=\epsilon_I Q_I$ be a discretisation of 
$\int_{c_I}\; d^4x\; Q$, the volume measured by $Q$ of an $n-$cell $c_I$ with 
centre $p_I$
into which $M=[-T,T]\times \sigma$ is partitioned (the partition is finite)
with $\epsilon_I$ its coordinate volume and $Q_I=Q(p_I)$. Then with 
our understanding of the symbolic continuum expression with $h_I=h(p_I)$
\be \label{4.16a}
\int\;d^4x\; H=\sum_I\; V_I\; \sqrt{h_I}
\ee
By lemma \ref{la.a.1} we have with $d=V_I, z=h_I$
\be \label{4.16b}
e^{(-1)^{1+s}\; i^s \sum_I V_I \sqrt{h_I}}
=\prod_I \int_{\mathbb{R}_+} \; \frac{d\lambda_I}{\lambda_I^{3/2}}\;
\prod_I \frac{e^{-i\;s\;\frac{\pi}{4}}\;V_I}{2\sqrt{\pi}}\; 
e^{(-1)^{1+s}\; i^s \sum_I [\lambda_I h_I+\frac{V_I^2}{4\lambda_I}]}
\ee
We substitute $\lambda_I\to \lambda_I V_I$ to obtain 
\be \label{4.16b1}
e^{(-1)^{1+s}\; i^s \sum_I V_I \sqrt{h_I}}
=\prod_I \int_{\mathbb{R}_+} \; \frac{d\lambda_I}{\lambda_I^{3/2}}\;
\prod_I \frac{e^{-i\;s\;\frac{\pi}{4}}\;V_I^{1/2}}{2\sqrt{\pi}}\; 
e^{(-1)^{1+s}\; i^s \sum_I V_I\;[\lambda_I h_I+\frac{1}{4\lambda_I}]}
\ee
Accordingly, dropping the power of $e^{-i s \pi/4}/(2\sqrt{\pi})$ we have 
\be \label{4.16c}
Z_s(f)=\int[Q^{1/2} dq]\;[dp]\;[\lambda^{-3/2}\;d\lambda]\; 
\overline{\Omega_0(q(t=-\infty))}\;
\Omega_0(q(t=\infty))\; 
e^{-\int\; d^n x[i\; p_{ab}\dot{q}^{ab}+ (-1)^s\;i^s\; (\lambda\; H(q,p)
+\frac{Q}{4\lambda})]}
e^{i^s <f,q>}
\ee
where it is understood that one integrates over $\lambda\in \mathbb{R}_+\;\;\;
q_{ab}, p^{ab}\in \mathbb{R}$ and 
\be \label{4.16d}
H=|-[C_g+\frac{1}{2}(S^{ab} h_a^g h_b^g+ Q\; S_{ab}; q^{ab})]|
\ee
The expression in (\ref{4.16d}) between the absolute values 
is constrained to be positive in the classical theory (this is possible 
because the gravitational contribution $C_g$ is not bounded 
from below). This can be 
ensured in the path integral only if one does not integrate unconditionally 
over $q_{ab},p^{ab}$. But then integrating out $p_{ab}$ is not just 
reduced to Gaussian integrals and therefore not possible explicitly. 

To explore (\ref{4.16d}) we assume without proof that the paths in phase 
for which one can drop the modulus dominate. This means that we have 
to perform the Gaussian integral
\be \label{4.16e}
\int\; [dp]\; e^{-[i<p,\dot{q}>-(-1)^s\; i^s[<p,O(q,\lambda)\;p>+
<Q,\lambda(\Lambda-R+\frac{1}{2}S\cdot q)+\frac{1}{4\lambda}]>}
\ee
where $O(q,\lambda)$ is the differential operator
\be \label{4.16f}
[O(q)\;p]_{ab}:=Q^{-1} G_{abcd}(q,\lambda) \;p^{cd}-2 
D_{(a} q_{b)c} \frac{\lambda S^{cd}}{Q} q_{de} D_f p^{fe},\;
G_{abcd}=\lambda(q_{a(c} q_{d)b}-\frac{1}{n-2} q_{ab} q_{cd})
\ee
This differential operator is known to be indefinite with respect 
to the first term (conformal mode problem) and the second term, although 
it looks like minus a Laplacian, is likely not to repair this. This 
is the reason why $s=1$ is preferred because oscillatory Gaussians 
$x\mapsto e^{i x^2}$ can be integrated. Completing the square we obtain 
after dropping (infinite) constant factors
\ba \label{4.16g}
Z_1(f) &=& \int[Q^{1/2} dq]\;[\lambda^{-3/2}\;d\lambda]\; 
\overline{\Omega_0(q(t=-\infty))}\;
\Omega_0(q(t=\infty))\; |\det[O(q,\lambda)]|^{-1/2}\;
\times \nonumber\\
&& e^{-i[\frac{1}{4} <\dot{q}\; O(q,\lambda)^{-1} \dot{q}>
-<\lambda Q,\Lambda-R+S\cdot q/2+(2\lambda)^{-2}>]}
e^{i <f, q>}
\ea
This expression has no resemblance any more with the naive expression 
(Lebesgue measure times exponential of the classical action). It has 
an exponential part but that part involves the inverse of 
the full differential operator
$O$ and not just the DeWitt metric term $G_{abcd}$ and thus is
spatially non-local. There is also 
a measure part which involves the determinant of $O$.\\ 
\\
In principle one can now release the ASQG machinery on (\ref{4.16g})
where no gauge fixing term and no ghost term appear because we have 
reduced the phase space prior to quantisation, i.e. this would 
be the background formalism without gauge invariance where 
a background is introduced just for $q_{ab}$ and not the full spacetime 
metric. There 
is no current for the Lagrange multiplier field $\lambda$ 
because its whole purpose 
was to get rid of the square root and thus plays a role similar to the 
ghosts in presence of gauge invariance which are just there in order to 
write the Fadeev-Popov determinant factor as an exponential. One can see 
this also from the fact that no time derivatives of $\lambda$ appear. 

In order to solve the Wetterich equation one needs to make an Ansatz for 
the EAA which is usually the logarithm of the integrand of $Z(0)$. 
If we write this integrand as a measure $d\mu[q]$ by formally integrating 
out $\lambda$ then this Ansatz would be $\bar{\Gamma}(\hat{q},\bar{q})=
i\ln(\frac{d\mu[q]}{[dq]})_{q=\bar{q}+\hat{q}}$. Since it is practically 
not possible to carry out the integral over $\lambda$ we instead 
artificially multiply the integrand by $\exp(i<\iota,\lambda>)$ where 
$\iota$ is a current for $\lambda$ thus obtaining $Z'_1(f,g)$ where of course 
we are only interested in $Z_1(f)=Z'_1(f,g=0)$. We can now define 
in the usual way 
$\bar{Z}'_1(f,g;\bar{q}),\; \bar{C}'_1(f,g;\bar{q}),\; 
\bar{\Gamma}'_1(f,g;\bar{q})$ 
but while $\bar{Z}_1(f,\bar{q})=\bar{Z}'_1(f,g=0,\bar{q})$ and   
$\bar{C}_1(f,\bar{q})=\bar{C}'_1(f,g=0,\bar{q})$ we {\it do not have}
$\bar{\Gamma}_1(\hat{q},\bar{q})=\bar{\Gamma}'_1(\hat{q},\hat{\lambda}=0,
\bar{q})$. Instead, $\bar{\Gamma}_1(\hat{q},\bar{q})$ can and must be 
retrieved from $\bar{\Gamma}'_1(\hat{q},\hat{\lambda},\bar{q})$ via 
lemma \ref{la4.1} proved in the appendix. The advantage of proceeding like 
this is that it frees us from integrating out $\lambda$ and thus the Ansatz 
for the EAA is now given by 
$\bar{\Gamma}(\hat{q},\hat{\lambda};\bar{q})=
i\ln(\frac{d\mu[q,\lambda]}{[dq][d\lambda]})_{q=\bar{q}+\hat{q},
\hat{\lambda}=\lambda}$ if we write the integrand of $Z(0)$ as 
$d\mu[q,\lambda]$. In a first approximation one may discard the contribution
from the determinant, the cyclic state and $[Q^{1/2}] [\lambda^{-3/2}]$.
We leave the investigation of this route for future publications.
\\
\\
\\
{\bf Unreduced configuration space path integral}\\
\\
\\ 
If one wants to avoid the above spatial non-localities one may 
proceed as in section \ref{s2}
to obtain a path integral over the unconstrained phase space given by 
\ba \label{4.17}
Z_s(f)&=&\int\;[dq]\;[dp]\;[d\phi]\;[d\pi]\;[dN]\;
|\int\; [d\eta]\;[d\rho] e^{-i^s\int\;d^n x \eta^\mu \{C_\mu,G^I\}\rho_I}|\;
\overline{\Omega_0(q(t=-\infty))}\;
\Omega_0(q(t=\infty))\; 
\;\delta[G]\; 
\times \nonumber\\
&& e^{-i\int\;d^nx\;(\;p^{ab}\;\dot{q}_{ab}
+(-1)^s\; i^s\;\pi_I\dot{\phi}^I+N^\mu C_\mu)}\;
\; 
e^{i^s\int\; d^n x f^{ab} q_{ab}}
\ea
where $C_\mu$ are the original constraints (\ref{4.3}) and 
\be \label{4.18}
\{C_\mu,G^I\}=\frac{\partial C_\mu}{\partial \pi_I}=
\delta_\mu^a\; \phi^I_{,a}+\delta_\mu^0\; Q^{-1}\; S^{IJ}\pi_J
\ee
Note that while the exponent now depends at most quadratically on $p,\pi$ 
there is again an obstacle for $s=0$: The symplectic term $\pi_I \dot{\phi}^I$ 
has a relative factor of $i$ as compared to $C_\mu$. Then performing 
the Gaussian inegral would not undo the Legendre transform    
(\ref{4.2}) but would rather replace $\pi_I$ by the complex valued rather than
real valued object
\be \label{4.19}
S_{IJ}\frac{1}{N}[i\dot{\phi}^J-N^a\phi^J_{,a}]
\ee
This is the problem already spelled out in the general case in section 
\ref{s2}. We thus fix $s=1$ for what follows. 

To integrate out $p^{ab}$ we collect the terms that depend on it
using (\ref{4.2})
\ba \label{4.20} 
&& p^{ab}\dot{q}_{ab}
-N Q^{-1}\; G_{abcd}\; p^{ab} p^{cd}-p^{ab} [L_{\vec{N}} q]_{ab}
=N Q^{-1}(2\; p^{ab}\; k_{ab}-G_{abcd}\; p^{ab} p^{cd})
\nonumber\\
&=& N \{Q G^{abcd} k_{ab} k_{cd} -Q^{-1}\; G_{abcd} 
[p^{ab}-QG^{abef} k_{ef}]
[p^{cd}-QG^{cdgh} k_{gh}]\}
\ea
where 
\be \label{4.21}
G_{abcd}=
q_{a(c}\;q_{d)b}-\frac{1}{n-2} q_{ab} q_{cd},\;
G^{abcd}=q^{a(c}\;q^{d)b}-q^{ab} q^{cd},\;
G_{abef}\; G^{efcd}=\delta_{(a}^c \delta_{b)}^d
\ee
The Gaussian integral over $p^{ab}$ can now be performed 
by shifting $p^{ab}$  
and producing a factor of $|\det(N Q^{-1} G)|^{-1/2}$ per spacetime point
up to a global phase that drops out in $Z_s(f)/Z_s(0)$
where the determinant 
is for square matrices of type $G$ which has rank $r=\frac{1}{2}n(n-1)$.
There is no confomal mode problem \cite{26} because the integral of an 
osciallatory Gaussian $e^{i a x^2},\;a\in \mathbb{R}-\{0\}$ 
exists while it does not for an undamped Gaussian $e^{ax^2},\; a>0$. 
Note that there would also not be a conformal mode problem if we had kept 
$s=0$ because the exponent $G_{abcd}(q) p^{ab} p^{cd}$ 
always has a prefactor $i$ as it comes from implementing $\delta[C]=\int
[dN]\; e^{i<N,C>}$. 
  
To compute this determinant, we note that 
\be \label{4.22}
G_{abcd}=\hat{G}_{ijkl} E^{ij}_{ab} E^{kl}_{cd},\;
\hat{G}_{ijkl}=\delta_{i(k}\delta_{l)j}-\frac{1}{n-2}\delta_{ij}\delta_{kl},\;
E^{ij}_{ab}=e^i_{(a} e^j_{b)}
\ee
where $q_{ab}=\delta_{ij} e^i_a e^j_b$ defines an $(n-1)-$Bein for $q$. 
The matrix $\hat{G}$ viewed as a matrix on symmetric tensors has signature 
$(r-1,1)$ i.e. $r-1$ positive eigenvalues $+1$ and one negative eigenvalue 
$-1/(n-2)$ corresponding to tracefree and pure trace tensors respectively.
Let $u^a_\alpha$ be an eigenvector of $e_a^i$ i.e. $e_a^i u^a_\alpha
=\lambda_\alpha \delta_a^i u^a_\alpha$. Then 
$u^{ab}_{\alpha\beta}=u^{(a}_\alpha u^{b)}_\beta$ is an eigenvector 
of $E^{ij}_{ab}$ i.e. $E^{ij}_{ab} u^{ab}_{\alpha\beta}
=\lambda_\alpha \lambda_\beta \delta^i_{(a} \delta^j_{b)} 
u^{ab}_{\alpha\beta}$. Thus 
\be \label{4.23}
\det(E)=\prod_{1\le\alpha\le\beta\le n-1} \lambda_\alpha \lambda_\beta
=[\prod_{1\le \alpha\le n-1} \lambda_\alpha]^n=\det(e)^n
\ee
and 
\be \label{4.24}
\det(G)=-\frac{1}{n-2} \det(E)^2=-\frac{1}{n-2} Q^{2n},\; Q=\sqrt{\det(q)}
\ee
Accordingly dropping overall factors
\ba \label{4.25}
Z_1(f) &=& \int\;[dq]\;[d\phi]\;[d\pi]\;[dN]\;
|\int\; [d\eta]\;[d\rho] 
e^{-i\int\;d^n x \eta^\mu \{C_\mu,G^I\}\rho_I}|\;
\overline{\Omega_0(q(t=-\infty))}\;
\Omega_0(q(t=\infty))\;
\times \nonumber\\
&& 
[(Q^{2n} [|N| Q^{-1}]^r)^{-1/2}] \;\delta[G]\; 
e^{-i\int\;d^nx\;(|g|^{1/2}[R^{(n)}(g)-\Lambda]
+\pi_I\dot{\phi}^I+N^\mu C^s_\mu)}\;
\; 
e^{i\int\; d^n x f^{ab} q_{ab}}
\ea
where we used the Gauss-Coddacci relation and $|g|^{1/2}=|N| Q$. 

To perform the integral over $\pi_I$ we note that $\pi_I$ also appears 
in the ghost term. The ghost integral appears in a modulus and contains a 
factor of $i$ in the exponent, but the ghost matrix is even dimensional 
and thus yields a real valued determinant. Thus the modulus just 
multiplies the the ghost integral without modulus by a factor $\pm 1$
depending on the non-ghost variables.  
We collect the $\pi_I$ dependent terms
\ba \label{4.26}
&& \pi_I [\dot{\phi}^I-N^a \phi^I_{,a}-Q^{-1} S^{IJ} \eta^0 \rho_J]
-\frac{1}{2} N Q^{-1} S^{IJ} \pi_I \pi_J 
\nonumber\\
&=& -\frac{N}{2Q} [\pi_I-[Q S_{IK} L_n\phi^K+N^{-1}\eta_0 \rho_I)] S^{IJ} 
[\pi_I-(Q S_{JL} L_n\phi_L+N^{-1}\eta_0 \rho_J)] 
\nonumber\\
&& +\frac{1}{2Q} (Q S_{IK} L_n\phi^K+N^{-1}\eta_0 \rho_I) S^{IJ} 
(Q S_{JL} L_n\phi_L+N^{-1} \eta_0 \rho_J) 
\ea
The term $\eta_0 \rho_I$ is bilinear in the ghosts and thus a commuting number,
we can shift it away in the integral together with the $L_n\phi$ term and 
then perform the $n$ Gaussians which up to a global factor gives 
$[|N|/Q]^{-n/2}$ per spacetime point. Note that the surviving term quartic
in the ghosts is ultra-local and thus drops out as $\eta_0$ is nil-potent. 
Thus we obtain
\ba \label{4.27}
Z_1(f) &=&\int\;[dq]\;[d\phi]\;[dN]\;
|\int \; [d\eta]\;[d\rho]
e^{-i\int\;d^n x [\eta^a \phi^I_{,a}+\eta_0 (L_n\phi^I)]\rho_I}|\;
\overline{\Omega_0(q(t=-\infty))}\;
\Omega_0(q(t=\infty))\;
\times \nonumber\\
&& 
[(Q^{2n} [|N| Q^{-1}]^r)^{-1/2} (|N|/Q)^{-n/2}] 
\;\delta[G]\; 
e^{-i\int\;d^nx\;|g|^{1/2}
[R^{(n)}(g)-\Lambda-\frac{1}{2}\;S_{IJ} g^{\mu\nu} \phi^I_{,\mu} \phi^J_{,\nu}]
}\;
\; 
e^{i\int\; d^n x f^{ab} q_{ab}}
\ea
The local measure factor can be written
\be \label{4.28}
[(Q^{2n} [|N| Q^{-1}]^r)^{-1/2} (|N|/Q)^{-n/2}] 
=[Q^{-1} (|N|/Q)^{-(n+1)/4}]^n=(|g|^{-(n+1)/8} Q^{(n+1)/2})^n=:f_1(g)^n
\ee
Instead of using $[dq][dN]$ 
we may want to use $[dg]$. We have $g_{tt}=-N^2+q_{ab} N^a N^b, 
g_{ta}=q_{ab} N^b, g_{ab}=q_{ab}$.
The change of variables gives the Jacobean
\be \label{4.29}
|\det(\partial g/\partial (N,\vec{N},q))|
=|-2N\det(q)|=2([-\det(g)\det(q)]^{-2/n})^n=:2 f_2(g)^{-n}
\ee
To see this note that the matrix $\partial g/\partial(N,\vec{N},q)$ is 
a rank $n(n+1)/2$ square matrix with rank $1,n-1,1/2 n(n-1)$ diagonal 
blocks given by $-2N, q_{ac}, \delta_{(a}^c \delta_{b)}^d$ 
respectively and which is upper triangular with respect to this block 
structure. Therefore its determinant is equal to the product of the 
determinant of its diagonal blocks.

We absorb $f_1, f_2$ into the ghost action and also carry out the gauge 
fixing $\delta$ distribution and find with $f(g):=f_1(g) f_2(g)$
\ba \label{4.29a}
Z_1(f)&=&\int\;[dg]\;
|\int\;[d\eta]\;[d\rho] 
e^{-i\int\;d^n x\; \eta^\mu M_\mu^I(g) \rho_I}| \; 
\overline{\Omega_0(q(t=-\infty))}\;
\Omega_0(q(t=\infty))\;
\times\nonumber\\
&& e^{-i\int\;d^nx\;|g|^{1/2}
[R^{(n)}(g)-\Lambda-\frac{1}{2}\;S_{IJ} \kappa^I_\mu \kappa^J_\nu g^{\mu\nu} 
]}\;
e^{i\int\; d^n x f^{ab} q_{ab}}
\ea
with the ghost matrix
\be \label{4.30}
M_\mu^I(g)
=f(g)[\delta_\mu^0 (L_n k)^I+\delta_\mu^a k^I_{,a}] 
=f(g)[N^{-1}\delta_\mu^0 (\kappa_0\delta^I_0-N^a[\kappa_1]^j_a \delta^I_j
+\delta_\mu^a [\kappa_1]^j_a \delta_j^I] 
\ee
Note that the gauge fixing condition requires $\phi^I(x)=\sigma(x) k^I(x)$
on the sector $\sigma$ but since the action depends only quadratically 
on $\phi^I$ and because the ghost determinant is formally a 
product over $x$ of determinants of 4x4 matrices depending 
linearly on $\phi^I$ the signs $\sigma(x)$ drop out of $Z(f)$ so that
the final expression only contains $k^I$. We have regularised 
$\partial_\mu \sigma\equiv 0$ by first relplacing by piecewise constant
$\sigma$ and taking one sided derivatives at jumps. 

One could of course integrate out the ghosts at the price of picking up 
a measure factor $|\det[M)(g)]|$ per spacetime point. Conversely we may 
formally absorb the state dependence into the ghost matrix by 
multiplying (\ref{4.30}) by the factor
\be \label{4.30a}
e^{h/n},\; h(x^0,\vec{x})=
\delta(x^0,\infty)\;I(x^0,\vec{x}),\;
+\delta(x^0,-\infty)\;\overline{I(x^0,\vec{x})},\;
\Omega_0(q(x^0))=:\exp(\int\;d^{n-1}x\; 
I(x^0,\vec{x})
\ee
It contributes only at $x^0=\pm\infty$ and is ultra-local in time
in the sense that $I(x^0,\vec{x})$ depends only on $q(x^0)$. 
We may pick for instance $I(x^0,\vec{x})=w(\vec{x})\;
\ln(\Omega_0(q(x^0)))$ where $w$ is any normalised weight function
$\int\; d^{n-1}x w(\vec{x})=1$. If $\Omega_0[q]=
\propto e^{-<q,l\cdot q>/2}$ 
is a Fock state for some spatial integral kernel $l$
it is natural to pick 
\be \label{4.30b}
I(x^0,\vec{x})=-1/2\; q_{ab}(x^0,\vec{x})\;\int\; d^{n-1}y\; 
l^{ab;cd}(\vec{x},\vec{y})\; q_{cd}(x^0,\vec{y})
\ee
~\\
\\
Remarkably the final result is very close to the Einstein-Hilbert path 
integral for vacuum GR with a cosmological constant which 
has been intensively studied in ASQG, the scalar field has 
disappeared completely. 
However, there are several {\it key} differences:\\
1.\\
There is no longer a (say de Donder) gauge fixing term for the metric and 
no corresponding 
Fadeev-Popov determinant (after integrating out the ghosts). 
This is because the model has $n$ additional 
degrees of freedom as compared to vacuum GR. These  are the $n$ scalar 
fields ``eaten by 
the gauge bosons'' $g_{\mu\nu}$ similar to the Higgs Kibble mechanism for 
the electroweak interaction. The gauge fixing term and FP determinant 
respectively have been replaced by the ``reduction term'' that 
corresponds to the $\kappa$ dependence in the exponent and respectively 
a more general 
determinant that depends on the Dirac matrix, the deWitt metric and 
ADM to metric variable transformations. While ghost and Dirac matrix
are cloesly related, the reduction term and the gauge fixing term are 
only related in the sense that both originate from gauge fixing, however,
otherwise they are logically indepnendent because we have solved 
the gauge fixing condition exactly via the $\delta$ distribution 
while the gauge fixing term just suppresses field configurations
violating the gauge condition. In this suppressed version, the 
reduction term would stilll be there but with $k^I$ replaced 
by $\phi^I$ together with the $\phi$ integral.\\   
2. \\
Accordingly in application of the background field method to 
the above generating functional we are in the situation of a theory 
with no gauge invariance left over any more, that is, we apply 
the method of section \ref{s2.2} rather than \ref{s2.3} and therefore 
do not need to worry about adapting background dependent and independent
gauge conditions (\ref{3.20}).\\
3.\\
The generating functional $Z_1(f)$ has a current $f$ only for the true 
degrees of freedom $q$. In principle one could therefore integrate out 
$N,N^a$ and obtain a path integral over $q$ only just as for a 
system without gauge symmetry at all. All correlators derivable 
from (\ref{4.30}) are directly observable, one does not need 
to worry about the gauge invariant meaning of correlators among the 
$g_{\mu\nu}$ as in the ususal setting.\\  
4.\\
The contribution from the cyclic state $\Omega_0(q(t=\pm \infty))$ 
reminds us that (\ref{4.29}) was truly derived from a canonical quantisation
time ordered correlation function. Note that the very notion of 
time ordering and a Hamiltonian 
is only meaningful in this relational observable setting where
one has a distinguished notion of time and Hamiltonian.\\
5.\\
The cosmological constant term is augmented by a field dependent 
contribution (reduction term)
$g^{\mu\nu} S_{IJ} \kappa^I_\mu \kappa^J_\nu$ which, like 
the measure factor $[\det(M(g)]$, breaks spacetime diffeomorphism 
or coordinate transformation invariance.
This is no problem because the coordinates have been fixed by the 
gauge kondition $\phi^I=\kappa^I_\mu x^\mu$ and thus acquire a 
physical, operational meaning.\\ 
6.\\
We found that one can obtain a closed expression in terms of a 
configuration path integral only for the generating functional 
of time ordered functions, rather than the Schwinger functions. 
Thus we are forced to consider the Lorentzian version of ASQG in what 
follows.\\
7.\\
The path integral is directly for Lorentzian signature QG, no Wick rotation 
(a questinable notion in a theory in which the metric is not a background)
from a Euclidian version is necessary.   

\subsection{ASQG treatment of the model}
\label{s4.2}

We consider $Z_1(f)$ as given in (\ref{4.29}) as our starting point with 
the ghost action not integrated out.  
We define in the usual way as reviewed in section \ref{s3} the functionals 
\be \label{4.31}
C_1(f)=i^{-1} \ln(Z_1(f)),\;\;
\Gamma_1(\hat{q})={\rm extr}_f[<f,\hat{q}>-C(f)]      
\ee
If everything would be well defined the definition of the theory would be 
complete at this point. As this is not the case we formally 
introduce additional terms to the current and extend with $F^{ab}=:f^{ab},\;
F^{t\mu}=:u^\mu$ 
\be \label{4.32}
<f,q>=\int d^nx f^{ab} q_{ab}\; \to \; 
<F,g>=\int d^nx F^{\mu\nu} g_{\mu\nu}=<f,q>+
\int\; d^nx\; u^\mu\; g_{t\mu} 
\ee
Thus we obtain $Z'_1(F)=Z'_1(f,u)$ with the property that 
$Z'_1(f,u)_{u=0}=Z_1(f)$ and consequently also $C_1'(F)=i^{-1} 
\ln Z_1'(F)$ satisfies $C_1'(f,u)_{u=0}=C_1(f)$. Next we introduce 
a background $\bar{g}$, write the dependence of the integrand of 
(\ref{4.29}) as $g=\bar{g}+h$ (i.e. in the cyclic state, the ghost 
matrix and Einstein-Hilbert like action term), replace $[dg]$ by $[dh]$ and 
$<F,g>$ by $<F,h>$, thereby obtaining 
$\bar{Z}'_1(F,\bar{g})=Z'_1(F)\; e^{-i<F,\bar{g}>}$ so that with 
$F=(f,u)$
\be \label{4.32a}
\bar{Z}'_1(f,u,\bar{g})_{u=0}
=Z'_1(f,u)_{u=0}\; e^{-i<f,\bar{q}>}
=Z_1(f)\; e^{-i<f,\bar{q}>}=\bar{Z}_1(f,\bar{q})
\ee
where $\bar{Z}_1(f,\bar{q})$ results from $Z_1(f)$ by just 
writing the spatial metric as $q=\bar{q}+\tilde{h}$ and replace
just $[dq]$ by $[d\tilde{h}]$ and $<f,q>$ by $<f,\tilde{h}>$, i.e.
the background method is just applied to $q$, not to the full $g$.

Next we introduce into $Z_1'(F,\bar{g})$ currents 
$e,f$ for the ghosts and the $R_k, R_k^g$ terms as 
displayed in section \ref{s3} thereby producing 
$\bar{Z}'_{1,k}(F,e,r,\bar{g})$ with 
$\bar{Z}'_{1,k=0}(F,e=0,r=0,\bar{g})=Z'_1(F,\bar{g})$ where we can 
set $e=r=0$ before or after setting $k=0$.
From here on one is now precisely in the situation to which we 
can apply the ASQG machinery and obtain a Wetterich equation 
for $\bar{\Gamma}'_{1,k}(\hat{g},\hat{\eta},\hat{\rho};\bar{g})$. 
Note that the Wetterich equation applies because its derivation 
does not make use of any details of the measure $d\mu(h,\eta,\rho)$, 
the only structural elements that it refers to are the $R_k, R_k^g$
terms and the current terms. After solving it we compute 
$\bar{\Gamma}'_1(\hat{g},\hat{\eta},\hat{\rho};\bar{g})
=\bar{Z}'_{1,k=0}(\hat{g},\hat{\eta},\hat{\rho};\bar{g})$.
Now we apply lemma \ref{la4.1} to get rid of $\hat{\eta},\hat{\rho}$ and 
obtain $\bar{\Gamma}'_1(\hat{g},\bar{g})$ from 
$\bar{\Gamma}'_1(\hat{g},\hat{\eta},\hat{\rho};\bar{g})$. Then 
$\Gamma'_1(\hat{g})=\bar{\Gamma}_1'(\hat{g}',\bar{g})_{\hat{g}'=0,
\bar{g}=\hat{g}}$. Finally, splitting $\hat{g}=(\hat{q},\hat{N})$ 
we apply lemma \ref{la4.1} once more to get rid of $\hat{N}$ to finally 
obtain the desired answer $\Gamma_1(\hat{q})$.\\
\\
Alternatively, one could refrain from working withe the extended 
primed objects but then would need to adapt the asymptotic saftey machinery 
and in particular heat kernel methods to a split treatment of spatial and 
temporal derivatives, see \cite{5}. Also alternatively one can 
integrate out the ghosts i.e. refrain from introducing 
$R^g_k, e,f$ which has the advantage that we do not need to invoke 
lemma \ref{la4.1} but then the Ansatz for $\bar{\Gamma}'$ needs to
be generalised to incorporate the logarithm of the ghost determinant as 
detailed in (\ref{3.53}). \\
\\
This analysis is carried out in our companion paper \cite{companion}.

\section{Conlusion and Outlook}
\label{s5}

The aim of this paper was twofold: One the one hand, to review, in a hopefully
not too technical fashion, the formalism of CQG in its reduced phase space 
incarnation for ASQG practitioners and vice versa the formalism 
of ASQG for CQG practitiones. On the other, to show 
that the two approaches can be brought into contact. 

From the point of ASQG,
the CQG input is to provide a starting point or starting action for the 
analysis of the Wetterich equation which is based on a clear physical 
interpretation of the system: 
First, there is a clear distinction 
between true i.e. physical and gauge degrees of freedom.
Next, if one uses the unitary group generated by the reduced Hamiltonian 
to obtain the generating functional of physical Feynman distributions then the
generating functional {\it can} be formulated as a 
a path integral over Lorentzian signature metrics involving 
the Lorentzian Einstein-Hilbert action by unfolding the reduced phase space 
integral to the unreduced phase space.
If one uses instead the contraction semi-group generated by the 
reduced Hamiltonian 
to obtain the generating functional of physical Schwinger distributions then the
generating functional {\it cannot} be formulated as a 
a path integral over Euclidian signature metrics involving 
the Euclidian Einstein-Hilbert action by unfolding the reduced phase space 
integral to the unreduced phase space. In this case the best one can 
do is to keep the reduced phase space path integral and integrate out 
the momenta. We have sketched how to do this but focussed 
mostly on the Lorentzian interpretation. Furthermore, 
the process of integrating out the momenta gives  precise measure 
corrections
to the naive ``Ansatz Lebesgue measure times exponential of 
the sum of classical plus gauge fixing plus ghost action''. Finally 
there is a non-trivial 
dependence 
of that measure on the cyclic state on which the QFT underlying 
the Hamiltonian formulation is based.

From the point of CQG, the ASQG input is to provide a powerful 
renormalisation scheme for the reduced Hamiltonian since from the 
exponential of the Legendre transform of the effective 
action that one obtains from the ASQG framework at $k=0$ one can reconstruct
the matrix elements of the reduced Hamiltonian (in the Euclidian setting 
this is known as Osterwalder-Schrader reconstruction \cite{23}). 
We emphasised that at least in the Lorentzian regime from 
the point of CQG the only object of physical interest is the effective 
action which is the $k=0$ limit of the effective average action. The 
effective average action and its fixed point analysis 
serves merely as a tool to derive the effective action. 

We also had the opportunity to address a number of reservations that 
practitioners from other approaches to QG 
sometimes articulate concerning ASQG:\\
1.\\
Signature:\\
We have seen that the ASQG framework can also be applied 
to Lorentzian GR although parts of the framewoirk have to be adapted.
Mathematically this concerns the analytic properties of oscillatory 
versus damped integrals \cite{24}.
The suppression kernels $R_k$ of the Euclidian framework can at best 
be interpreted as 
oscillation
kernels in the Lorentzian framework.\\
2. \\
Background independence:\\
We showed that there is in fact no background dependence in ASQG: 
In ASQG one simply uses the background field method as a convenient tool 
to construct the background independent (effective) average action which 
is a standard instrument in QFT. It was invented by DeWitt, one 
of the founding fathers of CQG. The important point is to 
keep the background unspecified because only then can one 
unambiguously reconstruct 
the background independent effective action from the background dependent
one.\\
3.\\
Truncations:\\
Truncations in ASQG should be viewed as an approximation method of the 
renormalisation scheme, unfortunately without error control at present.
To the best of our knowledge, there is no renormalisation scheme without 
such an error uncontrolled approximation due to the sheer complexity 
of renormalisation group equations. Although in LQG, an incarnation 
of CQG, one can give a concrete UV finite definition of the reduced theory 
\cite{31,32} which in principle is complete,
that definition is not free from quantisation ambiguities and thus 
also requires renormalisation methods. In that sense the issue of truncations 
exists for both CQG and ASQG, in fact, any approach to QG.\\
\\
To see both frameworks in their interaction we have studied a concrete model 
for Lorentzian GR minimally coupled to scalar fields which is not taken 
to be a serious dark matter model but rather as a proof of principle 
or showcase example. Starting from CQG quantisation we have 
cast that model into a form to which a Lorentzian ASQG treatment is 
immediately applicable. We reserve that analysis for our companion paper
\cite{companion}. The Lorentzian version of ASQG follows in spirit 
all the steps of the Euclidian version but some modifications are 
necessary in order that the flow induced by the Wetterich equation be well
defined. This concerns in particular the class of cut-off functions 
which are now defined by Fourier transforms rather than Laplace 
transforms and thus have to come equipped with different analytical 
behaviour due to the fact that the Lorentzian heat kernel is a group 
rather than a semi-group. We have proposed such an adapted class of cut-off 
functions and as a spin-off found that some of the cut-off functions 
used in the Euclidian regime are in fact not granted to be 
in the image of the Laplace
transform of any meaningful mathematical object 
which likely requires revisiting some results reported 
in the literature based on the existence of such a pre-image.

The analysis performed in this contribution and our companion paper
\cite{companion} leaves of course many interesting open questions such as 
the amount of freedom in the choice of Lorentzian cut-off functions,
gauge reduction by other than Klein-Gordon matter, exploration 
of the Euclidian interpretation of the reduced phase space path integral
and its ASQG treatment. We leave this for future research.\\
\\    
~\\
{\bf Acknowledgements}\\ 
We thank Renata Ferrero for in depth discussions about the details of 
Euclidian ASQG. 

\begin{appendix}

\section{Relations between effective actions upon restriction}
\label{sb}

We consider a generating function $Z_2(f_1,f_2)$ in two variables $f_1,f_2$ 
such that $Z_1(f_1)=Z_2(f_1,f_2=0)$ is a generating function 
in one variable. Then of course also 
$C_1(f_1)=C_2(f_1,f_2=0)$ for 
$C_1(f_1)=i^{-s} \ln(Z_1(f_1)),\;C_2(f_1,f_2)=i^{-s} \ln(Z_1(f_1,f_2))$.
The question arises how the Legendre transforms $\Gamma_1(q_1)$ and 
$\Gamma_2(q_1,q_2)$ of these two 
functionals are related. The naive guess $\Gamma_1(q_1)=\Gamma_2(q_1,q_2=0)$
turns out to be completely false.
\begin{Lemma} \label{la4.1}  ~\\
Let $C_1(f_1), C_2(f_1,f_2)$ be any functionals such that 
$C_1(f_1)=C_2(f_1,f_2=0)=(R_2\cdot C_1)(f_1)$ where $R_2$ denotes  
restriction to zero in the second argument. Denote by 
$\Gamma_1(q_1)=[T_1\cdot C_1](q_1),\;
\Gamma_2(q_1,q_2)=[T_2\cdot C_2](q_1,q_2)$ 
their non-singular Legendre transforms w.r.t. 
only the first and both variables respectively. Then
\be \label{4.33}
\Gamma_1=T_1\circ R_2\circ T_2 \circ \Gamma_2
\ee
\end{Lemma}
\begin{proof}: \\
The identity (\ref{4.33}) 
follows from the fact that the Legendre transform is idempotent
$T_2\circ T_2={\rm id}_2$ when non-singular.
\end{proof}
That generically $\Gamma_1\not=R_2\circ \Gamma_2$ even if 
$C_1=R_2\circ C_2$
can already seen for a generic polynomial $C_2(\vec{f})$ of second order in 2 
variables 
$\vec{f}=(f_1,f_2)$ for which we define $C_1(f_1):=C_2(f_1,f_2=0)$.
If one writes $C_2=A+\vec{B}^T\vec{f}+\frac{1}{2}\vec{f}^T C \vec{f}$ then 
we see that $\Gamma_1=R_2\circ \Gamma_2$ 
iff $B_2=0=C_{12}$, i.e. there is no ``condensate'' and no coupling
between the subsystems.

\section{Avoiding a square root Hamiltonian}
\label{sa}

\begin{Lemma} \label{la.a.1} ~\\
Let $d,z>0$ and $c^s_d=\frac{d\;e^{-is\frac{\pi}{4}}}{2\sqrt{\pi}}$
with $s=0,1$. Then 
\be \label{a.1}
[c^s_d]^{-1}\; e^{(-1)^{s+1} i^s\; d\; \sqrt{z}}
=\int_{\mathbb{R}^+}\; \frac{dl}{l^{3/2}}\;
e^{(-1)^{1+s} i^s \; [z\;l+\frac{d^2}{4l}]}
\ee
\end{Lemma}
Note that the integral converges absolutely for $s=0$ both at infinity
and zero. 
One idea of proof is based on the observation of \cite{27a} that 
integrals of exponentials of $-(x^2+x^{-2})$ are in fact analytically 
performable and the saddle point approximation of (\ref{a.1}) which 
is proportional to the exponential of $(-1)^{1+s}\; i^s
\sqrt{z}d$ in leading order. 
Another idea 
is obtained for $s=0$ by solving the equation $f_d(z)=e^{-d\sqrt{z}}=
c^0_d\; \int\;dl\;e^{-zl}\; \hat{f}_d(l)$ for $\hat{f}_d(l)$ using the inverse 
formula for the Laplace transform $\hat{f}_d(l)=\int_C\; \frac{dz}{2\pi i}
e^{lz} f_d(z)$ where $C$ is any path in $\mathbb{C}_+:=\{z\in\mathbb{C};
\Re(z)>0\}$ between end points $z_\pm,\; \Im(z_\pm)=\pm \infty$. 
Given Such 
a path is given by $u\in \mathbb{R} \mapsto\ \sqrt{z}(u)=\frac{d}{2l}+
|u|\exp(i{\rm sgn}(u)\frac{\pi}{4})$. The resulting integral over $u$ 
is then Gaussian which yields two contributions, one is the integrand 
of (\ref{a.1}), the other is proportional to 
$e^{-\frac{d^2}{4 l}}\delta(l)\equiv 0$. In what follows we verify (\ref{a.1})
dirictly using the methods of \cite{27a} because the Laplace method does 
not carry over to $s=1$ (the case $s=1$ is formally obtained from 
$s=0$ by analytic continuation).\\
\begin{proof}:\\
The claim is equivalent to 
\ba \label{a.2}
[c^s_d]^{-1} 
&=& 2\int_0^\infty\; \frac{du}{u^2}\;e^{(-1)^{s+1} i^s
[\frac{d}{2 u}-\sqrt{z}u]^2}
\nonumber\\
&=& 2\int_0^\infty\; dv\;e^{(-1)^{s+1} i^s
[\frac{d v}{2}-\frac{\sqrt{z}}{v}]^2}
\nonumber\\
&=& \frac{4}{d}\int_0^\infty\; dv\;e^{(-1)^{s+1} i^s
[v-\frac{r}{v}]^2}
\ea
where we succesively changed variables $l=u^2,\; u=v^{-1}, v\; d/2\to v$
and set $r:=\frac{d\sqrt{z}}{2}\in \mathbb{R}_+$. 
Now set $x(v):=v-\frac{r}{v}$ which is 
a bijection $\mathbb{R}_+\to \mathbb{R}$ since $x'(v)=1+\frac{r}{v^2}>0$ 
with inverse 
$v(x)=\frac{x}{2}+w(x),\; w(x)=\sqrt{r+[\frac{x}{2}]^2}$. It follows
\ba \label{a.3}
[c^s_d]^{-1}
&=& \frac{4}{d}\int_{-\infty}^\infty\; dx\;
\frac{\frac{x}{2}+w(x)}{2w(x)}\;
e^{(-1)^{s+1}\; i^s\;x^2}
\nonumber\\
&=& \frac{2}{d}\int_{-\infty}^\infty\; dx\;
e^{(-1)^{s+1}\; i^s\;x^2}
\nonumber\\
&=& \frac{2\sqrt{\pi}}{d}\; e^{i\;s\frac{\pi}{4}}
\ea
where we used that $\frac{x}{w(x)} e^{(-1)^{s+1}\; i^s\; x^2}$ is odd under
$x\to -x$.
\end{proof}

\section{Issues with the inverse Laplace transform}
\label{sc}

In the ASQG literature one introduces certain cut-off functions 
$R_k$ and assumes that they and their products with other functions 
are in the image of the Laplace transform. Assuming this, the 
required heat kernel time ``Q'' integrals over the pre-image that appear 
in the computation of the flow equations can be reduced to properties 
of the cut-off function itself, thus one never needs to know the pre-image
explicitly. In the first subsection we show that there are isses with the 
existence of the pre-image for selected proposed cut-offs and in the second 
we verify for a function that {\it is} in the image of the Laplace 
transform, that it qualifies as a Euclidian cut-off function.

\subsection{Inverse Laplace transform of selected proposed 
cut-offs}
\label{sc.1}

A standard reference on the (one sided) Laplace transform 
is \cite{24}.
Recall that given a function $\hat{f}:\; \mathbb{R}_+\to \mathbb{C}$ one 
defines its Laplace transform $f_k:\; \mathbb{R}_+\to \mathbb{C}$ by 
\be
\label{c.1}
f(z)=[L\cdot \hat{f}](z):=\int_0^\infty\; dt\; e^{-z\;t}\; \hat{f}(t)       
\ee
if the integral exists. A class of functions for which this is true are 
for example the piecewiese continuous functions of at most exponential 
growth, i.e. 
there exist constants $A,B>0$ such that $|\hat{f}(t)|\le A\; e^{Bt}$
so that $f$ exists for $z>B$. 

For this domain of functions for (\ref{c.1}), the image is a holomorphic 
function in the half plane $\Re(z)> B$ and by the theorem of Lerch, the  
Laplace tranmsform is an injection. Thus given two holomorphic 
functions in the image of this domain we can uniquely reconstruct
$\hat{f}_k$ by the Bromwich formula
\be \label{c.2}
\hat{f}(t)=[L^{-1}\cdot f](t)=
\int_{a+i\mathbb{R}}\; \frac{dz}{2\pi i}\; e^{z\;t}\;
f(z)
\ee
as a piecewise constant function of at most exponential
growth, where $a>B$ is some constant chosen such that the path 
in (\ref{c.2}) is to the right of all singularities of $f$. 

It is quite difficult to determine the exact image of the Laplace transform 
given a domain. For instance we may extend (\ref{c.1}) to a 
distribution which is the weak derivative of a piecewise constant function of 
compact 
suport in $(0,\infty)$, say $\hat{f}(t)=[\frac{d}{dt}]^n \chi_{[a,b]}(t)$ 
with $0<a<b<\infty$ and $\chi$ is the charcteristic function. 
Then $[L\cdot T](z):=z^n\; [L\cdot \chi]_{[a,b]}(z)$ exists, it is even 
a holomorphic function, but it is 
not clear whether the pre-image is unique if we consider the domain of 
tempered distributions of compact support in $(0,\infty)$.   

Therefore it is not clear whether we can turn the logic 
around, i.e. given a function  or distribution
$f(z)$ whether we can we find a function or distributioon
$\hat{f}$ and from what precise space such that $f=L\cdot \hat{f}$ holds 
and whether 
$\hat{f}$ is unique. This is however the situation we are confronted 
with in asymptotic safety: We are given a certain function $f_k(\Box)$ 
of the Laplacian and wish to write it as a Laplace transform. 

The concrete application in asymptotic safety are the functions
\be \label{c.3}
f_{k,N}(z)=\frac{k\partial_k R_k(x)}{[z+R_k(z)]^N},\; n\ge 1
\ee
where $R_k(x)$ is the chosen cut-off function. A popular choice,
applied in the Euclidian signature regime, is the 
{\it optimised cut-off} 
\be \label{c.4}  
R_k(z)=(k^2-z)\;\theta(k^2-z)\;\;\Rightarrow\;\;
f_{k,N}(z)=2\;k^{2(1-N)}\;\theta(k^2-z)
\ee
Accordingly we are interested in the question whether the function 
$f_k(x)=\theta(k^2-x)$ is in the image of the Laplace transform. We
will be satisfied to show that there exists at least one pre-image, 
even if it is not unique. 

A possible Ansatz is to apply the Bromwich integral but we expect 
trouble because $\theta(k^2-z)$ is obviously not a holomorphic function.
At best it is a holomorphic distribution. In fact, we can make use of 
the identity
\be \label{c.5}
\theta(z)=\frac{1}{2}({\sf sgn}(z)+1)=\frac{1}{2}+
\int_{\mathbb{R}}\; \frac{dp}{2\pi} \frac{e^{ipz}}{ip}     
\ee
which can be proved using the residue theorem. Formula (\ref{c.5}) 
applied to $f_k(z)=\theta(k^2-z)$ however cannot be 
analytically extended to a half space $\Re(z)>a$ as the integral over $p$
extends over the whole real axis. If we nevertheless formally apply 
the Bromwich integral we obtain an ill-defined expression. On the other 
hand (\ref{c.5}) demonstrates that $f_{k,n}$ for the optimised cutoff formally 
applied to the Lorentzian regime can be written as a {\it Fourier transform}
and thus has a Lorentzian heat kernel expression. 

Returning to the Euclidian regime we will now rigorously 
prove:
\begin{Proposition} \label{prop.c.1} ~\\
The function $f_k(z)=\theta(k^2-z),\; k\not=0$ on $\mathbb{R}_+$ is not in the  
image of the Laplace transform of a positive measure on the positive 
real line.
\end{Proposition}
\begin{proof}:\\ 
Suppose that a given $f_k(z)$ had a
pre-image $\hat{f}_k(t)$ under the Laplace transform. Then we are 
interested in the heat kernel time integrals 
\be \label{c.6}
I_{k,n}:=\int_0^\infty\; dt\; \hat{f}_k(t)\; t^n = \left\{
\begin{array}{cc}
(-1)^n\; ([\frac{d}{dz}]^n f_k)(z)_{z=0} & n\ge 0\\
\frac{1}{(|n|-1)!}\; \int_0^\infty\; dz\; z^{|n|-1} \; f_k(z) & n\le -1
\end{array}
\right.
\ee
This yields for $f_k(z)=\theta(k^2-z)$ the explicit expression 
\be \label{c.7}
I_{n,k}=\theta(n) \delta_{n,0}+\theta(-n) \frac{1}{|n|!} k^{2|n|}
\ee
We interpret $\hat{f}_k(t)\;dt=d\mu_k(t)$ as the symbolic notation of a 
possibly signed measure on $\mathbb{R}^+$. Then for $n\ge 0$ 
\be \label{c.8}
\mu_{k,n}:=\int \; d\mu_k(t)\; t^n =\delta_{n,0}
\ee
is a Stieltjes moment problem \cite{24} (see vol. 2, ch. X, problem 25).
It has a solution in the class of positive measures if and only if 
\be \label{c.9}
\sum_{m,n=0}^\infty\; \bar{z}_m \; z_n \;\mu_{k,m+n}\ge 0,\;
\sum_{m,n=0}^\infty\; \bar{z}_m \; z_n \;\mu_{k,m+n+1}\ge 0
\ee
and it is unique if there exist $A,B>0$ such that $\mu_{k,n}\le 
A\; B^{n}\; (2n)!$. It is not difficult to show that both critera
are satisfied. The unique positive measure solving (\ref{c.7})
is the Dirac measure supported at zero, formally $d\mu_k(t)=\delta(t)\;dt$.
But then for any $n>1$
\be \label{c.10}
\frac{k^{2n}}{n!}=\int\; d\mu_k(t)\; t^{-n}=+\infty
\ee
which is a contradiction. 
\end{proof}
A more formal way to see that (\ref{c.7}) is contradictory is 
to Taylor expand the exponential in the relation
$\theta(k^2-x)=\int\;dt \; \hat{f}_k(t)\; e^{-zt}$ and to use (\ref{c.7}) for 
$n\ge 0$ which gives $\theta(k^2-z)\equiv 1$ for all $z$, i.e. there is no 
jump at $z=k^2$. 

While we cannot exclude that there exists a distribution of a different 
kind than the ones discussed at the beginning of this section with 
non-holomorphic images under the Laplace transform, the 
claimed properties \ref{c.7} based on the existence 
of the pre-image $\hat{f}_k(t)$ of the  optimised cut-off are very doubtful.

In the Lorentzian 
regime we infer from (\ref{c.5}) that 
\be \label{c.14}
\theta(k^2-z)=\int_{-\infty}^\infty\;\frac{dt}{2\pi} \; e^{itz}
[\frac{e^{-itk^2}}{it}+2\pi \delta(t)] 
=:
\int_{-\infty}^\infty\;dt\; \hat{f}_k(t)
\ee
which means that indeed a heat kernel expansion exists but the 
heat kernel time integrals 
$\int_{-\infty}^\infty\;dt\; \hat{f}_k(t)\; t^n$ badly diverge 
for all $n\in \mathbb{Z}$. Hence the optimised cut-off is also not 
practical in the Lorentzian regime.\\
\\
There are also issues with suggested 
other ``holomorphic'' cut-offs such as 
$R_k(z)=k^2\;r(z/k^2),\; r(y)=y\;[e^y-1]^{-1}$ which 
implies $k\partial_k R_k/(z+R_k)=2\;r(y)=2\;R_k(z)/k^2$. 
The function $r$ has poles on the
imaginary axis at $y=2\pi i n,\;n\in \mathbb{Z}$, thus for any $a>0$
the path in the Bromwich integral is to the right of those poles.
Suppose $r(y)$ is the Laplace transform of $\hat{r}(t)$. Then pick a 
Bromwich contour in the right complex plane starting and ending at 
$-\mp i(2n+1)\pi$ respectively and take $n\to \infty$. We can close that 
contour in the left complex plane with a semi-circle of radius 
$r_n=(2n+1)\pi$ so that $e^y,\;e^{yt};t>0$ decay to zero exponentially as 
$n\to \infty$ on that semi-circle. Thus by the residue theorem 
the Bromwich integral returns the infinite series 
\be \label{c.14a}
\hat{r}(t)=\sum_{k\in \mathbb{Z}} (2\pi i k)\; e^{2\pi i k t}
=\frac{d}{dt} \delta_{S^1}(t)
\ee
which is a distribution, namely the derivative of the 
$\delta$ distribution on the circle,
and no function as it diverges poinwise in $t$. In particular it is 
not piecewise continuous and of exponential type thus the Bromwich 
integral is not granted to apply. Indeed we can can check
whether $r$ is the Laplace transform of (\ref{c.14a}) by integrating 
term-wise for $y>0$ to obtain
\be \label{c.14b}
\int_0^\infty \; dt\; e^{-yt} \; \hat{r}(t)=
-\sum_{k\in \mathbb{Z}}\;\frac{2\pi ik}{2\pi ik-y}=-8\pi^2\;
\sum_{k=1}^\infty\; \frac{k^2}{[2\pi k]^2+y^2}
\ee
which diverges for all $y>0$. Thus if $\hat{r}$ exists, then certainly 
not as a function. However, we may interpret $\delta_{S_1}(t)$
as the periodisation of the $\delta$ distribution on the real axis, i.e.
\be \label{c.14c}
\delta_{S^1}(t)=\sum_{n\in \mathbb{Z}}\; \delta(t,n)
\ee
Then for $t>0$ we may try the definition (dropping the $n\le 0$ contributions 
by hand) 
\be \label{c.14d}
\hat{r}(t):=\frac{d}{dt}\sum_{n=1}^\infty \delta(t,n)
\ee
which yields for $y>0$ upon formally integrating by parts 
\be \label{c.14e}
\int_0^\infty\; dt\; \hat{r}(t)\; e^{-yt}
=y\;\sum_{n=1}^\infty\; e^{-ny}        
=y\;\frac{e^{-y}}{1-e^{-y}}=r(y)
\ee
since the geomteric series converges for $y>0$. 

Accordingly, $r(y)$ can 
indeed be interpreted as the Laplace transform of the distribution 
(\ref{c.14d}). However, with this interpretation, the claimed identities 
(\ref{c.6}) a priori fail. We have for $N\in \mathbb{Z}$
\be \label{c.14f}
\int_0^\infty\; t^N\; \hat{r}(t)
=\left\{ \begin{array}{cc}
0 & N=0\\
-N\;\sum_{n=1}^\infty\; n^{N-1}=-N\; \zeta(1-N) & N\not=0
\end{array}
\right. 
\ee
where $\zeta$ is the Riemann $\zeta$ function. 
As it stands, the r.h.s. of (\ref{c.14f}) is zero for $N=0$,
diverges to $-\infty$ for $N>0$ and converges to a positive 
number for $N<0$. To make (\ref{c.6}) hold we 
interpret the ill defined r.h.s. of (\ref{c.14f}) in terms 
of the of the analytic continuation of 
the $\zeta$ function. We define $\zeta(1)=+\infty$ such that 
$0\cdot \zeta(1):=1$ so that (\ref{c.14f}) produces $r(0)$. Then for $N<0$  
(\ref{c.6}) and (\ref{c.14f}) agree as the integral in (\ref{c.6}) 
coincides with the integral definition
\be \label{c.14g}
\zeta(z)=\frac{\int_0^\infty\; dy\; y^{z-1} \; [e^y-1]^{-1}}
{\int_0^\infty\; dy\; y^{z-1} \; e^{-y}} 
\ee
of the $\zeta$ function for $\Re(z)>1$.
Finally the analytic continuation yields for $N< 0$
\be \label{c.14h}
\zeta(-N)=(-1)^N\;\frac{B_{N+1}}{N+1}\;\; (N\ge 0),\; 
\ee
where $B_N$ denotes the $N-$th Bernoulli number.

\subsection{Cut-off functions in the image of the Laplace transform}
\label{sc.2}

The lesson learnt from the previous subsection is that for the cut-off f
unctions 
proposed in ASQG for the Euclidian regime, a non-trivial task
is to show that they are in fact in the image of a function,
distribution or measure such that the important relations (\ref{c.7}) may 
be applied to compute and show the convergence of 
the heat kernel time integrals. This issue is open for the optimal
cut-off while for the ``holomorphic'' cut-off we found a distribution 
in the pre-image. The integrals (\ref{c.7}) are ill-defined for that 
distribution, only a ``reinterpretation'' ($\zeta$ function regularisation)
makes (\ref{c.7}) work. On the other hand, the fact that $\hat{r}(t)$ 
is a distribution rather than a smooth function carries the danger of 
producing further ill-defined expressions in multiple integrals over 
heat kernel times as products of distributions are ill-defined.  

A safer route is 
to define the cut-off function as the Laplace transform of a given 
function. But then it is not a priori clear whether the image 
of the Laplace transform satisfies the desired properties of a 
cut-off function. This touches on the branch of mathematics known as 
Paley-Wiener theory \cite{24}. For instance it is known that 
$L\cdot L_2(\mathbb{R}_+, dt)=H_2(\mathbb{C}_+)$, i.e. the 
image of the one sided Laplace transform of $L_2$ functions $\hat{r}$ on the 
positive real line are the entire functions $r$ on the upper half plane 
(with complex variable $z=x+iy$) such that the integrals of $|r|^2$ along
the real axis at constant $y$ are uniformly bounded in $y$ (Hardy space).
This is certainly not the case for the above $r(z)=-iz/(e^{-iz}-1)$ whose 
corresponding integrals over $x$ diverge for any $y>0$.  
We leave this interesting question for future research.     

The source of the problem is that acccording to (\ref{c.6}) 
the function $\hat{f}_k(t)$ must be able 
to integrate {\it all positive and negative powers of $t$}. This 
sugests that $\hat{f}_k(t)$ decays rapidly at both $t=0,=\infty$ 
and is the reason why we 
were lead to consider functions such as 
\be \label{c.11}
\hat{f}_k(t)=k^4\;\exp(-([k^2 t]^2+[k^2 t]^{-2})/2)
\ee
in the next section in application to the Lorentzian sector. 
The integrals of this function multiplied by $t^n,n\in \mathbb{Z}$ 
exist and are computed in \cite{27a}. But then we define 
$f_k(z)$ as the Laplace transform of $\hat{f}_k(t)$ with {\it given, existing}
$\hat{f}_k(t)$ and $f_k(z)$ is no longer a prescribed but a derived function.

We now estimate the 
the Laplace transform (\ref{c.11}) to demonstrate 
that it has the usually desired properties. We 
have for $z\ge 0$ 
\be \label{c.12}
f_k(z)=k^2\; f(z/k^2),\;
f(y)=\int_0^\infty\; dt\; e^{-t^2-t^{-2}} \; e^{-yt}
\ee
For fixed $k>0$ we are interested in $z\ll k^2$ and $z\gg k^2$ i.e.
$y\to 0+$ and $y\to +\infty$. We have \cite{27a}
\be \label{c.13}
\lim_{y\to 0+} f(y)=f(0)=\frac{e^{-2}\; \sqrt{\pi}}{2},\;\;
f(y)\le \int_0^\infty \; e^{-ty}=\frac{1}{y}
\ee
Thus at fixed $k$ the function 
$f_k$ indeed approaches a constant for $z\ll k^2$ and decays to zero 
at least as $1/z$ for $z\gg k^2$. Next, for fixed $z>0$ we are interested 
in $k\to 0+$ and $k\to +\infty$ i.e. $y\to +\infty$ and 
$y\to 0+$ respectively. Using (\ref{c.13}) we see that $f_k(z)$ 
decays as $k^4/z$ as $k\to 0$ and 
$f_k(z)\to k^2 f(0)$ as $k\to \infty$. Thus (\ref{c.11}) has all 
the properties that are usually required in the {\it Euclidian 
regime} and now it is secured that 
$f(y)$ is in the image of the Laplace transform. This also 
enables one to check that the product of two 
functions $f,g$ is in the image of the Laplace transform when $f,g$ are.
This will be the case if e.g. the convolution integral 
$\int_0^t\; ds\; \hat{f}(s) \; \hat{g}(t-s)$ exists for all $t\ge 0$
and grows at most polynomially in $t$.

Finally we compute the ``Q'' integrals for (\ref{c.11})
\be \label{c.13a}
Q_n=\int_0^\infty\; dt\; t^n\; \hat{f}(t)
=\int_0^\infty\; dt\; t^n\; e^{-t^2-t^{-2}}
\ee
for $n\in \mathbb{Z}$ which cannot be treated by (\ref{c.6}) because we do 
not know $f(y)$ in closed form, thus we must compute directly. We note 
that $Q_n=Q_{-(n+2)}$ which enables us to concentrate on $n\ge -1$.
Introducing the bijection $\mathbb{R}_+\to \mathbb{R}$ defined by 
$u=t-t^{-1}, t=u/2+w,\;w:=\sqrt{1+[u/2]^2}$ we have $dt=t/(2w)\;du$ and thus 
\be \label{c.13b}
Q_n=\frac{e^{-2}}{2}\;\int_{-\infty}^\infty\; \frac{du}{w}\;\; e^{-u^2}\;
(w+u/2)^{n+1}
\ee
It is immediate to see that $Q_{2n}$ 
is computable in closed form \cite{27a} in terms of elementary Gaussian 
integrals. These integrals are also related to modified 
Bessel functions of the second kind. For  $Q_{2n+1},n\ge -1$ we can make 
use of the elementary estimate $w^{-1}<1<w$ to see that $Q_{2n+1}$
has analytically computable upper and lower bounds \cite{27a}.
These integrals will be used in our companion paper \cite{companion}.
\\
\\
Accordingly (\ref{c.11}) qualifies as a Euclidian cut-off function.

\section{Lorentzian cut-off functions}
\label{sd}

In the Lorentzian signature case the Wetterich equation 
\be \label{d.1}
k\;\partial_k\; \Gamma_k
=\frac{1}{2i}\; {\rm Tr}([\Gamma_k^{(2)}+R_k]^{-1}\; [k\partial_k R_k])
\ee
can be usefully expanded as follows (we use the background formalism): 
We split 
\be \label{d.2}
\Gamma_k=F_k+U_k
\ee
where $F^{(2)}_k$ has the property to be independent of $k$ and to depend 
just on the background Laplacian, also having its dimension. 
If both numerator and denominator in (\ref{d.1})
contain a commpon prefactor $f_k$ that depends on the dimensionless 
couplings $g_k$ then $\partial_k R_k$ is to be replaced 
by $k\partial_k R_k+\eta_k R_k, \eta_k=k\partial_k \ln(f_k)$, 
the additional ``anomalous dimension'' 
term increases the non-linerarity of the flow as one 
must eventually solve the system for $k\partial_k g_k$. We refrain 
from displaying this term as it can be treated by the same 
methods as described below.  

We then expand (\ref{d.1}) into a 
geometric series 
\be \label{d.3} 
k\;\partial_k\; \Gamma_k
=(2i)^{-1}\;\sum_{n=0}^\infty\; (-1)^n\;
{\sf Tr}([k\partial_k R_k]\;[F_k^{(2)}+R_k]^{-1}\;(U_k^{(2)}\;
[F_k^{(2)}+R_k]^{-1})^n)
\ee
which can be rewritten as 
\be \label{d.4} 
k\;\partial_k\; [\Gamma_k-(2i)^{-1}\; {\sf Tr}(\ln([F_k^{(2)}+R_k]
=(2i)^{-1}\;\sum_{n=1}^\infty\; (-1)^n\;
{\sf Tr}([k\partial_k R_k]\;[F_k^{(2)}+R_k]^{-1}\;(U_k^{(2)}\;
[F_k^{(2)}+R_k]^{-1})^n)
\ee
The logarithmic term on the l.h.s converges to zero as $k\to 0$ if 
$R_k\to 0$. One can therefore ignore it in that limit. We also ignore 
it for the moment and come back to it at the end of this section. 
We just remark that since it can be written in 
terms of the functional 
determinant $\det(F_k^{(2)}+R_k)$ one may envisage also independent 
regularisation schemes such as zeta function regularisation. Note 
also that if $F_k=F$ is also independent of $k$ and not only 
$F_k^{(2)}$ then $k\partial_k \Gamma_k=k\partial_k U_k$ and (\ref{d.4})
becomes a flow equation entirely in terms of the effective potential 
$U_k$.

The right hand side of (\ref{d.4}) can be written
(we formally use cyclicity of the trace)
\be \label{d.5}
(2i)^{-1}\;\sum_{n=0}^\infty\; (-1)^n\;
{\sf Tr}([k\partial_k (F_k^{(2)}+R_k)^{-1}]\;U^{(2)}_k\;
[(F^{(2)}_k+R_k)^{-1})\; U_k^{(2)}]^n)
\ee
which displays the fundamental role of the operator 
\be \label{d.6}
(F^{(2)}_k+R_k)^{-1}=:[F_k^{(2)}]^{-1}+G_k  
\ee
where we assume that $G_k$ is an operator just depending on $\Box$
and $G_k\to 0$ as $k\to 0$ which is consistent with
$R_k\to 0$ as $k\to 0$. As an explicit example 
consider 
\be \label{d.7}
G_k(z)=k^{-2}\; e^{-k_0^2/k^2}\;\int_{-\infty}^\infty\; 
dt\; e^{-t^2-t^{-2}}
e^{i\;t\;z/k^2}
\ee
which is a smooth, real valued function of rapid decrease at 
$|z|=\infty$. The fixed constant $k_0$ is there just for dimension reasons
and does not run with $k$. Instead of $k_0^2/k^2$ one could also use 
$\lambda_k=\Lambda_k/k^2$ the dimensionless cosmological constant where 
$k_0^2=\Lambda_{k=0}$. For both choices $G_k\to 0$ rapidly at $k=0$ while 
$G_k$ has the correct dimension of $1/R_k$.

Since the integrand is smooth at $t=0$ where it and 
all of its derivatives vanish we can extend the integrand smoothly to zero 
for all $t<0$ and therefore may also consider 
\be \label{d.8}
G_k(z)=k^{-2} \; e^{-k_0^2/k^2}\;\;G(z/k^2),\; G(y)=
\int_0^\infty\; dt\; e^{-t^2-t^{-2}}
e^{i\;t\;y}
\ee
which however is no longer real valued but still vanishes at $k=0$. We may 
also consider the case $k_0=k$ which has the advantage that this 
cut-off function leads to autnomous flow equations.

Note that 
\be \label{d.9}
k\partial_k (F^{(2)}_k+R_k)^{-1}=
k\partial_k G_k 
\ee
just depends on $G_k$.
The reason for choosing (\ref{d.7}) or (\ref{d.8}) is that in the 
heat kernel expansion we are therefore interested in the heat kernel time 
integrals 
\be \label{d.9a}
\int\; dt\; \hat{G}_k(t) e^{i\pi/4\;(2-d){\sf sgn}(t)}\; |t|^{-d/2}\; t^l,\;
l\in \mathbb{N}_0
\ee
where $\hat{G}_k$ is the integrand of the integral defining $G_k$, i.e. 
its Fourier transform and $d$ the spacetime dimension. These integrals
all converge due to the properties of $\hat{G}_k$ at $t=0$. 

We consider the first two orders on the right hand side of (\ref{d.5})
which already display the essential features of the expansion
\be \label{d.10}
i^{-1}\;{\sf Tr}([k\partial_k G_k]\;U_k^{(2)}\;
[1-((F_k^{(2)})^{-1}+G_k)\; U_k^{(2)}])  
\ee
which requires us to know also the heat kernel expansion of $(F_k^{(2)})^{-1}$.
If $F_k^{(2)}(z)$ can be written as $Az+B$ with $A,B\not=0$ real constants 
independent of $k$ we have the {\it Schwinger proper time expression}
\be \label{d.11}
(Az+B)^{-1}=-\frac{i}{A}\;[\int_0^\infty\; dt\;e^{i\;t\;z-t\; C}]_{
C\to -i B/A}
\ee
where the integral is performed at $C>0$ and then analytically extended 
to $C=-iB/A$. By rescaling one can assume $A=1$ and it is natural to 
identify $B=k_0^2=\Lambda_{k=0}$.   

If we now use the version (\ref{d.8}) we will be confronted 
in (\ref{d.10}) with heat kernel expressions of the form
form 
\ba \label{d.12}
&& \int_0^\infty\;dt_1 \int_0^\infty dt_2   
\hat{G}'_k(t_1) \; \tilde{G}_k(t_2)
{\sf Tr}(U_k^{(2)}\; e^{it_1\Box}\; U_k^{(2)}\; e^{it_2 \Box})
\nonumber\\
&=& \int_0^\infty\;dt_1 \int_0^\infty dt_2   
\hat{G}'_k(t_1) \; \tilde{G}_k(t_2)
{\sf Tr}(U_k^{(2)}\; 
[e^{it_1\Box}\; U_k^{(2)}\; e^{-it_1\Box}]\; 
e^{i(t_1+t_2) \Box})
\ea
where $\hat{G}'_k(t)=\hat{G}_k(t)$ or $\hat{G}'_k(t)=\dot{\hat{G}}_k(t)$
and
$\tilde{G}_k(t)=\hat{G}_k(t)$ or $\tilde{G}_k(t)=e^{-Ct}$. 
As indicated in the main body of the paper we now expand
\be \label{d.13}
e^{it_1\Box}\; U_k^{(2)}\; e^{-it_1\Box}
=\lim_N\sum_{n=0}^N\;\frac{(it_1)^n}{n!}[\Box,U^{(2)}_k]_{(n)}
\ee
and truncate at some finite $N$. The heat kernel expansion now 
results in time integrals of the type (note that the sign of 
$t_1,t_2,t_1+t_2$ is positive)
\be \label{d.14}
=\int_0^\infty\;dt_1 \int_0^\infty dt_2   
\hat{G}_k(t_1) \; \tilde{G}_k(t_2)\; e^{i\pi/4\;(2-d)}\;
t_1^r\; (t_1+t_2)^s,\; r\ge 0, s\ge -d/2
\ee
which all converge due to the estimate $(t_1+t_2)^{-s}\le t_1^{-s},\; 
s\ge 0$ and the fact that all times are positive. Since there is always 
at least one factor of $\hat{G}_k(t_1)$ also at higher orders, all those 
integrals converge. Note that this would not be the case had we used 
(\ref{d.7}). Note however that with $R_k$ being complex valued we expect 
the coupling constants $g_k$ also to be complex valued for $k>0$. 
Thus the number of real coupling constants is in principle doubled.
We confine to the original number of real couplings by 
considering only {\it admissible} trajectories which are those 
such that all dimensionful couplings are real valued as $k\to 0$. 
This halves the number of admissible initial conditions and 
presents some form of fine tuning.  
  
Coming back to the logarithmic term in (\ref{d.4}) we note that with 
the above choice of $R_k, F_k$ we have 
\ba \label{d.16}
&&k\partial_k\ln(F_k^{(2)}+R_k)  
=-k\partial_k\ln([F_k^{(2)}]^{-1}+G_k)  
\nonumber\\
&=&-2\frac{F_k^{(2)}\;G_k}{1+F_k^{(2)}\;G_k}
=\sum_{n=1}^\infty\; (-1)^n\;[F_k^{(2)}\;G_k]^n
\nonumber\\
&& F_k^{(2)}(z)\;G_k(z)
=k^2 \;k_0^2\;\int\; dt\; e^{itz}\;[-i \;A\; \frac{d}{dt}+B]
e^{-(t\;k_0)^2-(t\; k_0)^{-2}}
\ea
and if the integration domain of $t$ is positive (note that due to rapid decay
there are no boundary terms picked up by the integration by parts) 
then the heat kernel traces and time integrals of all powers of 
(\ref{d.16}) converge. Thus in this case we can even tame the logarithmic 
term if we truncate the geometric series. \\
\\  
The above choice of $R_k$ in  terms of $G_k$ has the advantage to 
enable a simple series expansion of the Wetterich equation just in terms 
of $U_k$. However, the resulting equation is not autonomous as we 
were forced to introduce $k_0$. Had we replaced $B=k_0^2$ by $\Lambda_k$ then 
the heat kernel time integrals would not always contain a factor 
of $\hat{G}_k$ and therefore in general diverge. We therefore now consider a 
more complicated expansion which however is autonomous.

This is obtained by refraining from introducing $G_k$ and instead 
we use
\be \label{d.17}
R_k(z)=k^2 r(z/k^2),\;\;r(y)=\int_0^\infty\; dt \; \hat{r}(t) 
e^{ity},\;\hat{r}(t)=e^{-t^2-t^{-2}}\;
\; \Rightarrow\;\; k\partial_k R_k=2k^2[r-y r'(y)]=2k^2[r-ir_1],
\hat{r}_1(t)=\dot{\hat{r}}(t)
\ee
Consider again (\ref{d.2}) where now $F_k^{(2)}=z+\Lambda_k$. Here 
$\Lambda_k$ is the dimensionful cosmological constant. Then we expand
the r.h.s. of (\ref{d.1}) with $u_k=U_k^{(2)}/k^2;\;\lambda_k=\Lambda_k/k^2,
y=\Box/k^2$ 
\ba \label{d.18}
&& 2\;{\sf Tr}[(r-y\; r')\;[(y+\lambda_k)+(r+u_k)]^{-1}]
=2\;\sum_{n=0}\; (-1)^n\; {\sf Tr}(\frac{r-i r_1}{y+\lambda_k}\;
([r+u_k]\; [y+\lambda_k]^{-1})^n)
\nonumber\\
&& [y+\lambda_k]^{-1}=-i [\int_0^\infty\; dt\; 
e^{t(iy-\epsilon)}]_{\epsilon\to -i\lambda_k}
\ea
where all traces and $t$ integrals 
are computed at finite $\epsilon>0$ and at the end are analytically 
continued as displayed. The essential mechanism that makes all $t$ integrals  
converge namely that at least one factor of $\hat{r},\dot{\hat{r}}$ appears 
is preserved, the system is autonomous but the expansion is more complicated 
than before.

\end{appendix}

\end{document}